\documentclass[12pt]{article}
\usepackage{authblk}
\usepackage{amssymb}
\usepackage{amsthm}
\usepackage{amsmath}
\usepackage{cleveref}
\usepackage{enumerate}
\usepackage{mathtools}
\usepackage{enumitem}
\usepackage{cmll}

\allowdisplaybreaks[4]
\usepackage{geometry}
\usepackage{bm}
\usepackage{natbib}
\usepackage{xcolor}
\usepackage{colonequals}
\usepackage{float}
\usepackage{graphicx}
\usepackage{algorithm}
\usepackage{color}
\usepackage{algpseudocode}
\usepackage{epsfig}
\usepackage{dcolumn}
\usepackage{multirow}
\usepackage{booktabs}
\usepackage{geometry}
\usepackage{subfigure}
\usepackage{mathrsfs}  
\usepackage{titlesec}
\usepackage{url}

\usepackage[nodisplayskipstretch]{setspace}

\titleformat{\section}
{\normalfont\fontsize{12.5}{12.5}\bfseries}{\thesection}{1em}{}
\titleformat{\subsection}
{\normalfont\fontsize{12}{12}\bfseries}{\thesubsection}{1em}{}

\titlespacing{\section}{0pt}{10pt}{3pt}
\titlespacing{\subsection}{0pt}{10pt}{1pt}
\titlespacing{\subsubsection}{0pt}{1pt}{1pt}

\newcommand{\nodisplayskips}{
	\setlength{\abovedisplayskip}{5pt}
	\setlength{\belowdisplayskip}{5pt}
	\setlength{\abovedisplayshortskip}{5pt}
	\setlength{\belowdisplayshortskip}{5pt}
}

\makeatletter
\newcommand{\red}{\color{red}}

\newcommand{\re}{\color{cyan}}

\usepackage{etoolbox} 

\makeatletter
\patchcmd{\@makecaption}
{\parbox}
{\advance\@tempdima-\fontdimen2} 
{}{}
\makeatother 

\newtheorem{theorem}{\bf{Theorem}}
\newtheorem{remark}{\bf{Remark}}
\newtheorem{lemma}{\bf{Lemma}}

\newtheorem{proposition}{\bf{Proposition}}
\newtheorem{condition}{\bf{Condition}}

\def\T{\rm T}

\newcommand{\diag}{{\rm diag}}

\newcommand{\var}{{\rm var}}
\newcommand{\tr}{{\rm tr}}

\newcommand{\bX}{\boldsymbol{X}}
\newcommand{\bY}{\boldsymbol{Y}}
\newcommand{\bx}{\boldsymbol{x}}
\newcommand{\bZ}{\boldsymbol{Z}}
\newcommand{\be}{\boldsymbol{e}}
\newcommand{\bg}{\boldsymbol{g}}
\newcommand{\ba}{\boldsymbol{a}}
\newcommand{\bxi}{\boldsymbol{\xi}}
\newcommand{\bzero}{\boldsymbol{0}}

\newcommand{\bbeta}{\boldsymbol{\beta}}
\newcommand{\balpha}{\boldsymbol{\alpha}}
\newcommand{\bbias}{\boldsymbol{b}}
\newcommand{\bdelta}{\boldsymbol{\delta}}
\newcommand{\bvareps}{\boldsymbol{\varepsilon}}
\newcommand{\etab}{\boldsymbol{\eta}}

\newcommand{\bbX}{\mathbf{X}}

\newcommand{\bbU}{\mathbf{U}}

\newcommand{\bbPi}{\mathbf{\Pi}}

\newcommand{\bbV}{\mathbf{V}}

\newcommand{\bbQ}{\mathbf{Q}}

\newcommand{\bbI}{\mathbf{I}}
\newcommand{\bbA}{\mathbf{A}}
\newcommand{\bbB}{\mathbf{B}}
\newcommand{\bbM}{\mathbf{M}}
\newcommand{\bbH}{\mathbf{H}}
\newcommand{\bbSig}{\mathbf{\Sigma}}

\usepackage{scalerel}

\newcommand{\sA}{\mathscr{A}}

\newcommand{\cE}{\mathcal{E}}

\newcommand{\cT}{\scaleto{\mathcal{T}\mathstrut}{5.5pt}}
\newcommand{\cS}{\scaleto{\mathcal{S}\mathstrut}{5.5pt}}
\newcommand{\cC}{\scaleto{\mathcal{C}\mathstrut}{5.5pt}}
\newcommand{\SURE}{\scaleto{\rm SURE\mathstrut}{5.5pt}}
\newcommand{\IVW}{\scaleto{\rm IVW\mathstrut}{5.5pt}}

\newcommand{\hB}{\widehat{\bbB}}
\newcommand{\hH}{\widehat{\bbH}}
\newcommand{\hQ}{\widehat{\bbQ}}
\newcommand{\hSig}{\widehat{\bbSig}}
\newcommand{\hbeta}{\widehat{\bbeta}}
\newcommand{\tbeta}{\widetilde{\bbeta}}
\newcommand{\halpha}{\widehat{\balpha}}

\newcommand{\hdelta}{\widehat{\bdelta}}
\newcommand{\tdelta}{\widetilde{\bdelta}}

\usepackage{longtable}

\def\T{{\mathrm{\scriptscriptstyle T}}}

\titlespacing{\section}{0pt}{3pt}{3pt}
\titlespacing{\subsection}{0pt}{1pt}{1pt}
\titlespacing{\subsubsection}{0pt}{1pt}{1pt}

\graphicspath{{./figures/}}

\usepackage{etoolbox} 

\makeatletter
\patchcmd{\@makecaption}
{\parbox}
{\advance\@tempdima-\fontdimen2} 
{}{}
\makeatother

\newcommand\blfootnote[1]{%
	\begingroup
	\renewcommand\thefootnote{}\footnote{#1}%
	\addtocounter{footnote}{-1}%
	\endgroup
}

\title{Divide-and-shrink: An efficient and heterogeneity-agnostic approach for transfer estimation using summary statistics }
\author[1]{Ruoyu Wang}
\author[1,2,$*$]{Xihong Lin}
\affil[1]{Department of Biostatistics, Harvard School of Public Health, Boston, USA}
\affil[2]{Department of Statistics, Harvard University, Boston, USA}

\begin{document}
	\date{}
\maketitle
\blfootnote{$*$ Corresponding author. Email: xlin@hsph.harvard.edu}
	\begin{abstract}
         Knowledge transfer across data sources holds great promise for improving the estimation of target population parameters by leveraging the growing availability of data from different sources. However, the effectiveness of knowledge transfer is often challenged by the complex and pervasive heterogeneity between data sources and the lack of access to individual-level data. This paper proposes the divide-and-shrink (dShrink) method, a transfer estimation method that estimates target population parameters in a closed form using summary statistics from a target population and some external source populations while accounting for population heterogeneity. The dShrink estimator is guaranteed to outperform the estimator based solely on the target population in terms of expected quadratic error under arbitrary population heterogeneity. The gain can be substantial when the target and source populations are similar, or the underlying true parameter values are near zero. Notably, dShrink is model-free, requires no user-specified tuning parameters, robust to various types of heterogeneity between data sources, and applies to a broad range of parameter estimation problems. dShrink remains effective even when the covariance matrix is not accessible for the external summary statistics and offers flexibility in incorporating side information and summary statistics from multiple source populations. Simulations and real data analyses demonstrate the superior performance of the dShrink estimator and its potential as a robust tool for transfer estimation.
		 
	\end{abstract}
	
	\noindent%
	\vfill
	 
    \setcounter{page}{1}
    \pagestyle{plain}
    \setstretch{1.9}
	\section{Introduction}    
     Data from diverse sources are increasingly available in contemporary data science tasks.  Knowledge transfer across data sources holds significant promise in reducing the parameter estimation error in the target population. In real-world applications, access to individual-level data is frequently limited due to privacy considerations or data transmission limitations. Consequently, researchers may only have access to specific summary statistics for certain data sources. While these summary statistics can provide valuable information, the heterogeneity between data sources poses a significant challenge to transferring knowledge for parameter estimation. A good amount of recent research has been dedicated to  
    transfer learning for prediction and estimation using individual-level data with heterogeneous data sources, developed under both frequentist \citep{tang2016fused, gu2025robust, cai2024semi} and Bayesian frameworks \citep{suder2025bayesian}.
    Despite these advances, methods capable of effectively using external summary statistics in the presence of general parameter and population heterogeneity are still relatively rare, and the existing methods usually entail strong assumptions on the heterogeneity between source and target populations \citep{sheng2022synthesizing, zhai2022data, huang2023simultaneous,chen2024integrating}. This situation underscores the necessity for developing heterogeneity-agnostic transfer estimation approaches using summary statistics that require minimal assumptions regarding heterogeneity.
	
	Suppose researchers conduct a study using their own data set, referred to as the internal data source, with the goal of estimating a specific parameter of interest in a target population. They can access individual-level data or summary statistics from this internal data source. Additionally, suppose that some summary statistics from an external data source are also available, offering potential utility in estimating the target parameter. However, heterogeneity between these data sources is a common and practical challenge. Such heterogeneity can be categorized into two main types. One is \emph{parameter heterogeneity}, which occurs when the summary statistics from the external data source estimate a parameter that differs from the target parameter. The other is \emph{population heterogeneity} which means the source population that generates the external summary statistics is different from the target population. 
	
	\emph{Parameter heterogeneity} usually occurs when the external data sources collect relatively crude information, while more detailed information is available from the internal data source. This situation is common, for example, when external data sources include fewer covariates than those collected by the internal data source \citep{qin2015using, chatterjee2016constrained, huang2016efficient}. This discrepancy can lead to challenges in incorporating the external summary statistics. Several transfer estimation methods have been developed to address parameter heterogeneity and effectively integrate summary statistics, including the confidence distribution method \citep{liu2015multivariate}, generalized method of moments-based methods \citep{huang2020unified, sheng2020censored, zhai2022data}, empirical likelihood-based methods \citep{qin2015using, chatterjee2016constrained, zhang2020generalized}, calibration-based methods \citep{yang2020combining}, among others \citep{taylor2023data}. While these approaches offer valuable strategies for dealing with parameter heterogeneity, they do not take into account population heterogeneity. Thus, these methods can be severely biased when the target and source populations are different and may have larger estimation errors than the target population-based estimator without integrating external summary statistics.
	
	\emph{Population heterogeneity} is ubiquitous when integrating or transferring information from different populations \citep{higgins2009re}.
	If not properly accounted for,  population heterogeneity can introduce a large bias in the transferred estimator in the target population \citep{hu2024efusion}, and result in a larger estimation error than the target population-based estimator. \cite{sheng2022synthesizing} attempts to mitigate this issue by applying a density ratio model for adjusting the covariate shift between source and target populations in the context of logistic regression. Yet, this method relies critically on both the assumption of covariate shift and the assumption of the correct specification of the density ratio model, which can be fragile in practice. \cite{hickey2024transfer} addresses  population heterogeneity by modeling the ratio between certain structural components of the two populations as a Cauchy random variable. However, the effectiveness of this method is well understood only in the specific case where the outcome and covariates follow a linear model and the observed outcomes in the target population are normally distributed, repeated samples from a fixed covariate value. In the presence of population heterogeneity,  directly using the summary statistics in the source population can yield biased estimators in the target population in general.
    
    To address this, \cite{chen2021combining}, \cite{zhai2022data}, \cite{huang2023simultaneous}, and \cite{han2025federated} employ penalization techniques to identify and incorporate the transportable external summary statistics. \cite{yang2023elastic} and \cite{gao2023pretest} employ a pretest procedure that incorporates the external data only when the test does not reject the null hypothesis that the external data are transportable. The success of these methods hinges on accurately identifying transportable summary statistics. They demonstrate promising results when transportable summary statistics are accurately selected with probability approaching one. However, attaining such selection consistency is possible only when the untransportable summary statistics have large biases so that they can be distinguished from the transportable ones. In practice, while the target and source populations can be similar, they are generally not identical. Additionally, in settings where transfer estimation methods are applied, the sample size from the target population is typically limited. In such cases, it may not be possible to accurately select the transportable summary statistics.
    Then, the above transferred estimators can be biased,
    have larger estimation errors than the target population-based estimator, and are not safe in general,
    where a	transfer estimation method is called \emph{safe} if the inclusion of external summary statistics does not increase the estimation error in terms of expected quadratic error compared to the target population-based estimator. 
    
    In scenarios where populations are similar but not identical, shrinkage-based estimators—such as the method of \cite{hector2024turning}, the James–Stein (JS) estimator for data combination \citep{green1991james, han2024improving}, and the data enriched linear regression \citep{chen2015data} \cite[and its extension to logistic regression,][]{zheng2025data}--exhibit desirable safety properties. The ISEDI method of \cite{hector2024turning} focuses on generalized linear models and requires access to individual-level data. Its safety is guaranteed when a data-independent ``oracle” tuning parameter is employed.
    The JS estimator and data enriched regressions are largely restricted to mean estimation or linear/logistic regression settings and require the summary statistics to be the estimated mean/coefficients in the source population. Please refer to Section \ref{subsec: JS and data enrich} for details about the properties and limitations of the JS estimator and data enriched linear regression.
	
	In this paper, we propose the divide-and-shrink (dShrink) method for improving parameter estimation and inference in a target population by robustly integrating individual-level data or summary statistics from the target population and the summary statistics from external populations that may be heterogeneous from the target population. The proposed method applies to general parameter estimation problems. It distinguishes itself by being safe in general parameter estimation problems with both parameter and population heterogeneity. 
    
    To motivate the proposed method, we begin by examining the idea of calibration \citep{chen2000unified} and pretest \citep{yang2023elastic}, which serve as insightful approaches 
    for
    data integration and transfer estimation. We show that neither the calibration method nor the pretest method can achieve safe transfer estimation without assumptions on population heterogeneity. Upon deeper analysis, we discover that both the calibration estimator and the pretest estimator fall within a general class of the divide-and-combine estimators. Specifically, a divide-and-combine estimator divide the target population-based estimator into two independent components utilizing the external summary statistics. The integrated estimator is then formed by adopting a linear combination of these two components. Nevertheless, both the calibration estimator and the pretest estimator apply unnecessary restrictions on the combination coefficients. Moreover, they overlook the implications of the estimation errors in the process of defining these coefficients. 
    
    Motivated by these,  we consider a general approach to combining the two estimators by estimating the combination coefficients 
    that minimizes the error of the combined estimator. The estimation error of the resulting combined estimator can be approximated using Stein's Unbiased Risk Estimate (SURE). Then, the estimated optimal combination coefficients can be obtained by minimizing SURE, which gives rise to the SURE estimator for the target parameter. However, we find the SURE estimator is not safe even when the target population is identical to the source population. This is because the combination coefficients that minimize SURE can be extremely unstable due to the ill-posedness of SURE in the transfer estimation context. To mitigate this problem, instead of directly minimizing an estimate of the error, we propose to minimize a separable surrogate objective function that can approximate the estimation error. The separable surrogate objective motivates a stabilized combination coefficient estimator, which shrinks each component in the combined estimators toward zero and leads to the proposed dShrink estimator. Furthermore, a re-bootstrap procedure is proposed to construct confidence intervals (CIs) of general functions of the parameter based on the dShrink estimator. 
	
	The proposed dShrink method is designed to accommodate both \emph{parameter heterogeneity} and \emph{population heterogeneity}, offering a safe and model-free transfer estimation approach to improving estimation and inference in the target population.  Moreover, the dShrink estimator enjoys a closed form and can be efficiently calculated using summary statistics of both the target and source populations.
	We further extend the dShrink method to accommodate the case where the covariance matrix of external summary statistics is not accessible and/or summary statistics from multiple source populations are available. We discuss the incorporation of side information to the dShrink method in Supplementary Material Section \ref{app: mds}. These extensions can be achieved without compromising the advantageous features of the dShrink method.
    
	We conduct extensive simulations to study the empirical performance of the dShrink method and compare it with the existing methods. In addition, we apply the dShrink method to the analysis of two data sets: One studies the association between infant death and various covariates, and one evaluates the treatment effect of rectal indomethacin on preventing endoscopic retrograde cholangiopancreatography (ERCP). These analysis results empirically show the attractive performance of dShrink across various settings. Code and data to produce numerical results in this paper can be found at \url{https://github.com/ryw-stats/dShrink}.
	
	The rest of this paper is organized as follows. Section \ref{sec: cal and pretest} delves into the idea of calibration and pretest in our context. Section \ref{sec: dShrink} proposes the dShrink method. Section \ref{sec: extension} extends the dShrink method to the case where the covariance matrix of the external summary statistics is unavailable and/or summary statistics from multiple source populations are available. We report extensive numerical results in simulation studies in Section \ref{sec: simulation} and real data analyses in Section \ref{sec: data analysis}, followed by discussions. 
	
	\section{Set-up and classical ideas}\label{sec: cal and pretest}
	\subsection{Problem setup}
	Let $P_{\cT}$ be a target population of primary interest, and $\{\bZ_{i}\}_{i=1}^{n_{\cT}}$ be a sample that consists of $n_{\cT}$ observations from $P_{\cT}$, where $n_{\cT}$ is the sample size from the target population. Let $\bbeta_{\cT} = \bbeta(P_{\cT})$ be the $p$-dimensional target parameter vector, where $\bbeta(\cdot)$ is a functional of distributions. Suppose one can assess the individual-level data $\{\bZ_{i}\}_{i=1}^{n_{\cT}}$ or  summary statistics calculated based on $\{\bZ_{i}\}_{i=1}^{n_{\cT}}$ in the target population. In addition, assume some external information is available. Specifically, suppose some summary statistics $\halpha_{\cS}$ are available that are derived from an independent sample of size $n_{\cS}$ from the source population $P_{\cS}$. In this paper, we allow the source population $P_{\cS}$ to deviate from the target population $P_{\cT}$. The summary statistics $\halpha_{\cS}$ can be mismatched with the target parameter $\bbeta_{\cT}$. For instance, suppose the parameter vector of interest $\bbeta_{\cT}$ is the joint regression coefficient vector under a generalized linear model in $P_{\cT}$. The external summary statistics $\halpha_{\cS}$ might be the marginal regression coefficients under a marginal model in $P_{\cS}$. Let $\balpha_{\cS} = \balpha(P_{\cS})$ be the $q$-dimensional external parameter defined as the probability limit of $\halpha_{\cS}$ with respect to $P_{\cS}$.  We say the external parameter is mismatched with the target parameter if the functional $\balpha(\cdot) \neq \bbeta(\cdot)$. We allow $q \neq p$ and assume $\balpha(P_{\cT})$ is well-defined and estimable using $\{\bZ_{i}\}_{i=1}^{n_{\cT}}$. Despite of heterogeneity, $\halpha_{\cS}$ can still provide valuable information for estimating $\bbeta_{\cT}$, especially when $P_{\cS}$ share some similarity with $P_{\cT}$. However, the similarity between populations is typically hard to quantify and validate. Our objective is to devise methods that can improve the estimation accuracy of $\bbeta_{\cT}$ utilizing $\halpha_{\cS}$ under minimal assumptions on the similarity between $P_{\cT}$ and $P_{\cS}$. 
    
    For any $p$-dimensional random vector $\widehat{\bxi}$ and deterministic vector $\bxi$, we introduce the expected quadratic error
	$\cE_{\bbH}(\widehat{\bxi}, \bxi) = E\left\{(\widehat{\bxi} - \bxi)^{\T}\bbH(\widehat{\bxi} - \bxi)\right\}$,
	 where $\bbH$ is a user-specified weighting matrix. Subsequently, the subscript $\bbH$ in $\cE_{\bbH}$ is omitted for simplicity.
	The estimation accuracy for any estimator $\widehat{\bbeta}$ is evaluated by
	\begin{equation}
	\small
	\nodisplayskips
	\cE(\widehat{\bbeta}, \bbeta_{\cT}) = E
    \left \{(\widehat{\bbeta} - \bbeta_{\cT})^{\T}\bbH(\widehat{\bbeta} - \bbeta_{\cT})\right\},
	\label{eq:epsilon-error}
    \end{equation}
	which is referred to as the $\cE$-error throughout. The choice of $\bbH$ can be problem-dependent.
	The $\cE$-error reduces to the mean squared error (MSE) of $\widehat{\bbeta}$ when $\bbH = \bbI_{p}$ where $\bbI_{p}$ is the identical matrix of size $p$. When $\bbeta_{\cT}$ is the parameter vector in a linear regression model  $\bX^{\T}\bbeta_{\cT}$, for a  estimator $\hbeta$ of $\bbeta_{\cT}$, the MSE of the predicted regression function $\bX^{\T}\hbeta$ is $(\hbeta - \bbeta_{\cT})^{\T} \bbSig_{\cT}(\hbeta - \bbeta_{\cT})$, where $\bX$ is the covariate vector, $\bbSig_{\cT}=E_{\cT}(\bX\bX^T)$ and the expectation $E_{\cT}(\cdot)$ is taken in the target population.
    In this case, a reasonable choice is $\bbH =  \bbSig_{\cT}$ if one concerns about the prediction performance in the target population. 
    Given a commonly-used standard estimator $\hbeta_{\cT}$ for $\bbeta_{\cT}$ based on target population data $\{\bZ_{i}\}_{i = 1}^{n_{\cT}}$, a 
    transfer estimation method is called \emph{safe} if its $\cE$-error is guaranteed to be no larger than $\cE(\hbeta_{\cT}, \bbeta_{\cT})$. 
	
	\subsection{The calibration and pretest estimators}\label{subsec: cal and pre}
	We first consider two classic approaches to transfer estimation in the target population, the calibration method and the pretest method, under the setting in this paper. To make use of the estimator $\halpha_{\cS}$ from the source population, we construct an estimator $\halpha_{\cT}$ for $\balpha_{\cT} = \balpha(P_{\cT})$ which is $\balpha_{\cS}$'s counterpart in the target population $P_{\cT}$. 
    Assume $(\hbeta_{\cT},\halpha_{\cT})$ are independent of the summary statistics $\halpha_{\cS}$ from the source population. For the sake of simplicity, we derive the estimator under the normality condition
	\begin{equation}\label{eq: normality}
		\small
		\nodisplayskips
	    \halpha_{\cS} \sim N(\balpha_{\cS}, \bbSig_{\alpha,\cS}),\
	\text{and} \ 
	\begin{pmatrix}
		\hbeta_{\cT}\\
		\halpha_{\cT}
	\end{pmatrix} \sim 
	N\left\{
	\begin{pmatrix}
		\bbeta_{\cT}\\
		\balpha_{\cT}
	\end{pmatrix}\\, 
	\begin{pmatrix}
		\bbSig_{\beta, \cT} & \bbSig_{\beta\alpha, \cT}\\
		\bbSig_{\beta\alpha, \cT}^{\T} & \bbSig_{\alpha, \cT}
	\end{pmatrix}
	\right\},
	\end{equation}
	where $\bbSig_{\alpha,\cS}$, $\bbSig_{\alpha, \cT}$, $\bbSig_{\beta\alpha, \cT}$, and $\bbSig_{\beta, \cT}$ are the variance and covariance matrices. The parameters $\balpha_{\cS}$, $\balpha_{\cT}$, $\bbeta_{\cT}$ and matrices $\bbSig_{\alpha,\cS}$, $\bbSig_{\alpha, \cT}$, $\bbSig_{\beta\alpha, \cT}$, and $\bbSig_{\beta, \cT}$ are all allowed to depend on $n_{\cT}$ and $n_{\cS}$ throughout this paper.
	In practice, $\halpha_{\cS}$, $\halpha_{\cT}$, and $\hbeta_{\cT}$ may not be exactly normal, while a broad spectrum of estimators are usually asymptotically normal under general conditions \citep{spokoiny2012parametric}, including scenarios with high-dimensional parameters and correlated data \citep{bentkus2003dependence,hall1981rates}. We motivate our proposed estimator under normality and propose to apply the estimator when normality holds asymptotically.
    Theoretical results when $\halpha_{\cS}$, $\halpha_{\cT}$, and $\hbeta_{\cT}$ are asymptotically normal are presented in Section \ref{subsec: asy normal}. In this section, we proceed under the assumption that the variance and covariance matrices are known. Discussions on using estimated covariance matrices are reserved for Section \ref{subsec: est var}. 
    
 \emph{Calibration Estimator:} Assume $\balpha_{\cT} = \balpha_{\cS}$. Then, the estimator $\hbeta_{\cT} - \bbQ(\halpha_{\cT} - \halpha_{\cS})$ is unbiased for $\bbeta_{\cT}$ for any $p\times q$ matrix $\bbQ$. It is not hard to verify that the matrix $\bbQ$ that minimizes the variance of the resulting estimator is $\bbQ_{\alpha} = \bbSig_{\beta\alpha, \cT}(\bbSig_{\alpha, \cT} + \bbSig_{\alpha,\cS})^{-1}$. Such a choice of $\bbQ$ leads to the  calibration estimator 
	\begin{equation}\label{eq:calib-est}
		\small
		\nodisplayskips
	\hbeta_{\cC} = \hbeta_{\cT} - \bbQ_{\alpha}(\halpha_{\cT} - \halpha_{\cS}).
	\end{equation}
	The variance of $\hbeta_{\cC}$ is no larger than that of $\hbeta_{\cT}$.
	The idea of calibration is widely adopted in survey sampling \citep{deville1992calibration}, data analysis with measurement error \citep{chen2000unified}, data integration \citep{yang2020combining}, and missing data \citep{cannings2022correlation}, and has been proved to be able to achieve the semiparametric efficiency bound in the context of transfer estimation in the presence of population homogeneity \citep{hu2024efusion}. Recently, the idea of calibration has attracted growing interest in machine learning-assisted inference \citep{miao2024valid, miao2025task, kluger2025prediction}.
    
   If one defines $\halpha_{\cS} = \hbeta_{\cS}$ and $\halpha_{\cT} = \hbeta_{\cT}$, where $\hbeta_{\cS}$ is an estimator of the target parameter using source population data and $\bbSig_{\beta,\cS}$ denotes the covariance of $\hbeta_{\cS}$, we have $\bbSig_{\beta\alpha,\cT} = \bbSig_{\beta,\cT}$ and the calibration estimator
    \begin{equation}\label{eq:meta}
    	\small
    	\nodisplayskips
      \hbeta_{\cC} = \hbeta_{\cT} - \bbQ_{\alpha}(\hbeta_{\cT} - \hbeta_{\cS}) = (\bbSig_{\beta,\cT}^{-1} + \bbSig_{\beta,\cS}^{-1})^{-1}(\bbSig_{\beta,\cT}^{-1}\hbeta_{\cT} + \bbSig_{\beta,\cS}^{-1}\hbeta_{\cS}),
    \end{equation}
    which is the commonly-used fixed effect meta-analysis estimator \citep{jackson2011multivariate}. Thus, $\hbeta_{\cC}$ in (\ref{eq:calib-est}) can be viewed as a generalized meta-analyzed estimator integrating information from different populations.
    
   \emph{Pretest Estimator:} Although being efficient when $\balpha_{\cT} = \balpha_{\cS}$, $\hbeta_{\cC}$ can be severely biased if  $\balpha_{\cT} \neq \balpha_{\cS}$.  Thus, $\hbeta_{\cC}$ is ideally employed if $\balpha_{\cT} = \balpha_{\cS}$,
   whereas $\hbeta_{\cT}$ is preferable to circumvent potential bias otherwise. However, it's often unclear whether $\balpha_{\cT}$ aligns with $\balpha_{\cS}$.  To address this, the pretest strategy conducts a test for the hypotheses $\balpha_{\cT} = \balpha_{\cS}$, with the choice of estimator informed by the test result. To be specific, let $T_{\rm g} = (\halpha_{\cT} - \halpha_{\cS})^{\T}(\bbSig_{\alpha, \cT} + \bbSig_{\alpha,\cS})^{-1}(\halpha_{\cT} - \halpha_{\cS})$ be the $\chi^{2}$-statistics. Then, the pretest estimator is
    \[
    \small
    \nodisplayskips
	\hbeta_{\rm prt} = 1\{T_{\rm g} \leq c_{\rm g}\}\hbeta_{\cC} + 1\{T_{\rm g} > c_{\rm g}\}\hbeta_{\cT},
    \]
    where $c_{\rm g}$ is the threshold that depends on the significance level. The concept of pretest has been a foundational strategy in data integration and transfer estimation, tracing back decades with notable early applications by \cite{mosteller1948pooling} and continuing to inspire contemporary advancements as seen in works by \cite{cai2022semi}, \cite{yang2023elastic}, and \cite{gao2023pretest}. 
   
    \subsection{The calibration and pretest estimators are unsafe}
    The calibration estimator is generally unsafe, as its bias and $\cE$-error can be large when $\|\bbQ_{\alpha}(\balpha_{\cS} - \balpha_{\cT})\|$ is large, where we use $\|\cdot\|$ to denote the Euclid/spectral norm of a vector/matrix throughout the paper. 
    The situation for the pretest estimator is more subtle, owing to the additional pretesting step designed to mitigate bias.
    When the difference $\balpha_{\cT} - \balpha_{\cS}$ is large in the sense that $\|\balpha_{\cT} - \balpha_{\cS}\|^{2} \gg \|\bbSig_{\alpha, \cT} + \bbSig_{\alpha,\cS}\|$, the test statistics $T_{\rm g}$ tends to be large and then $\hbeta_{\rm prt}$ is identical to the unbiased estimator $\hbeta_{\cT}$ with a high probability. This suggests that $\hbeta_{\rm prt}$ can mitigate the bias issue of the calibration estimator $\hbeta_{\cC}$ in this case. On the other hand, $\hbeta_{\rm prt} = \hbeta_{\cT}$ with a positive probability when $\balpha_{\cS} = \balpha_{\cT}$, due to the type-I error of the test. Hence, the pretest estimator $\hbeta_{\rm prt}$ may not be as efficient as $\hbeta_{\cC}$ when populations are homogeneous. 
    
     Suppose that $\|\balpha_{\cT} - \balpha_{\cS}\|^{2} = 0$ when the populations are homogeneous and that $\|\balpha_{\cT} - \balpha_{\cS}\|^{2} / \|\bbSig_{\alpha, \cT} + \bbSig_{\alpha,\cS}\| \to \infty$ as $n_{\cT}, n_{\cS} \to \infty$ when they are heterogeneous. Then, the efficiency loss in the homogeneous setting can be avoided asymptotically by letting the significance level of the test go to zero at a proper rate as $n_{\cT}, n_{\cS} \to \infty$. In this case, $\hbeta_{\rm prt}$ enjoys the oracle property: it is asymptotically equivalent to $\hbeta_{\cC}$ when populations are homogeneous and to $\hbeta_{\cT}$ when they are heterogeneous. The oracle properties of penalization-based transfer estimation methods \citep{chen2021combining, zhai2022data,huang2023simultaneous} follow the same rationale as that of $\hbeta_{\rm prt}$. In practice, while the target and source populations may be similar, they are generally not identical. Moreover, in settings where transfer estimation methods are applied, the sample size from the target population is typically limited. This situation makes the assumption $\|\balpha_{\cT} - \balpha_{\cS}\|^{2} / \|\bbSig_{\alpha, \cT} + \bbSig_{\alpha,\cS}\| \to \infty$ unrealistic. The problem becomes particularly challenging when the difference between $\balpha_{\cT}$ and $\balpha_{\cS}$ is moderate, placing $\|\balpha_{\cT} - \balpha_{\cS}\|^{2}$ and $\|\bbSig_{\alpha, \cT} + \bbSig_{\alpha,\cS}\|$ on comparable scales. In such instances, the bias introduced to $\hbeta_{\cC}$ may be significant enough to warrant concern, yet insufficient for detection through pretest, leading $\hbeta_{\rm prt}$ to be unsafe. 
    This dilemma highlights the limitations of relying solely on testing to guide the integration of external summary statistics, and cannot be resolved by carefully choosing the threshold $c_{\rm g}$ in general. To demonstrate this issue, we provide a concrete example. Suppose $n_{\cT}, n_{\cS} \geq 1$,
	\begin{equation}\label{eq: counter eg}
		\small
		\nodisplayskips
		\begin{aligned}
		    \hbeta_{\cT} \sim N(\bbeta_{\cT}, 2n_{\cT}^{-1}\bbI_{p}),\  \hbeta_{\cS} \sim N(\bbeta_{\cS}, 2n_{\cS}^{-1}\bbI_{p}), \ \halpha_{\cT} = \hbeta_{\cT}, \ \halpha_{\cS} = \hbeta_{\cS},       \bbeta_{\cT} = \boldsymbol{\mu} + \bbias, \ \text{and}\ \bbeta_{\cS} = \boldsymbol{\mu} - \bbias,
		\end{aligned} 
	\end{equation}
	where $\boldsymbol{\mu}$ is a location parameter and $\bbias$ characterizes the difference between $\balpha_{\cT}$ and $\balpha_{\cS}$. Then we have the following proposition.
	\begin{proposition}\label{prop: pretest}
		In the example stated in \eqref{eq: counter eg}, for any $c_{\rm g} > 0$, $n_{\cT}$, $n_{\cS}$, $\boldsymbol{\mu}$, and positive definite $\bbH$, there is some $\bbias$ such that $\cE(\hbeta_{\rm prt}, \bbeta_{\cT}) > \cE(\hbeta_{\cT}, \bbeta_{\cT})$.
	\end{proposition}
    Proposition \ref{prop: pretest} demonstrates that, for any chosen significance level, there exists a bias magnitude such that the pretest estimator's $\cE$-error surpasses that of the baseline estimator $\hbeta_{\cT}$. This proposition thus suggests the pretest estimator is unsafe in transfer estimation. Its proof can be found in Supplementary Material Section \ref{app: proof prop}.
    Further explorations in the Supplemental Material  Section \ref{app: comp test} expand on the above issue by considering component-wise testing strategies.	
	
	\section{Safe transfer using the dShrink estimator}\label{sec: dShrink}
	\subsection{The dShrink estimator}\label{subsec: ds}
	In this section, we introduce a safe transfer estimator that is guaranteed to outperform the target population-based estimator $\hbeta_{\cT}$ in $\cE$-error. To lay the groundwork for our proposed estimator, we offer a new interpretation of the calibration estimator $\hbeta_{\cC}$ and the pretest estimator $\hbeta_{\rm prt}$.
	
	Let $\hdelta_{\cC} = \bbQ_{\alpha}(\halpha_{\cT} - \halpha_{\cS})$, which is the projection of $\hbeta_{\cT}$ on $\halpha_{\cT} - \halpha_{\cS}$, the estimate of the parameter difference $\balpha_{\cT} - \balpha_{\cS}$ between the target and source populations. Then, $\hbeta_{\cC}$ is independent of $\hdelta_{\cC}$ and the target population-based estimator $\hbeta_{\cT}$ can be expressed as the sum of two independent components $\hbeta_{\cT} = \hbeta_{\cC} + \hdelta_{\cC}$. The independence between $\hbeta_{\cC}$ and $\hdelta_{\cC}$ ensures that $\var(\hbeta_{\cT}) = \var(\hbeta_{\cC}) + \var(\hdelta_{\cC})$. In addition, notice that $E[\hdelta_{\cC}] = \bzero$ when $\balpha_{\cT} = \balpha_{\cS}$. These observations indicate that replacing $\hdelta_{\cC}$ with zero reduces $\hbeta_{\cT}$'s variance without affecting its expectation, leading to the calibration estimator $\hbeta_{\cC} = \hbeta_{\cC} + \bzero$. However, this substitution introduces bias if $\balpha_{\cT} \neq \balpha_{\cS}$. The pretest estimator, $\hbeta_{\rm prt} = \hbeta_{\cC} + 1\{T_{\alpha} > c_{\rm g}\}\hdelta_{\cC}$, offers mitigation by adopting this substitution based on data; it replaces $\hdelta_{\cC}$ with $\bzero$ only if the test fails to reject the null hypothesis that $\balpha_{\cT} = \balpha_{\cS}$. It is noteworthy that $\hbeta_{\cT}$, $\hbeta_{\cC}$, and $\hbeta_{\rm prt}$ are all linear combinations of $\hbeta_{\cC}$ and $\hdelta_{\cC}$. Among them, $\hbeta_{\cT}$ and $\hbeta_{\cC}$ utilize a fixed combination coefficient, while $\hbeta_{\rm prt}$ data-adaptively determines the combination coefficient. However, the combination coefficients for all three estimators can only be equal to $``0"$ or $``1"$, which is an unattractive restriction. Moreover, these combination coefficients do not directly take into account the error of the resulting estimator. 
    
     To achieve {\it safe} and {\it efficient} transfer estimation that reduces the estimation error of the parameter of interest in the target population, we propose to construct an estimator of the target population parameter $\bbeta_{\cT}$ by optimally combining $\hbeta_{\cC}$ and $\hdelta_{\cC}$ through the estimation of the combination coefficients that minimize a objective function. The resulting dShrink estimator is robust, efficient, and guaranteed to outperform the target population-based estimator $\hbeta_{\cT}$.

    Motivated by the above observations, we consider to estimate $\bbeta_{\cT}$ using the linear combination $\lambda_{1}\hbeta_{\cC} + \lambda_{2}\hdelta_{\cC}$ for some combination coefficients $\lambda_{1}$ and $\lambda_{2}$. Let $\bbeta_{\cC} = E[\hbeta_{\cC}] = \bbeta_{\cT} - \bdelta_{\cC}$, $\bdelta_{\cC} = E[\hdelta_{\cC}] = \bbQ_{\alpha}(\balpha_{\cT} - \balpha_{\cS})$, $\bbSig_{\cC} = \var(\hbeta_{\cC}) = \bbSig_{\beta, \cT} - \bbSig_{\beta\alpha, \cT}(\bbSig_{\alpha, \cT} + \bbSig_{\alpha,\cS})^{-1}\bbSig_{\beta\alpha, \cT}^{\T}$, and $\bbSig_{\delta} = \var(\hdelta_{\cC}) = \bbSig_{\beta\alpha, \cT}(\bbSig_{\alpha, \cT} + \bbSig_{\alpha,\cS})^{-1}\bbSig_{\beta\alpha, \cT}^{\T}$. We call $\bdelta_{\cC}$ the heterogeneity parameter. Instead of directly minimizing the $\cE$-error $\cE(\lambda_{1}\hbeta_{\cC} + \lambda_{2}\hdelta_{\cC}, \bbeta_{\cT})$, we consider obtaining the combination coefficients $(\lambda_1,\lambda_2)$ by minimizing the following separable surrogate objective function
    \begin{equation}\label{eq: separable surrogate}
    \small
    \nodisplayskips
    \cE(\lambda_{1}\hbeta_{\cC}, \bbeta_{\cC}) + \cE(\lambda_{2}\hdelta_{\cC}, \bdelta_{\cC}).
	\end{equation}
    The rationale behind the objective function \eqref{eq: separable surrogate}, as well as the reasons it is preferable to the original $\mathcal{E}$-error, will be discussed in the next section.
    It can verified that $(\lambda_{{\rm s}, 1}, \lambda_{{\rm s}, 2}) = (1 - \tr\{\bbH\bbSig_{\cC}\} / [\tr\{\bbH\bbSig_{\cC}\} + \bbeta_{\cC}^{\T}\bbH\bbeta_{\cC}],
	1 - \tr\{\bbH\bbSig_{\delta}\} / [\tr\{\bbH\bbSig_{\delta}\} + \bdelta_{\cC}^{\T}\bbH\bdelta_{\cC}])$ is the minimizer of \eqref{eq: separable surrogate}. Note that $E(\hbeta_{\cC}^{\T}\bbH\hbeta_{\cC}) = \tr\{\bbH\bbSig_{\cC}\} + \bbeta_{\cC}^{\T}\bbH\bbeta_{\cC}$ and $E(\bdelta_{\cC}^{\T}\bbH\bdelta_{\cC}) = \tr\{\bbH\bbSig_{\delta}\} + \bdelta_{\cC}^{\T}\bbH\bdelta_{\cC}$. According to the discussions in Supplementary Material Section \ref{app: derive dShrink}, we propose the adjusted plug-in estimators 
	\begin{equation}\label{eq: lambda ds}
	\small
	\nodisplayskips
	\widehat{\lambda}_{{\rm s}, 1} = \left(1 -  \frac{c_{\beta}}{\hbeta_{\cC}^{\T}\bbH\hbeta_{\cC}}\right)_{+}\ \text{and}\ 
	\widehat{\lambda}_{{\rm s}, 2} = \left(1 - \frac{c_{\delta}}{\hdelta_{\cC}^{\T}\bbH\hdelta_{\cC}}\right)_{+}
	\end{equation}
    for  $\lambda_{{\rm s}, 1}$ and $\lambda_{{\rm s}, 2}$,
     where $c_{\beta} = (\tr\{\bbH\bbSig_{\cC}\} - 2\sigma_{\rm max}(\bbH\bbSig_{\cC}))_{+}$, $c_{\delta} = (\tr\{\bbH\bbSig_{\delta}\} - 2\sigma_{\rm max}(\bbH\bbSig_{\delta}))_{+}$,
    $(u)_{+} = \max\{0, u\}$ for any $u$, and
    $\sigma_{\rm max}(\bbM)$ is the maximum singular value of $\bbM$ for any matrix $\bbM$. 
    By plugging in $\widehat{\lambda}_{{\rm s}, 1}$ and $\widehat{\lambda}_{{\rm s}, 2}$, we obtain the dShrink estimator
	\begin{equation}\label{eq: dShrink}
		\small
		\nodisplayskips
		\hbeta_{\rm ds} = \widehat{\lambda}_{{\rm s}, 1}\hbeta_{\cC} + \widehat{\lambda}_{{\rm s}, 2}\hdelta_{\cC}.
	\end{equation}
    Here we assume that $\bbH$ and the variance-covariance matrices are known. The discussion on plugging in estimators for these matrices is deferred to Section \ref{subsec: est var}.
	Note that $\widehat{\lambda}_{{\rm s}, 1}, \widehat{\lambda}_{{\rm s}, 2} \leq 1$. The dShrink estimator $\hbeta_{\rm ds}$ divides the baseline estimator $\hbeta_{\cT}$ into two independent components, $\hbeta_{\cC}$ and $\hdelta_{\cC}$, and shrinks them to zero adaptively. The shrinkage enables the resulting estimator to achieve a significant efficiency gain compared to $\hbeta_{\cT}$ when either of the targeted true values of the two components, i.e., $\bbeta_{\cC}$ or $\bdelta_{\cC}$, is close to zero. The term $\bdelta_{\cC}$ is close to zero when the source population is similar to the target population. Moreover, if the true parameter $\bbeta_{\cT}$ is close to zero, i.e., with small effect sizes, 
   and the populations are alike, both $\bbeta_{\cC}$ and $\bdelta_{\cC}$ are expected to be small, thereby allowing the dShrink estimator to achieve substantial efficiency gains. Importantly, the shrinkage process is adaptively data-driven, ensuring that efficiency is not compromised regardless of the magnitudes of $\bbeta_{\cC}$ and $\bdelta_{\cC}$. 
	Theorem \ref{thm: dominance ds} shows that the dShrink estimator $\hbeta_{\rm ds}$ is guaranteed to outperform the target population-based estimator $\hbeta_{\cT}$.
	
	\begin{theorem}\label{thm: dominance ds}
		Suppose $\halpha_{\cS}$, $\halpha_{\cT}$, and $\hbeta_{\cT}$ are normally distributed as in \eqref{eq: normality}. Then, $\cE(\hbeta_{\rm ds},\bbeta_{\cT}) \leq \cE(\hbeta_{\cT}, \bbeta_{\cT})$ for any $\balpha_{\cS}$, $\balpha_{\cT}$, and  $\bbeta_{\cT}$.
		If $c_{\beta} > 0$ or $c_{\delta} > 0$, then $\cE(\hbeta_{\rm ds},\bbeta_{\cT}) < \cE(\hbeta_{\cT}, \bbeta_{\cT})$ for any $\balpha_{\cS}$, $\balpha_{\cT}$, and  $\bbeta_{\cT}$.
	\end{theorem}

      Proof for Theorem \ref{thm: dominance ds} is in Supplementary Material Section \ref{app: proof of dominance}. Notably, Theorem \ref{thm: dominance ds} does not rely on additional conditions other than the normality of $\hbeta_{\cT}$, $\halpha_{\cT}$ and $\halpha_{\cS}$. It theoretically justifies that the dShrink estimator $\hbeta_{\rm ds}$ always outperforms the target population-based estimator $\hbeta_{\cT}$ in terms of $\cE$-error regardless of the potential parameter and population heterogeneity. 
    
    \subsection{Rationales behind the dShrink estimator}\label{subsec: SURE}
    
     This section provides the rationale behind the objective function \eqref{eq: separable surrogate} that leads to the dShrink estimator. 
    As discussed in the previous section, we would like to explore estimators of the form $\lambda_{1}\hbeta_{\cC} + \lambda_{2}\hdelta_{\cC}$. A natural idea is to obtain 
    the coefficients $\lambda_{1}$ and $\lambda_{2}$ that minimize
    the $\cE$-error of $\lambda_{1}\hbeta_{\cC} + \lambda_{2}\hdelta_{\cC}$, which can be expressed as
	\begin{equation}\label{eq: original objective}
		\small
		\nodisplayskips
           \begin{aligned}
		\cE(\lambda_{1}\hbeta_{\cC} + \lambda_{2}\hdelta_{\cC}, \bbeta_{\cT}) & = \tr\{\bbH\bbSig_{\beta, \cT}\} + 
		(\lambda_{1} - 1, \lambda_{2} - 1)
		\bbB
		\begin{pmatrix}
			\lambda_{1} - 1\\
			\lambda_{2} - 1
		\end{pmatrix}
	    + 
		(2\tr\{\bbH\bbSig_{\cC}\}, 2\tr\{\bbH\bbSig_{\delta}\})
		\begin{pmatrix}
			\lambda_{1} - 1\\
			\lambda_{2} - 1
		\end{pmatrix},
	   \end{aligned}
        \end{equation}
	where
	\[
	\small
	\nodisplayskips
	\bbB = 
	E\begin{pmatrix}
		\hbeta_{\cC}^{\T}\bbH\hbeta_{\cC} & \hbeta_{\cC}^{\T}\bbH\hdelta_{\cC}\\
		\hdelta_{\cC}^{\T}\bbH\hbeta_{\cC} & \hdelta_{\cC}^{\T}\bbH\hdelta_{\cC}
	\end{pmatrix} = 
        \begin{pmatrix}
	\bbeta_{\cC}^{\T}\bbH\bbeta_{\cC} + \tr\{\bbH\bbSig_{\cC}\} & 
        \bbeta_{\cC}^{\T}\bbH\bdelta_{\cC}\\
		\bdelta_{\cC}^{\T}\bbH\bbeta_{\cC} & \bdelta_{\cC}^{\T}\bbH\bdelta_{\cC} + \tr\{\bbH\bbSig_{\delta}\}
	\end{pmatrix}.
	\]
Then, $(\lambda_{{\rm opt}, 1}, \lambda_{{\rm opt}, 2}) = (1, 1) - (\tr\{\bbH\bbSig_{\cC}\}, \tr\{\bbH\bbSig_{\delta}\})\bbB^{-1}$ is the optimal coefficient that minimizes the $\cE$-error.
An unbiased estimate of the above $\cE$-error, SURE, is  
	\[
	\small
	\nodisplayskips
	\tr\{\bbH\bbSig_{\beta, \cT}\} + 
	(\lambda_{1} - 1, \lambda_{2} - 1)
	\hB
	\begin{pmatrix}
		\lambda_{1} - 1\\
		\lambda_{2} - 1
	\end{pmatrix} +  
	(2\tr\{\bbH\bbSig_{\cC}\}, 2\tr\{\bbH\bbSig_{\delta}\})
	\begin{pmatrix}
		\lambda_{1} - 1\\
		\lambda_{2} - 1
	\end{pmatrix},
	\]
	where
	\[
	\hB = 
	\begin{pmatrix}
		\hbeta_{\cC}^{\T}\bbH\hbeta_{\cC} & \hbeta_{\cC}^{\T}\bbH\hdelta_{\cC}\\
		\hdelta_{\cC}^{\T}\bbH\hbeta_{\cC} & \hdelta_{\cC}^{\T}\bbH\hdelta_{\cC}
	\end{pmatrix}.\]
	The combination coefficient that minimizes SURE is
	\begin{equation}\label{eq: lambda sure}
		\small
		\nodisplayskips
	    \begin{pmatrix}
		\widehat{\lambda}_{\SURE,1}\\
		\widehat{\lambda}_{\SURE, 2}
	\end{pmatrix} = 
	\begin{pmatrix}
		1 \\
		1
	\end{pmatrix} - 
	\widehat{\bbB}^{-1}
	\begin{pmatrix}
		\tr\{\bbH\bbSig_{\cC}\}\\
		\tr\{\bbH\bbSig_{\delta}\}
	\end{pmatrix}.
	\end{equation}
	This yields the SURE estimator
	$\hbeta_{\SURE} = \widehat{\lambda}_{\SURE,1} \hbeta_{\cC} + \widehat{\lambda}_{\SURE,2}\hdelta_{\cC}$.
    The SURE estimator uses an unbiased estimator for the $\cE$-error to determine the combination coefficients $\lambda_1$ and $\lambda_2$. However, it is observed in our numerical results that $\hbeta_{\SURE}$ may exhibit a higher estimation error than $\hbeta_{\cT}$   in finite samples when $\bbB$ is nearly singular.

    For illustration, we conduct a simulation study under the setting of \eqref{eq: counter eg}. Figure \ref{fig: eg-ds} shows the MSEs of different estimators under different heterogeneity levels (characterized by $t$) calculated based on $5000$ simulation runs. The results clearly demonstrate that $\hbeta_{\rm ds}$ consistently yields smaller error than $\hbeta_{\cT}$, empirically validating the dShrink estimator's capability in reducing the MSE. Both the calibration and pretest estimators behave as anticipated. However, 
    the SURE estimator underperforms  $\hbeta_{\cT}$, especially at small values of $t$.
        \begin{figure}
		\centering
		\includegraphics[scale = 0.25]{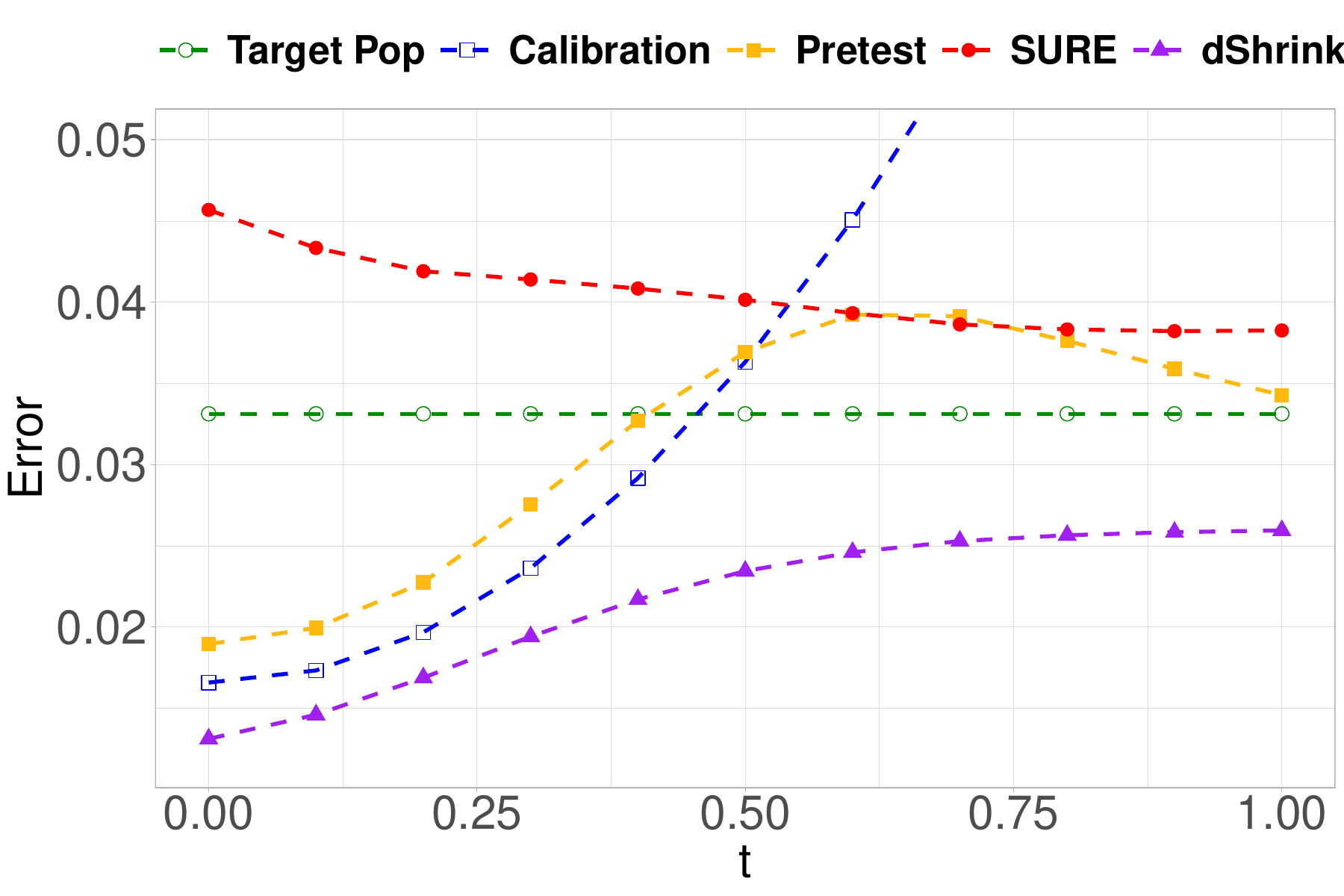}
		\caption{\it MSEs of $\hbeta_{\cT}$, $\hbeta_{\cC}$, $\hbeta_{\rm prt},$ and $\hbeta_{\rm ds}$ in the simulation under the setting of \eqref{eq: counter eg} with $\bbH = \bbI_{p}$, $p = 5$, $n_{\cT} = n_{\cS} = 300$, $\bbeta_{\cT} = \balpha_{\cT} = \boldsymbol{\mu} + \bbias,$  $\bbeta_{\cS} = \balpha_{\cS} = \boldsymbol{\mu} - \bbias$, $\boldsymbol{\mu} = (0.05, 0.02, 0.1, 0.1, 0.1) ^{\T} / \sqrt{5}$,  $\bbias = t\etab$, $\etab = (0.2, 0.3, 0.3, 0.3, 0.3)^{\T} / \sqrt{5}$ and different $t$'s.}\label{fig: eg-ds}
	\end{figure}

    We delve into the reasons behind the observed phenomena concerning the performance of the SURE estimator, which will further motivate the objective \eqref{eq: separable surrogate}.
    Specifically, the estimation error of  $\widehat{\bbB}^{-1}$ can be large when $\bbB$ is nearly singular. The matrix $\bbB$ will be nearly singular either when $\bbeta_{\cC} = \boldsymbol{\mu}$ is small, i.e., the effects are weak or zero, or when $\bdelta_{\cC} = \bbias$ is small,  e.g., the target and source populations are similar. In the illustrative example,  $\delta_{\cC}$ is small when $t$ is small.
    In addition, the condition number of $\bbB$ is larger than 
	$
	\max\{E[\hbeta_{\cC}^{\T}\bbH\hbeta_{\cC}], E[\hdelta_{\cC}^{\T}\bbH\hdelta_{\cC}]\}/\min\{E[\hbeta_{\cC}^{\T}\bbH\hbeta_{\cC}], E[\hdelta_{\cC}^{\T}\bbH\hdelta_{\cC}]\}$.
    Thus, the condition number of $\bbB$ is large when $E[\hbeta_{\cC}^{\T}\bbH\hbeta_{\cC}] / E[\hdelta_{\cC}^{\T}\bbH\hdelta_{\cC}]$ or $E[\hdelta_{\cC}^{\T}\bbH\hdelta_{\cC}] / E[\hbeta_{\cC}^{\T}\bbH\hbeta_{\cC}]$ is large.
	The large condition number suggests that the estimation error associated with $\widehat{\bbB}^{-1}$ can be substantial compared to the error of $\widehat{\bbB}$ itself  \citep{higham1994survey}. This increased estimation error can adversely affect the performance of the $\hbeta_{\SURE}$ estimator, leading to situations where $\hbeta_{\SURE}$ underperforms compared to $\hbeta_{\cT}$.

     The above challenges in estimating the minimizer of $\cE(\lambda_{1}\hbeta_{\cC} + \lambda_{2}\hdelta_{\cC}, \bbeta_{\cT})$ lies in the ill-posed nature of $\bbB$  in transfer estimation problems.
     To mitigate this problem, we use the surrogate objective \eqref{eq: separable surrogate} instead of the original objective function. The surrogate objective is separable in $\lambda_{1}$ and $\lambda_{2}$ and hence it can be minimized by considering two $1$-dimensional quadratic optimization problems, simplifying the estimation of the combination coefficients. Given that the condition number for a one-dimensional quadratic function is one, this approach is expected to alleviate the ill-posedness associated with the original objective function.  We revisit the example considered in Figure \ref{fig: eg-ds} to compare the estimation errors of the combination coefficients $(\widehat{\lambda}_{{\rm s}, 1}, \widehat{\lambda}_{{\rm s}, 2})$ and $(\widehat{\lambda}_{\SURE, 1}, \widehat{\lambda}_{\SURE, 2})$ adopted in the dShrink and SURE estimators with respect to their population counterparts $(\lambda_{\rm s, 1}, \lambda_{\rm s, 2})$ and $(\lambda_{{\rm opt}, 1}, \lambda_{{\rm opt}, 2})$. 
     Figure \ref{fig: eg-weight} shows the MSEs of the combination coefficient estimators under different $t$'s. In addition, although $\lambda_{{\rm s}, t} \neq \lambda_{{\rm opt}, t}$ in general for $t = 1, 2$, Figure \ref{fig: eg-weight2} in Supplementary Material shows that the MSE $E[(\widehat{\lambda}_{{\rm s}, t} - \lambda_{{\rm opt}, t})^{2}]$ is smaller than $E[(\widehat{\lambda}_{{\rm opt}, t} - \lambda_{{\rm opt}, t})^{2}]$ in most cases in the above example. These results suggest that the separable surrogate effectively reduces the estimation error of the combined coefficient estimator.
     
      One may notice that the expressions of $\widehat{\lambda}_{{\rm s}, 1}$ and $\widehat{\lambda}_{{\rm s}, 2}$ involve $\hbeta_{\cC}^{\T}\bbH\hbeta_{\cC}$ and $\hdelta_{\cC}^{\T}\bbH\hdelta_{\cC}$ in the denominator. At first glance, this could raise concerns about potential instability when $\bbeta_{\cC}$ or $\bdelta_{\cC}$ is small. However, Figure \ref{fig: eg-weight} shows that these estimates are quite stable and Theorem 1 establishes the safety of $\hbeta_{\rm ds}$ with no assumption on $\bbeta_{\cC}$ and $\bdelta_{\cC}$. We provide some explanation of this stability. To illustrate, take the case with a small $\bdelta_{\cC}$ as an example. We focus on the estimation error of $\widehat{\lambda}_{{\rm s}, 2}$ because the value of $\bdelta_{\cC}$ has no impact on the estimator $\widehat{\lambda}_{{\rm s}, 1}$. In such cases, the population-level combination coefficient $\lambda_{{\rm s}, 2}$ is close to zero and $\hdelta_{\cC}$ tend to be small as well. If $\hdelta_{\cC}$ is very small, $\widehat{\lambda}_{{\rm s}, 2} = (1 - c_{\delta}/\hdelta_{\cC}^{\T}\bbH\hdelta_{\cC})_{+}$ is zero
     and hence close to $\lambda_{{\rm s}, 2}$. This ensures that the estimation error of $\widehat{\lambda}_{{\rm s}, 2}$ remains controlled even if $\bdelta_{\cC}$ is small. Moreover, noticing that $\|\widehat{\lambda}_{{\rm s}, 2}\hdelta_{\cC} - \lambda_{{\rm s}, 2}\hdelta_{\cC}\| \leq \|\hdelta_{\cC}\||\widehat{\lambda}_{{\rm s}, 2} -  \lambda_{{\rm s}, 2}|$, the estimation error of $\widehat{\lambda}_{{\rm s}, 2}$ does not contribute significantly to the resulting dShrink estimator if $\bdelta_{\cC}$ and hence $\hdelta_{\cC}$ are small. This further enhances the robustness of the dShrink estimator against a small $\bdelta_{\cC}$. On the other hand, although $\|\widehat{\lambda}_{\SURE, 2}\hdelta_{\cC} - \lambda_{{\rm opt}, 2}\hdelta_{\cC}\| \leq \|\hdelta_{\cC}\||\widehat{\lambda}_{\SURE, 2} -  \lambda_{{\rm opt}, 2}|$, the above robustness property of the dShrink estimator $\hbeta_{\rm ds}$ is not shared by the SURE estimator $\hbeta_{\SURE}$ because a small $\bdelta_{\cC}$ may affect the estimation accuracy of not only $\widehat{\lambda}_{\SURE, 2}$ but also $\widehat{\lambda}_{\SURE, 1}$.
	
	\begin{figure}
		\centering
		\includegraphics[scale = 0.4]{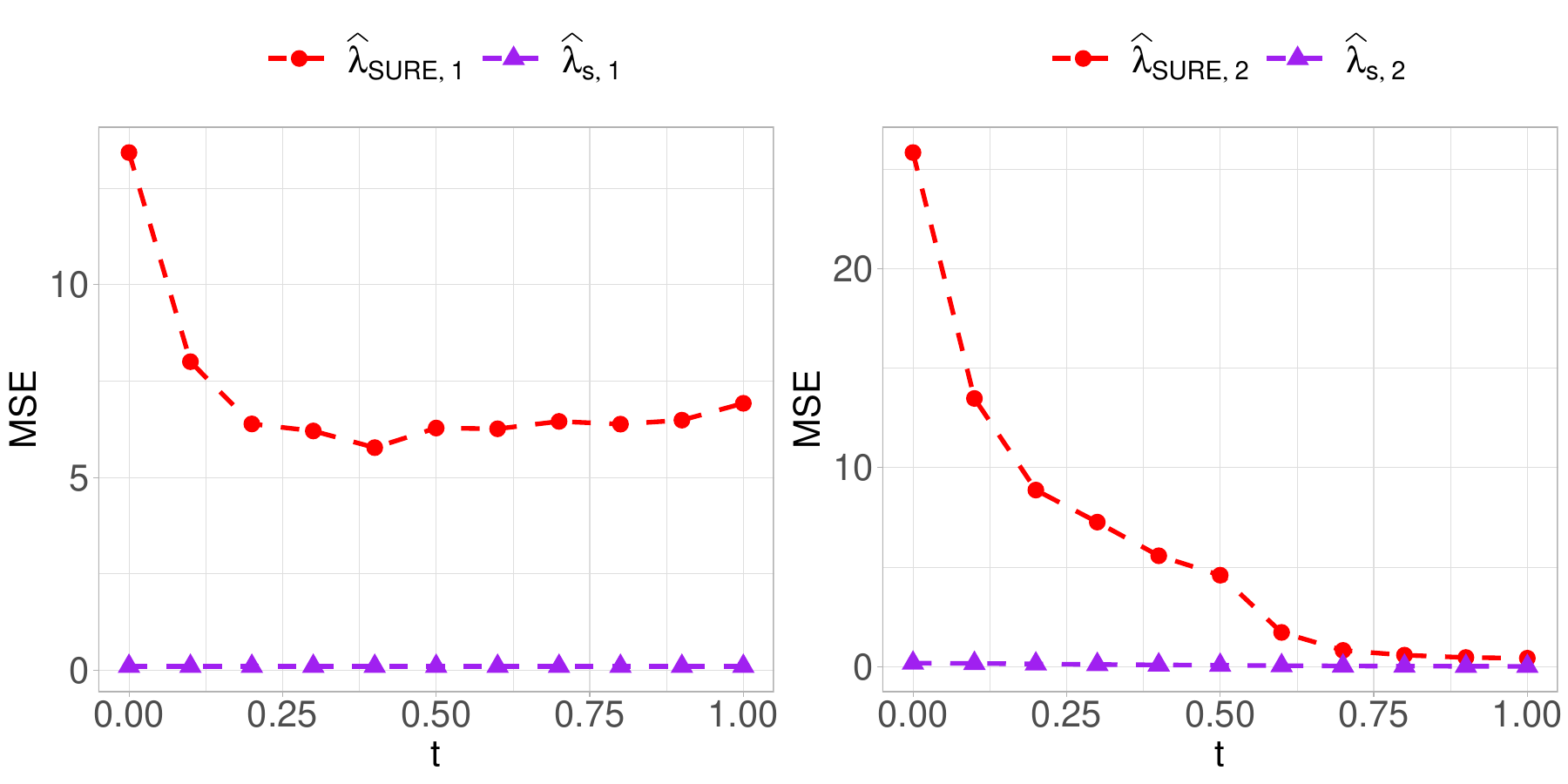}
		\caption{\it MSEs of $\widehat{\lambda}_{{\rm s}, 1}$, $\widehat{\lambda}_{{\rm s}, 2}$ $\widehat{\lambda}_{{\SURE}, 1}$, and $\widehat{\lambda}_{{\SURE}, 2}$ defined in \eqref{eq: lambda ds} and \eqref{eq: lambda sure} with respect to their population counterparts $\lambda_{\rm s, 1}$, $\lambda_{\rm s, 2}$, $\lambda_{{\rm opt}, 1}$, and $\lambda_{{\rm opt}, 2}$  in the simulation under the setting of \eqref{eq: counter eg} with $\bbH = \bbI_{p}$, $p = 5$, $n_{\cT} = n_{\cS} = 300$,  $\bbeta_{\cT} = \balpha_{\cT} = \boldsymbol{\mu} + \bbias,$  $\balpha_{\cS} = \boldsymbol{\mu} - \bbias,$ $\boldsymbol{\mu} = (0.05, 0.02, 0.1, 0.1, 0.1) ^{\T} / \sqrt{5}$, 
        $\bbias = t\etab$, $\etab = (0.2, 0.3, 0.3, 0.3, 0.3)^{\T} / \sqrt{5}$ and different $t$'s.}\label{fig: eg-weight}
	\end{figure}
    
    \begin{remark} We provide some intuition on the error of using the surrogate objective \eqref{eq: separable surrogate} to approximate \eqref{eq: original objective}.  Note that the surrogate objective \eqref{eq: separable surrogate} can be derived by setting the off-diagonal elements of $\bbB$ (which are equal to $\bbeta_{\cC}^{\T}\bbH\bdelta_{\cC}$) in the original objective function \eqref{eq: original objective} to zero. The surrogate objective serves as a good approximation of the original objective when $\bbeta_{\cC}^{\T}\bbH\bdelta_{\cC}$ is small. For any vectors $\boldsymbol{v}_{1}$ and  $\boldsymbol{v}_{2}$, define the inner product $\langle \boldsymbol{v}_{1}, \boldsymbol{v}_{2} \rangle_{\bbH} =  \boldsymbol{v}_{1}^{\T}\bbH\boldsymbol{v}_{2}$, norm $\|\boldsymbol{v}_{1}\|_{\bbH} = \sqrt{\langle \boldsymbol{v}_{1}, \boldsymbol{v}_{1} \rangle_{\bbH}}$ and the cosine similarity $ \cos_{\bbH}(\boldsymbol{v}_{1}, \boldsymbol{v}_{2}) = \langle\boldsymbol{v}_{1}, \boldsymbol{v}_{2}\rangle_{\bbH}/(\|\boldsymbol{v}_{1}\|_{\bbH} \|\boldsymbol{v}_{2}\|_{\bbH})$  between $\boldsymbol{v}_{1}$ and $\boldsymbol{v}_{2}$.
    Then, the off-diagonal term $\bbeta_{\cC}^{\T}\bbH\bdelta_{\cC} = \|\bbeta_{\cC}\|_{\bbH} \|\bdelta_{\cC}\|_{\bbH}\cos_{\bbH}(\bbeta_{\cC}, \bdelta_{\cC})$. Consequently, the surrogate objective closely approximates the original objective when either $\bbeta_{\cC}$ or $\bdelta_{\cC}$ is close $\bzero$, which happens under the null or the homogeneous case  or  $\bbeta_{\cC}$ is largely unrelated to $\bdelta_{\cC}$. 
   
     On the other hand, if both diagonal elements of $\bbB$, i.e., $\|\bbeta_{\cC}\|_{\bbH} + \tr\{\bbH\bbSig_{\cC}\}$ and $\|\bdelta_{\cC}\|_{\bbH} + \tr\{\bbH\bbSig_{\bdelta}\}$, are large, $\bbB^{-1}$ is close to zero and the population-level optimal combination coefficient $(\lambda_{{\rm opt}, 1}, \lambda_{{\rm opt}, 2})$ that minimizes
     the $\cE$-error in (\ref{eq: original objective})
     is close to $(1, 1)$.
    At the point $(\lambda_{1}, \lambda_{2}) = (1, 1)$, the surrogate objective \eqref{eq: separable surrogate} has the identical value and gradient to the original objective function $\cE(\lambda_{1}\hbeta_{\cC} + \lambda_{2}\hdelta_{\cC}, \bbeta_{\cT})$ in 
    \eqref{eq: original objective}, indicating a locally close approximation around $(\lambda_{1}, \lambda_{2}) = (1, 1)$. Consequently, employing the separable surrogate objective function \eqref{eq: separable surrogate} used by the dShrink estimator $\hbeta_{\rm ds}$ results in minimal efficiency loss, even when both $\|\bbeta_{\cC}\|_{\bbH}$ and $\|\bdelta_{\cC}\|_{\bbH}$ are large. 
    
    Recall that $(\lambda_{{\rm s}, 1}, \lambda_{{\rm s}, 2})$ is the minimizer of the surrogate objective \eqref{eq: separable surrogate}. We show in  Supplemental Materials \ref{app: proof pop dominance} that 
	\begin{equation}\label{eq: pop dominance}
		\small
		\nodisplayskips
		\cE(\lambda_{{\rm s},1}\hbeta_{\cC} + \lambda_{{\rm s},2}\hdelta_{\cC}, \bbeta_{\cT}) < \cE(\hbeta_{\cT}, \bbeta_{\cT}).
	\end{equation}
	Inequality \eqref{eq: pop dominance} indicates that, although the combination coefficients $(\lambda_{{\rm s}, 1}, \lambda_{{\rm s}, 2})$ may not represent the minimizer of the $\cE$-error, it ensures that the resulting estimator outperforms the target population-based estimator $\hbeta_{\cT}$.	
    \end{remark}

    \begin{remark} Another, perhaps more standard, approach to mitigate the ill-posedness of $\bbB$ is to replace it by $\bbB + r\bbI_{2}$ where $r$ is a small constant. Then, one can define the ridge-SURE estimator $\hbeta_{\SURE}^{(r)} = \widehat{\lambda}_{\SURE,1}^{(r)} \hbeta_{\cC} + \widehat{\lambda}_{\SURE,2}^{(r)}\hdelta_{\cC}$ where $(\widehat{\lambda}_{\SURE,1}^{(r)},		\widehat{\lambda}_{\SURE, 2}^{(r)}) =  (1, 1) - 
	(\tr\{\bbH\bbSig_{\cC}\}, \tr\{\bbH\bbSig_{\delta}\})(\widehat{\bbB} + r\bbI_{2})^{-1}
	$. Simulation results in Supplementary Material Section \ref{app: sim ridge SURE} demonstrate that the performance of $\hbeta_{\SURE}^{(r)}$  is  sensitive to the choice of
    the tuning parameter $r$ which can be difficult to choose in practice. By contrast, the dShrink estimator $\hbeta_{\rm ds}$ is requires no user-specified tuning parameters and exhibits consistently competitive performance across various settings in our numerical experiments.
\end{remark}

    \begin{remark}\label{remark: linear}
      It‘s helpful to consider the linear regression to get some insights. Suppose $\bY_{\cT} = \bbX_{\cT}\bbeta_{\cT} + \be_{\cT}$ and $\bY_{\cS} = \bbX_{\cS}\bbeta_{\cS} + \be_{\cS}$, where $\bbeta_{\cT}$, $\bY_{\cT}$, $\bbX_{\cT}$ are the regression coefficient vector, the outcome vector, and the covariate matrix from the target population,  and  $\bbeta_{\cS}$, $\bY_{\cS}$, $\bbX_{\cS}$ are the corresponding quantities from the source population. Assume for simplicity that all the variables are centered and the errors $\be_{\cT} \sim N(\bzero, \sigma_{\cT}^{2}\bbI_{n_{\cT}})$, $\be_{\cS} \sim N(\bzero, \sigma_{\cS}^{2}\bbI_{n_{\cS}})$ with $\sigma_{\cT}^{2}$ and $\sigma_{\cS}^{2}$ known. Let $\halpha_{\cT} = \hbeta_{\cT}$ and $\halpha_{\cS}=\hbeta_{\cS}$ be the least squares regression coefficients in the target and source data, respectively. Define the pooled estimator $\hbeta_{\rm pool}$ as the minimizer of the pooled least squares loss 
            $\frac{1}{\sigma_{\cT}^{2}}\|\bY_{\cT} - \bbX_{\cT}\bbeta\|^{2} 
                +
                \frac{1}{\sigma_{\cS}^{2}}\|\bY_{\cS} - \bbX_{\cS}\bbeta\|^{2}$,
            which is equivalent to the calibration estimator \eqref{eq:meta}. Then, the SURE estimator and the dShrink estimator can be obtained from the penalized least squares problem
            \begin{equation}\label{eq: linear loss ds}
            	\small
            	\nodisplayskips
            	\frac{1}{\sigma_{\cT}^{2}}\|\bY_{\cT} - \bbX_{\cT}\bbeta\|^{2} + \lambda_{0}\|\bbX_{\cT}\bbeta\|^{2} + \lambda_{\rm pool}\|\bbX_{\cT}(\bbeta - \hbeta_{\rm pool})\|^{2}
            \end{equation}
            with some specific choices of the penalization parameters $\lambda_{0}$ and $\lambda_{\rm pool}$.   The proof is given in Supplementary Material Section S2.5.
           Thus,  the SURE estimator and the dShrink estimator shrink the estimate towards $\bzero$ and $\hbeta_{\rm pool}$ using ridge-type penalties with adaptively determined penalization parameters, which is expected to lead to improved estimation accuracy when $\bbeta_{\cT}$ is close to $\bzero$ or $\hbeta_{\rm pool}$.

           Next, we examine the relationship between the SURE estimator $\hbeta_{\SURE}$, the dShrink estimator $\hbeta_{\rm ds}$ and the least squares estimators $\hbeta_{\cT}$ and $\hbeta_{\cS}$ in the target and source populations. One can easily see that $\lambda_{1}\hbeta_{\cC}+ \lambda_{2}\hdelta_{\cC}=\lambda_{1}\hbeta_{\cT}+ (\lambda_{2} - \lambda_{1})\bbSig_{\beta,\cT}(\bbSig_{\beta,\cT} + \bbSig_{\beta,\cS})^{-1}(\hbeta_{\cT} - \hbeta_{\cS})$. Simple calculations show that the SURE estimator $\hbeta_{\SURE}$ and the dShrink estimator $\hbeta_{\rm ds}$ takes the form
            \begin{equation}
            \small
            \nodisplayskips
             (\lambda_{1} \bbSig_{\beta,\cS} + \lambda_{2}\bbSig_{\beta,\cT})
            (\bbSig_{\beta,\cT} + \bbSig_{\beta,\cS})^{-1}
            \hbeta_{\cT}+(\lambda_1-\lambda_2)
            \bbSig_{\beta,\cT}(\bbSig_{\beta,\cT} + \bbSig_{\beta,\cS})^{-1}
            \hbeta_{\cS},
              \label{eq:ds1} 
            \end{equation}
         with $\Sigma_{\beta,\cS} = \var(\hbeta_{\cS})$, $\lambda_1$ and $\lambda_2$ determined based on the estimated $\cE$-error 
            or the separable surrogate objective function. This implies that the SURE estimator $\hbeta_{\SURE}$ and the drink estimator $\hbeta_{\rm ds}$ 
            are weighted averages of $\hbeta_{\cT}$ and $\hbeta_{\cS}$. However,  the weights are generally not scalar unless $\bbSig_{\beta,\cT}=c\bbSig_{\beta,\cS}$ for some constant $c$.
\end{remark}

	\subsection{Connections with the JS type estimators and data enriched linear regression}\label{subsec: JS and data enrich}

	As discussed in the Introduction, many existing transfer estimation methods are safe only under strong assumptions on the difference between the target and source
	populations.
    In contrast, the dShrink method works in a heterogeneity-agnostic way in the sense that it can incorporate external summary statistics safely with minimal assumptions on the difference between the target and source
    populations.
    The rationale behind this intriguing property is akin to that of the JS estimator and its variants. The JS estimator is well-known for achieving a lower MSE compared to the target population-based maximum likelihood estimator in the absence of any further assumptions other than the normality of the estimator. This improvement is achieved by adaptively shrinking the estimator toward zero or other targets suggested by side information \citep{stein1981estimation, green1991james, han2024improving}. Although motivated differently, the dShrink estimator employs adaptive shrinkage in a similar spirit to
    that of JS type estimators and achieves the dominance result. 
    
    Next, we compare the dShrink estimator with the JS estimator. Existing JS type methods and are largely limited to mean estimation or linear regression, and are subject to restrictive conditions in the form of the summary statistics. \citep{green1991james, han2024improving}. 
    On the contrary, the dShrink estimator is applicable to general estimation problems in the presence of parameter heterogeneity and population heterogeneity.   In addition to its broader applicability, the dShrink estimator offers certain advantages over the James–Stein estimator, even in scenarios where both estimators are applicable.
    
      Specifically, we consider the setting studied in \cite{green1991james} where $\halpha_{\cS} = \hbeta_{\cS} \sim N(\bzero, \sigma_{\cS}^{2}\bbI)$ and $\halpha_{\cT} = \hbeta_{\cT} \sim N(\bzero, \sigma_{\cT}^{2}\bbI)$ are the estimators of the target parameter in the source and target populations respectively. The JS type estimator $\hbeta_{\rm JS}$ proposed in \cite{green1991james} estimates $\bbeta_{\cT}$ using a weighted average $w\hbeta_{\cT} + (1 - w)\hbeta_{\cS}$, 
    where $w$ is determined based on an estimate of the $\cE$-error of the resulting estimator.     
     The weight proposed in \cite{green1991james} is $w=\{1 - (p - 2)\sigma_{\cT}^{2}/(\hbeta_{\cT} - \hbeta_{\cS})^{2}\}_{+}$.
    Note that, although the fixed effect meta-analysis estimator is also a weighted average of $\hbeta_{\cT}$ and $\hbeta_{\cS}$ in this case, it uses the weight $w=\sigma_{\cT}^{-2} / (\sigma_{\cT}^{-2} + \sigma_{\cS}^{-2})$ that is different from the JS estimator. As a linear combination of $\hbeta_{\cC}$ and $\hdelta_{\cC}$ can be equivalently expressed as a linear combination of $\hbeta_{\cT}$ and $\hbeta_{\cS}$ in this scenario,   one can easily see that the dShrink method considers estimators of the form $w_{1}\hbeta_{\cT} + w_{2}\hbeta_{\cS}$ without requiring $w_{1} + w_{2} = 1$ and determines $w_{1}$ and $w_{2}$ by minimizing a separable surrogate objective function \eqref{eq: separable surrogate} that approximates the $\cE$-error. This can be seen using equation (\ref{eq:ds1}) by setting $\bbSig_{\beta, \cT}=\sigma_{\cT}^{2}\bbI$ and $\bbSig_{\beta, \cS}=\sigma_{\cS}^{2}\bbI$, which gives $w_1=(\lambda_{1}\sigma_{\cS}^{2}+\lambda_{2}\sigma_{\cT}^{2})/(\sigma_{\cT}^{2} + \sigma_{\cS}^{2})$ and $w_2=(\lambda_1-\lambda_2)\sigma_{\cT}^{2}/(\sigma_{\cT}^{2} + \sigma_{\cS}^{2})$.
    
    By optimizing over a broader class of estimators, the dShrink method has the potential to achieve a lower estimation error than the JS type estimator. For example, the estimator $w_{1}\hbeta_{\cT} + w_{2}\hbeta_{\cS}$ can outperform $w\hbeta_{\cT} + (1 - w)\hbeta_{\cS}$ when $\bbeta_{\cT} \approx \bzero$ in MSE by taking $w_{1} = w_{2} = 0$. In addition, setting $w_{1} = w$ and $w_{2} = c(1 - w)$ in $w_{1}\hbeta_{\cT} + w_{2}\hbeta_{\cS}$ can yield a lower estimation error than $w\hbeta_{\cT} + (1 - w)\hbeta_{\cS}$ when $\bbeta_{\cT} \approx c\bbeta_{\cS}$ for some $c \neq 0$. The scenario $\bbeta_{\cT} \approx \bzero$ can arise if covariates have weak effects on the outcome in regression analysis, while the scenario $\bbeta_{\cT} \approx c\bbeta_{\cS}$ may occur in transfer estimation problems when, in different data sources, the outcomes are on different scales (e.g., continuous v.s. dichotomous; raw v.s. standardized data) \citep{gu2025robust} or different covariates are adjusted in the regression \citep{taylor2023data}. 
    
    The trade-off for optimizing over a broader estimator class is that determining $w_{1}, w_{2}$ (equivalently $\lambda_{1}$, $\lambda_{2}$ in previous sections) in a data-adaptive manner is more challenging. For the James--Stein type estimator, by minimizing the estimated $\cE$-error of the resulting estimator, one can obtain an estimator that dominates the target population-based estimator $\hbeta_{\cT}$ \citep{green1991james}. 
    However, this does not extend to the estimator of the form $w_{1}\hbeta_{\cT} + w_{2}\hbeta_{\cS}$ or equivalently $\lambda_{1} \hbeta_{\cC} + \lambda_{2}\hdelta_{\cC}$ as demonstrated by the performance of $\hbeta_{\SURE}$ in Section \ref{subsec: SURE}. 
    To address this issue, dShrink adopts a separable surrogate objective function using $\hbeta_{\cC}$ and $\hdelta_{\cC}$, thereby restoring theoretical safeguards against negative transfer.
    
    Under the linear model considered in Remark \ref{remark: linear}, the loss function \eqref{eq: linear loss ds} with $\lambda_{0} = 0$ and $\hbeta_{\rm pool}$ is replaced by $\hbeta_{\cS}$ reduces to that of data enriched linear regression \citep{chen2015data} with a quadratic penalty, whose safety is established when covariance matrices of the covariates are identical in the target and source populations. Compared with data enriched linear regression, dShrink introduces an additional penalty toward zero in the linear case, which improves efficiency when the true effect is small. Moreover, dShrink applies to a broader class of models beyond the linear setting, accommodates differences in variables measured across populations, and maintains safety in general settings. 
    
    We conclude this section with a numerical comparison of JS, the data enriched linear regression, and the dShrink estimators in Figure \ref{fig: eg-JS} under \eqref{eq: counter eg} with the same parameter configuration used in Figure \ref{fig: eg-ds}. Consider the linear regression setup in Remark \ref{remark: linear}, where $\bbeta_{\cT} = \boldsymbol{\mu} + \bbias$, $\bbeta_{\cS} = \boldsymbol{\mu} - \bbias$, $n_{\cT}^{-1}\bbX_{\cT}^{\T}\bbX_{\cT} = n_{\cS}^{-1}\bbX_{\cS}^{\T}\bbX_{\cS} = \bbI_{p}$, and $\sigma_{\cT}^{2} = \sigma_{\cS}^{2} = 2$. Under these settings, the least squares estimators $\hbeta_{\cT}$ and $\hbeta_{\cS}$ satisfy model \eqref{eq: counter eg}, and the data enriched linear regression can be implemented accordingly. We use the bias-adjusted weight in (5.5) of \cite{chen2015data} when implementing the data enriched linear regression.
    The results in Figure \ref{fig: eg-JS} underscore the advantageous performance of the dShrink estimator.
   \begin{figure}
		\centering
		\includegraphics[scale = 0.25]{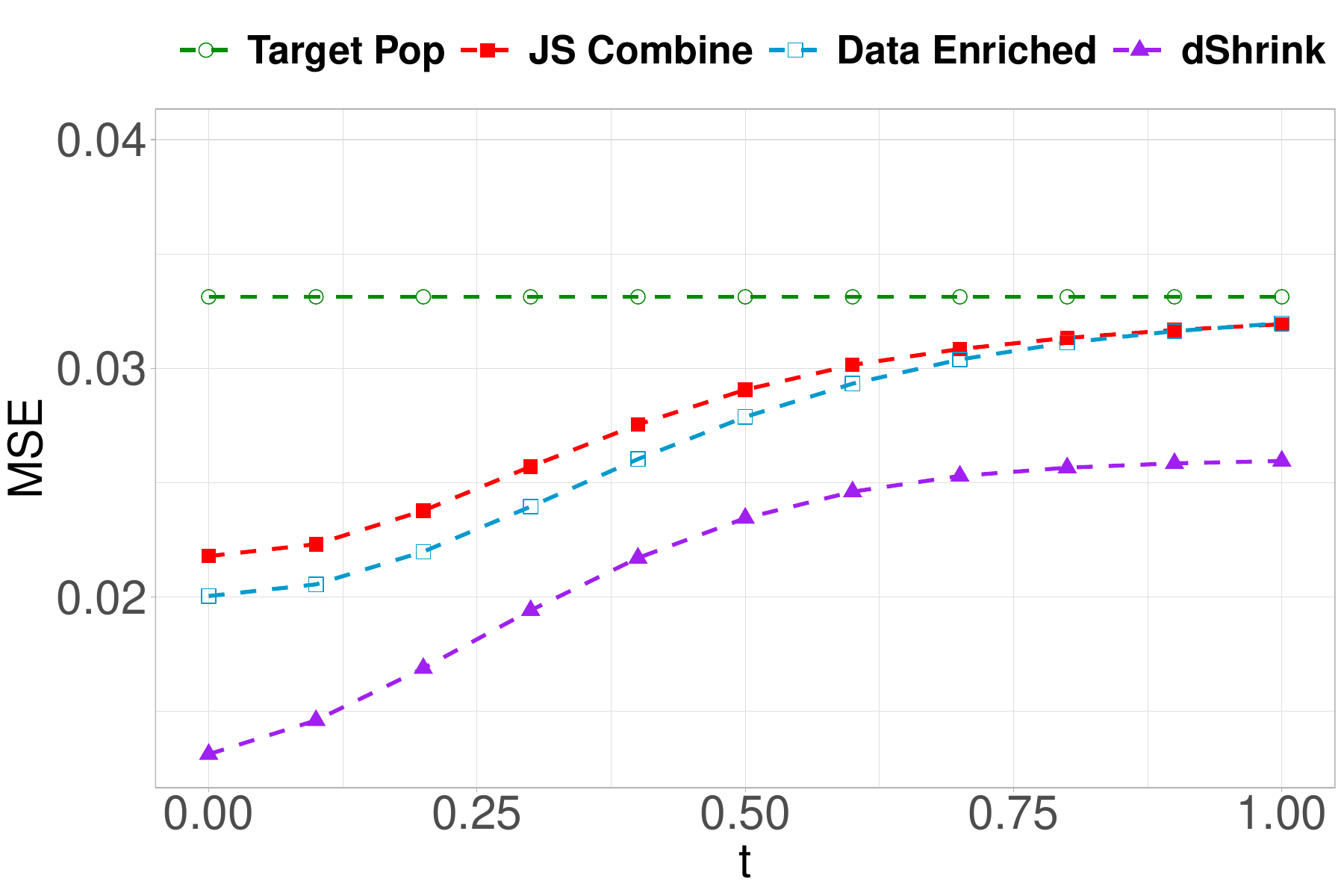}
		\caption{\it MSEs of the estimator using only the data from the target population $\hbeta_{\cT}$, the JS estimator $\hbeta_{\rm JS}$, the data enriched linear regression, and the dShrink estimator $\hbeta_{\rm ds}$  in the simulation under the setting of \eqref{eq: counter eg} with $\bbH = \bbI_{p}$, $p = 5$, $n_{\cT} = n_{\cS} = 300$, $\boldsymbol{\mu} = (0.05, 0.02, 0.1, 0.1, 0.1) ^{\T} / \sqrt{5}$, $\bbias = t\etab$, $\etab = (0.2, 0.3, 0.3, 0.3, 0.3)^{\T} / \sqrt{5}$ and different $t$'s.}\label{fig: eg-JS}
	\end{figure}

    \bigskip
    
    \subsection{An asymptotic version of the theoretical guarantee}\label{subsec: asy normal}  
    Theorem~\ref{thm: dominance ds} establishes the safety of the dShrink estimator under the assumption that $\halpha_{\cS}$, $\halpha_{\cT}$, and $\hbeta_{\cT}$ are normally distributed. In practice, this finite-sample normality assumption may not hold exactly. In this section, we extend the theoretical guarantee for the dShrink estimator to settings in which $\halpha_{\cS}$, $\halpha_{\cT}$, and $\hbeta_{\cT}$ are only asymptotically normal. In the theoretical development, we use $C$ to denote a generic positive constant whose value may vary from line to line. For any matrix $\bbM$, let $\sigma_{\min}(\bbM)$ denote its minimum singular value. We allow vectors and matrices to have infinite entries, and call a vector or matrix finite if all of its components or entries are finite. We study the asymptotic regime where $n_{\cT} \to \infty$ and $n_{\cS}$ goes to infinity as $n_{\cT} \to \infty$.
    \begin{condition}\label{cond: asy normal}
        (i) $\halpha_{\cS}$ is independent of $(\hbeta_{\cT}^{\T}, \halpha_{\cT}^{\T})^{\T}$ and for finite vectors $\balpha_{\cS}$, $\balpha_{\cS}$, and $\bbeta_{\cT}$,
        \[
        \small
        \nodisplayskips
        \sqrt{r_{\cS}}(\halpha_{\cS} - \balpha_{\cS}) \to N(\bzero, \bbV_{\alpha,\cS}),\
	\sqrt{r_{\cT}}
        \begin{pmatrix}
		\hbeta_{\cT} - \bbeta_{\cT}\\
		\halpha_{\cT} - \balpha_{\cT}
	\end{pmatrix}
        \to 
	N\left\{
	\begin{pmatrix}
		\bzero\\
		\bzero
	\end{pmatrix}, 
	\begin{pmatrix}
		\bbV_{\beta, \cT} & \bbV_{\beta\alpha, \cT}\\
		\bbV_{\beta\alpha, \cT}^{\T} & \bbV_{\alpha, \cT}
	\end{pmatrix}
	\right\}
	\]
    in distribution as $n_{\cT} \to \infty$, where $\{r_{\cS}\}$ and $\{r_{\cT}\}$ are positive sequences reflecting the convergence rates of $\halpha_{\cS}$ and $(\hbeta_{\cT}^{\T}, \halpha_{\cT}^{\T})^{\T}$ that depend on $n_{\cS}$ and $n_{\cT}$, and $\bbV_{\alpha,\cS}$, $\bbV_{\alpha, \cT}$, $\bbV_{\beta\alpha, \cT}$, and $\bbV_{\beta, \cT}$ are finite variance and covariance matrices; (ii) $r_{\cS}, r_{\cT} \to \infty$, $r_{\cS} / r_{\cT} \to \rho_{\cS, \cT}$, $C^{-1} \leq \sigma_{\rm min}(\bbV_{\beta, \cT}) \leq \sigma_{\rm max}(\bbV_{\beta, \cT}) \leq C$ and $C^{-1} \leq \sigma_{\rm max}(\bbV_{\alpha, \cS}) \leq \sigma_{\rm max}(\bbV_{\alpha, \cS}) \leq C$, $\|\bbV_{\beta\alpha, \cT}\| > 0$; (iii) $E\{(\sqrt{r_{\tau}}\|\hbeta_{\cT} - \bbeta_{\cT}\|)^{2 + \tau}\} \leq C$  for some positive constants $\rho_{\cS, \cT}$ and $C$. 
    \end{condition}
    Under Condition \ref{cond: asy normal}, one can define $\hbeta_{\cC}$, $\bbeta_{\cC}$, $\hdelta_{\cC}$, $\bdelta_{\cC}$, $c_{\beta}$, $c_{\delta}$, and $\hbeta_{\rm ds}$ similarly to the previous sections with $\bbSig_{\alpha, \cS} = r_{\cS}^{-1}\bbV_{\alpha, \cS}$, $\bbSig_{\beta, \cT} = r_{\cT}^{-1}\bbV_{\beta, \cT}$, $\bbSig_{\beta, \cT} = r_{\cT}^{-1}\bbV_{\beta\alpha, \cT}$, and $\bbSig_{\alpha, \cT} = r_{\cT}^{-1}\bbSig_{\alpha, \cT}$. We have the following result.
    \begin{theorem}\label{thm: dominance ds asy}
        Suppose Condition \ref{cond: asy normal} holds, $\bbH$ is positive definite, $\lim_{n_{\cT}\to \infty}r_{\cT}c_{\beta} > 0$, $\lim_{n_{\cT}\to \infty}r_{\cT}c_{\delta} > 0$, $\lim_{n_{\cT}\to \infty}\sqrt{r_{\cT}}\bbeta_{\cC} = \bxi_{\cC}$, and $\lim_{n_{\cT}\to \infty}\sqrt{r_{\cT}}\bdelta_{\cC} = \bxi_{\delta}$, where $\bxi_{\cC}$ and $\bxi_{\delta}$ are vectors with possibly infinite components. Then, we have
        $
        \lim_{n_{\cT}\to \infty} \cE(\hbeta_{\rm ds},  \bbeta_{\cT})/\cE(\hbeta_{\cT}, \bbeta_{\cT}) \leq 1.
        $
        Further assuming that at least one of $\bxi_{\cC}$ and $\bdelta_{\cC}$ is finite, we have 
       $\lim_{n_{\cT}\to \infty} \cE(\hbeta_{\rm ds},  \bbeta_{\cT})/\cE(\hbeta_{\cT}, \bbeta_{\cT}) < 1$.
    \end{theorem}
    
    \subsection{Estimating unknown matrices}\label{subsec: est var}
	The estimator $\hbeta_{\rm ds}$ cannot be directly calculated in practice because it involves possibly unknown matrices $\bbH$, $\bbQ_{\alpha}$, $\bbSig_{\cC}$, and $\bbSig_{\delta}$. Their estimators $\hH$, $\hQ_{\alpha}$, $\hSig_{\cC}$, and $\hSig_{\delta}$ can be obtained in various contexts, including the linear regression, logistic regression, and general M-estimation.
	Let $\tbeta_{\cC} = \hbeta_{\cT} - \hQ_{\alpha}(\halpha_{\cT} - \halpha_{\cS})$, $\tdelta_{\cC} = \hQ_{\alpha}(\halpha_{\cT} - \halpha_{\cS})$, $\widetilde{\lambda}_{{\rm s}, 1} = \left\{1 -  \hat{c}_{\beta}/(\hbeta_{\cC}^{\T}\hH\hbeta_{\cC})\right\}_{+}$, and 
	$\widetilde{\lambda}_{{\rm s}, 2} = \left\{1 - \hat{c}_{\delta}/(\hdelta_{\cC}^{\T}\hH\hdelta_{\cC})\right\}_{+}$
	where $\hat{c}_{\beta}  = (\tr\{\hH\hSig_{\cC}\} - 2\sigma_{\rm max}(\hH\hSig_{\cC}))_{+}$ and $\hat{c}_{\delta} = (\tr\{\hH\hSig_{\delta}\} - 2\sigma_{\rm max}(\hH\hSig_{\delta}))_{+}$.   In practice, we propose to use the estimator 
	$  \tbeta_{\rm ds} = \widetilde{\lambda}_{{\rm s}, 1}\tbeta_{\cC} + \widetilde{\lambda}_{{\rm s}, 2}\tdelta_{\cC}$
	which can be efficiently calculated using summary statistics $\hbeta_{\cT}$, $\halpha_{\cT}$, $\halpha_{\cS}$, $\hH$, $\hQ_{\alpha}$, $\hSig_{\cC}$, and $\hSig_{\delta}$.
     Next, we delve into the influence of estimation errors introduced by substituting the estimators of the unknown matrices. 
     We consider the asymptotic regime in Section \ref{subsec: asy normal} and study the property of $\tbeta_{\rm ds}$ under the following condition.
     
     \begin{condition}\label{cond: matrices error}
     	(i) $\|\hH - \bbH\| = o_{P}(\|\bbH\|)$, $\|\hQ_{\alpha} - \bbQ_{\alpha}\| = o_{P}(\|\bbQ_{\alpha}\|)$, $\|\hSig_{\cC} - \bbSig_{\cC}\| = o_{P}(\|\bbSig_{\cC}\|)$, and $\|\hSig_{\delta} - \bbSig_{\delta}\| = o_{P}(\|\bbSig_{\delta}\|)$; (ii) $E\left\{(r_{\cT}\hat{c}_{\beta})^{1 + \tau/2}\right\} \leq C$ and $E\left\{(r_{\cT}\hat{c}_{\delta})^{1 + \tau/2}\right\} \leq C$ for some positive constant $C$.
     \end{condition}
     Next, we establish the theoretical properties of $\tbeta_{\rm ds}$ in the following theorem, whose proof can be found in Supplementary Material Section \ref{app: estimate matrix}.
     \begin{theorem}\label{thm: error unknown matrix}
     	Suppose $\bbH$ is positive definite. Assume $\lim_{n_{\cT}\to \infty}r_{\cT}c_{\beta} > 0$, $\lim_{n_{\cT}\to \infty}r_{\cT}c_{\delta} > 0$, $\lim_{n_{\cT}\to \infty}\sqrt{r_{\cT}}\bbeta_{\cC} = \bxi_{\cC}$, and $\lim_{n_{\cT}\to \infty}\sqrt{r_{\cT}}\bdelta_{\cC} = \bxi_{\delta}$, where $\bxi_{\cC}$ and $\bxi_{\delta}$ are vectors with possibly infinite components. Then $\|\tbeta_{\rm ds} - \hbeta_{\rm ds}\| = o_{P}(\sqrt{r_{\cT}})$ under Conditions \ref{cond: asy normal}(i)(ii), and \ref{cond: matrices error}(i).
        In addition, we have $
        \lim_{n_{\cT}\to \infty} \cE(\tbeta_{\rm ds},  \bbeta_{\cT})/\cE(\hbeta_{\cT}, \bbeta_{\cT}) \leq 1$ under Conditions \ref{cond: asy normal} and \ref{cond: matrices error}.
        Furthermore, if at least one of $\bxi_{\cC}$ and $\bdelta_{\cC}$ is finite, we have 
        $\lim_{n_{\cT}\to \infty} \cE(\tbeta_{\rm ds},  \bbeta_{\cT})/\cE(\hbeta_{\cT}, \bbeta_{\cT}) < 1$.
     \end{theorem}
	
	\section{Extensions}\label{sec: extension}	
	\subsection{Conditional dShrink estimator}\label{subsec: cds}
	This section makes some extensions for the dShrink estimator. We focus on the normal setting with known covariance matrices as in Section \ref{subsec: ds} for clarity. Extensions to the asymptotic regime with estimated covariance matrices can be developed analogously to those in Sections \ref{subsec: asy normal} and \ref{subsec: est var}. An extension of the dShrink estimator to incorporate side information can be found in Supplementary Material Section \ref{app: mds}.
	
	The dShrink estimator requires the covariance matrix $\bbSig_{\alpha,\cS}$ of the external summary statistics, which is not always available.
	To avoid this requirement, one can construct the dShrink estimator conditional on $\halpha_{\cS}$. Let $\bbQ_{\alpha}^{\dag} = \bbSig_{\beta\alpha, \cT}\bbSig_{\alpha, \cT}^{-1}$, $\hbeta_{\cC}^{\dag} = \hbeta_{\cT} - \bbQ_{\alpha}^{\dag}(\halpha_{\cT} - \halpha_{\cS})$, and
	$\hdelta_{\cC}^{\dag} = \bbQ_{\alpha}^{\dag}(\halpha_{\cS} - \halpha_{\cT})$. 
	Conditional on $\halpha_{\cS}$, $\hbeta_{\cC}^{\dag}$ is independent of $\hdelta_{\cC}^{\dag}$. 
	Then, the same procedure that derives the dShrink estimator gives rise to the conditional dShrink (c-dShrink) estimator
	$\hbeta_{\rm ds}^{\dag} = \widehat{\lambda}_{{\rm s}, 1}^{\dag}\hbeta_{\cC}^{\dag} + \widehat{\lambda}_{{\rm s}, 2}^{\dag}\hdelta_{\cC}^{\dag}$,
	where 
	$
	\widehat{\lambda}_{{\rm s}, 1}^{\dag} = \left\{1 -  c_{\beta}^{\dag}/(\hbeta_{\cC}^{\dag\T}\bbH\hbeta_{\cC}^{\dag})\right\}_{+}$, $\widehat{\lambda}_{{\rm s}, 2}^{\dag} = \left\{1 - c_{\delta}^{\dag}/(\hdelta_{\cC}^{\dag\T}\bbH\hdelta_{\cC}^{\dag})\right\}_{+}$,	$c_{\beta}^{\dag} = (\tr\{\bbH\bbQ_{\alpha}^{\dag}\bbSig_{\alpha, \cT}\bbQ_{\alpha}^{\dag\T}\} - 2\sigma_{\rm max}(\bbH\bbQ_{\alpha}^{\dag}\bbSig_{\alpha, \cT}\bbQ_{\alpha}^{\dag\T}))_{+}$,  $c_{\delta}^{\dag} = (\tr\{\bbH\bbSig_{\alpha, \cT}\} - 2\sigma_{\rm max}(\bbH\bbSig_{\alpha, \cT}))_{+}$.
	Theorem \ref{thm: dominance cds} establishes the dominance result for $\hbeta_{\rm ds}^{\dag}$. The proof of Theorem 3 is given in Supplemental Material Section \ref{app: proof cds}.
	\begin{theorem}\label{thm: dominance cds}
		Suppose $\halpha_{\cT}$ and $\hbeta_{\cT}$ are normally distributed as in \eqref{eq: normality}. Then, $\cE(\hbeta_{\rm ds}^{\dag},\bbeta_{\cT}) \leq \cE(\hbeta_{\cT}, \bbeta_{\cT})$ for any $\balpha_{\cS}$, $\balpha_{\cT}$, and  $\bbeta_{\cT}$. If $c_{\beta}^{\dag} > 0$ or $c_{\delta}^{\dag} > 0$, then $\cE(\hbeta_{\rm ds}^{\dag},\bbeta_{\cT}) < \cE(\hbeta_{\cT}, \bbeta_{\cT})$ for any $\balpha_{\cS}$, $\balpha_{\cT}$, and  $\bbeta_{\cT}$.
	\end{theorem}
	 Although $\hbeta_{\rm ds}^{\dag}$ dominates $\hbeta_{\cT}$ in terms of $\cE$-error, our numerical results in Supplementary Material Section~\ref{app: sim m&c} show that $\hbeta_{\rm ds}^{\dag}$ can have a slightly larger $\cE$-error than $\hbeta_{\rm ds}$ due to ignoring information regarding the uncertainty of $\halpha_{\cS}$.
	 An analytical comparison between the $\cE$-errors of $\hbeta_{\rm ds}^{\dag}$ and $\hbeta_{\rm ds}$ is challenging.
	 Under the setting of example \eqref{eq: counter eg} with $p > 2$ and $\boldsymbol{\mu} = \bbias = \bzero$, we show in Supplementary Material Section~\ref{app: compr ds&cds} that
	 \begin{equation}\label{eq: compr ds&cds}
	 	\small
	 	\nodisplayskips
	 	\begin{aligned}
	 		\cE(\hbeta_{\rm ds}^{\dag}, \bbeta_{\cT}) - \cE(\hbeta_{\rm ds}, \bbeta_{\cT})
	 		& = E\left(1\left\{(p - 2)v_{\delta}^{2}\|\hdelta_{\cC}\|^{-2} > 1\right\}\|\hdelta_{\cC}\|^{2}\left[\min\{\pi_{\cS}^{-1}, (p - 2)v_{\delta}^{2}\|\hdelta_{\cC}\|^{-2}\} - 1\right]^{2}\right)\\
	 		&\quad + (\pi_{\cS}^{-1} - 1)E\left[1\left\{(p - 2)v_{\delta}^{2}\|\hdelta_{\cC}\|^{-2}\leq 1\right\}(p - 2)v_{\delta}^{2}\{2 - (p - 2)v_{\delta}^{2}\|\hdelta_{\cC}\|^{-2}\}\right]\\
	 		&\quad
	 		+ (\pi_{\cS}^{-1} - 1)E\left[1\left\{(p - 2)v_{\delta}^{2}\|\hdelta_{\cC}\|^{-2} > 1\right\}\|\hdelta_{\cC}\|^{2}\right] > 0,
	 	\end{aligned}
	 \end{equation}
	 where $\pi_{\cS} = n_{\cS} / (n_{\cT} + n_{\cS})$ is the proportion of the sample from the source population. Notably, $\cE(\hbeta_{\rm ds}^{\dag}, \bbeta_{\cT}) - \cE(\hbeta_{\rm ds}, \bbeta_{\cT}) \approx 0$ when $\pi_{\cS}$ is close to one.
	 This is intuitive, as the uncertainty of $\halpha_{\cS}$ becomes negligible compared to that of $\hbeta_{\cT}$ when $\pi_{\cS} \approx 1$, so ignoring it does not substantially inflate the estimation error.

\bigskip

\subsection{Multiple source populations}\label{subsec: multi-source}
The proposed methods can be naturally extended to integrate summary statistics from multiple source populations. For $k = 1, \dots, K$, let $\halpha_{\cS,k}$ denote a $q_{k}$-dimensional summary statistic based on an independent sample of size $n_{\cS,k}$ from the $k$th source population, which estimates a functional $\balpha_{\cS,k} = \balpha_{k}(P_{\cS,k})$ of the corresponding population distribution $P_{\cS,k}$.
With some abuse of notation, let $\halpha_{\cT}$ denote an estimator for $\balpha_{\cT} = (\balpha_{1}(P_{\cT})^{\T}, \dots, \balpha_{K}(P_{\cT})^{\T})^{\T}$ based on the target population data, and let $q = \sum_{k=1}^{K} q_{k}$.
Assume $\halpha_{\cT} \sim N(\balpha_{\cT}, \bbSig_{\alpha,\cT})$ and $\halpha_{\cS,k} \sim N(\balpha_{\cS,k}, \bbSig_{\alpha,\cS,k})$ for $k = 1, \dots, K$.
Define $\halpha_{\cS} = (\halpha_{\cS,1}^{\T}, \dots, \halpha_{\cS,K}^{\T})^{\T}$ and $\bbSig_{\alpha,\cS} = \diag\{\bbSig_{\alpha,\cS,1}, \dots, \bbSig_{\alpha,\cS,K}\}$.
Then, the dShrink estimator can be defined analogously to \eqref{eq: dShrink}, and a similar theoretical guarantee can be established.

When $\{\halpha_{\cS,k}\}_{k=1}^{K}$ estimate the same functional, i.e., $\balpha_{1}(\cdot) = \dots = \balpha_{K}(\cdot)$ (for example, marginal regression coefficients), $\halpha_{\cT}$ can be constructed as $\halpha_{\cT} = \boldsymbol{1}_{K} \otimes \halpha_{\cT, 1}$, where $\halpha_{\cT,1}$ estimates $\balpha_{1}(\cdot)$, $\boldsymbol{1}_{K}$ is a $K$-dimensional vector of ones, and $\otimes$ denotes the Kronecker product.
In this case, the variance–covariance matrix $\bbSig_{\alpha,\cT}$ of $\halpha_{\cT}$ is singular; however, this does not pose any issue in defining the dShrink estimator.
After some algebra, one can show that the dShrink estimator based on $K$ source populations is equivalent to the dShrink estimator that employs the inverse-variance–weighted meta-analysis estimator $\halpha_{\cS, \IVW} = (\sum_{k = 1}^{K}\bbSig_{\alpha, \cS, k}^{-1})\sum_{k = 1}^{K}\bbSig_{\alpha, \cS, k}^{-1}\halpha_{\cS, k}$ based on $K$ source populations as external summary statistics.

The potential singularity of $\bbSig_{\alpha,\cT}$ is more consequential in defining the c-dShrink estimator, since $\bbSig_{\alpha,\cT}^{-1}$ in $\bbQ_{\alpha}^{\dag} = \bbSig_{\beta\alpha,\cT} \bbSig_{\alpha,\cT}^{-1}$ is undefined when $\bbSig_{\alpha,\cT}$ is singular.
Let $\bbSig_{\alpha,\cT} = \bbU_{\cT} \mathbf{\Lambda}{\cT} \bbU{\cT}^{\T}$ denote the eigendecomposition of $\bbSig_{\alpha,\cT}$, where $\bbU_{\cT}$ is a $q \times d_{\cT}$ column-orthogonal matrix, $\mathbf{\Lambda}{\cT}$ is a $d{\cT} \times d_{\cT}$ invertible diagonal matrix, and $d_{\cT}$ is the rank of $\bbSig_{\alpha,\cT}$.
Then, the c-dShrink estimator can be defined by replacing $\bbSig_{\alpha,\cT}^{-1}$ in $\bbQ_{\alpha}^{\dag}$ with $\bbU_{\cT}\mathbf{\Lambda}{\cT}^{-1}\bbU{\cT}^{\T}$.
Equivalently, this corresponds to defining the c-dShrink estimator based on $\bbU_{\cT}^{\T}\halpha_{\cT}$ and $\bbU_{\cT}^{\T}\halpha_{\cS}$ instead of $\halpha_{\cT}$ and $\halpha_{\cS}$.
Under this reparameterization, the result in Theorem~\ref{thm: dominance cds} continues to hold.
    \bigskip
    
	\subsection{Confidence intervals based on the dShrink estimator}\label{subsec: reboot}
	In this section, we construct the confidence interval (CI) of $\Psi(\bbeta_{\cT})$ based on $\hbeta_{\rm ds}$, where $\Psi$ is a smooth function. An important example is $\Psi(\bbeta_{\cT}) = \beta_{\cT, j}$ where $\beta_{\cT, j}$ is the $j$-th component of $\bbeta_{\cT}$ for some $j = 1,\dots, p$. We consider CIs based on $\Psi(\hbeta_{\rm ds})$. The main challenge is that $\hbeta_{\rm ds}$ can be a biased estimator for $\bbeta_{\cT}$.
	
	Let $q(\tau; \bbeta_{\cC}, \bdelta_{\cC})$ be the $\tau$-quantile of the distribution of $\Psi(\hbeta_{\rm ds}) - \Psi(\bbeta_{\cT}) = \Psi(\hbeta_{\rm ds}) - \Psi(\bbeta_{\cC} + \bdelta_{\cC})$ for any $0 < \tau < 1$. Then, $[\Psi(\hbeta_{\rm ds}) - q(1 - \tau/2; \bbeta_{\cT}, \bdelta_{\cC}), \Psi(\hbeta_{\rm ds}) - q(\tau/2; \bbeta_{\cT}, \bdelta_{\cC})]$ is a valid $\tau$-CI for $\Psi(\bbeta_{\cT})$. Note that $q(\tau; \bbeta_{\cC}, \bdelta_{\cC})$ can be viewed as a function of $\tau$, $\bbeta_{\cC}$, and $\bdelta_{\cC}$. For any given $\tau^{*}$, $\bbeta^{*}$ and $\bdelta^{*}$, we can approximate $q(\tau^{*}; \bbeta^{*}, \bdelta^{*})$ using the following bootstrap procedure.
	\begin{enumerate}[topsep=0pt,itemsep=-1ex,partopsep=1ex,parsep=1ex]\label{num: bootstrap}
		\item For $b^{*} = 1,\dots, B^{*}$, independently generate
		$\hbeta_{\cC}^{(b^{*})} \sim N(\bbeta^{*}, \bbSig_{\cC}) \ \text{and} \ \hdelta_{\cC}^{(b^{*})} \sim
		N(\bdelta^{*}, \bbSig_{\delta})$;
		\item Compute the $\tau^{*}$-quantile $\widehat{q}(\tau^{*}; \bbeta^{*}, \bdelta^{*})$ of 
		$\left\{\Psi(\hbeta_{\rm ds}^{(b^{*})}) - \Psi(\bbeta^{*} + \bdelta^{*})\right\}_{b^{*} = 1}^{B^{*}}$.
	\end{enumerate}

    For given $\tau^{*}$, $\bbeta^{*}$ and $\bdelta^{*}$, the difference between $\widehat{q}(\tau^{*}; \bbeta^{*}, \bdelta^{*})$ and $q(\tau^{*}; \bbeta^{*}, \bdelta^{*})$ can be neglected if $B^{*}$ is sufficiently large.
    However, $q(\tau^{*}; \bbeta_{\cC}, \bdelta_{\cC})$ depends on the unknown quantities $\bbeta_{\cC}$ and $\bdelta_{\cC}$.
    Although $\sqrt{n_{\cT}}$-consistent estimators $\hbeta_{\cC}$ and $\hdelta_{\cC}$ are available for $\bbeta_{\cC}$ and $\bdelta_{\cC}$, respectively, the estimation error still can have a non-negligible impact on the resulting quantile when plugging in these estimates. To improve the robustness against the error for estimating $\bbeta_{\cC}$ and $\bdelta_{\cC}$, we calculate the quantiles under multiple reasonable candidate parameters and use the most conservative quantile to construct the CI. We impose working priors $N(\bzero, \sigma_{\beta}^{2}\bbI_{p})$ and $N(\bzero, \sigma_{\delta}^{2}\bbQ_{\alpha}\bbQ_{\alpha}^{\T})$ on $\bbeta_{\cC}$ and $\bdelta_{\cC}$, estimate $\sigma_{\beta}^{2}$ and $\sigma_{\delta}^{2}$ by their maximum likelihood estimators $\hat{\sigma}_{\beta}^{2}$ and $\hat{\sigma}_{\delta}^{2}$, and sample from the estimated posterior distributions to obtain candidate parameters. Here, the prior on $\bdelta_{\cC}$ is induced by a normal prior $N(\bzero, \sigma_{\delta}^{2}\bbI_{q})$ on $\balpha_{\cT} - \balpha_{\cS}$. The resulting re-bootstrap procedure is summarized as follows.
    \begin{enumerate}[topsep=0pt,itemsep=-1ex,partopsep=1ex,parsep=1ex]
    	\item For $b = 1,\dots, B$, independently generate
    	\[\hbeta_{\cC}^{(b)} \sim N\left\{\hat{\sigma}_{\beta}^{2}(\bbSig_{\cC} + \hat{\sigma}_{\beta}^{2}\bbI_{p})^{-1}\hbeta_{\cC}, \hat{\sigma}_{\beta}^{2}\bbI_{p} - \hat{\sigma}_{\beta}^{4}(\bbSig_{\cC} + \hat{\sigma}_{\beta}^{2}\bbI_{p})^{-1}\right\}
    	\] 
        and
        $\hdelta_{\cC}^{(b)} = \bbQ_{\alpha}\zeta^{(b)}$
    	with
    	\[
    	\zeta^{(b)} \sim
    	N\left\{\hat{\sigma}_{\delta}^{2}(\bbSig_{\alpha, \cT} + \bbSig_{\alpha, \cS} + \hat{\sigma}_{\delta}^{2}\bbI_{q})^{-1}(\halpha_{\cT} - \halpha_{\cS}), 
    	\hat{\sigma}_{\delta}^{2}\bbI_{q} - \hat{\sigma}_{\delta}^{4}(\bbSig_{\alpha, \cT} + \bbSig_{\alpha, \cS} + \hat{\sigma}_{\delta}^{2}\bbI_{q})^{-1}\right\};
    	\]
    	\item Compute the quantiles $\widehat{q}(\tau / 2; \hbeta_{\cC}^{(b)}, \hdelta_{\cC}^{(b)})$  and $\widehat{q}(1 - \tau / 2;  \hbeta_{\cC}^{(b)}, \hdelta_{\cC}^{(b)})$ using the bootstrap procedure introduced in the last paragraph; 
    	\item Construct the $(1 - \tau)$-CI for $\Psi(\bbeta_{\cT})$ as 
    	\begin{equation}\label{eq: CI ds}
    		\small
    		\nodisplayskips
    		\left[\Psi(\hbeta_{\rm ds}) - \max_{b = 1,\dots, B}\widehat{q}(1 - \tau / 2;  \hbeta_{\cC}^{(b)}, \hdelta_{\cC}^{(b)}), \Psi(\hbeta_{\rm ds}) - \min_{b = 1,\dots, B}\widehat{q}(\tau / 2;  \hbeta_{\cC}^{(b)}, \hdelta_{\cC}^{(b)})\right].
    	\end{equation}
    	    \end{enumerate}
    Confidence intervals based on the c-dShrink estimator can be defined similarly.
    
	The above procedure constructs the CI using the most conservative quantile derived from $B$ candidate parameters. The quantile functions $q(\tau/2; \bbeta^{*}, \bdelta^{*})$ and $q(1 - \tau/2; \bbeta^{*}, \bdelta^{*})$ are continuous functions of $(\bbeta^{*}, \bdelta^{*})$ due to the continuity of the normal distribution and the dShrink estimator with respect to observed data. If $B$ is large, with high probability, there are some candidate parameters that are very close to the true parameters $(\bbeta_{\cC}, \bdelta_{\cC})$ \citep{guo2023causal,guo2025robust} and the resulting quantiles are very close to the true quantiles. Thus, the coverage rate can be ensured by adopting the most conservative quantile given by different candidate parameters. 
	On the other hand, the candidate parameters typically cluster around $(\hbeta_{\cC}, \hdelta_{\cC})$, as they are from the working posterior distributions. Consequently, the quantiles calculated from these parameters are likely to be similar, which prevents the CI from being overly conservative.  We take $B = 10$ in our implementation. The idea to take the most conservative criterion induced by multiple randomly generated candidate parameters has been adopted in \cite{guo2023causal, guo2025robust} to tackle other non-standard inference problems. We leverage the idea here and combine it with a bootstrap procedure to make inferences based on the dShrink estimator. Numerical results confirm the efficacy of this approach in transfer estimation.

	\section{Simulation Studies}\label{sec: simulation}
	\subsection{Non-real data based simulation study}\label{subsec: sim}
	In this section, we evaluate the performance of the estimators discussed in this paper through a numerical example. 
	Consider a linear regression problem. Suppose the observations $\bZ = (Y, \bX)$ in the target population follow
	\[
	\small
	\nodisplayskips
	Y = \beta_{0,\cT} + \bX^{\T}\bbeta_{X, \cT} + \epsilon,
	\]
	under the distribution in the target population $P_{\cT}$,
	where  $\epsilon \sim N(0, 1)$, $\bX \sim N(0, \bbSig_{\cT})$ is a $p$-dimensional covariate vector, and $\bbSig_{\cT}$ is the covariance matrix of $\bX$ under $P_{\cT}$. The target parameter $\bbeta_{\cT} = (\beta_{0,\cT}, \bbeta_{X, \cT}^{\T})^{\T}$ is the regression coefficient vector in $P_{\cT}$. In the source population, under $P_{\cS}$, assume 
	\[
	\small
	\nodisplayskips
	Y = \beta_{0,\cS} + \bX^{\T}\bbeta_{X, \cS} + \epsilon,
	\]
	with  $\epsilon \sim N(0, 1)$ and $\bX \sim N(0, \bbSig_{\cS})$, where $\bbSig_{\cS}$ is the covariance matrix under $P_{\cS}$.
	Let $\bX_{1}$ be the sub-vector of $\bX$ consists of the first $q_{X} = 2p_{X}/3$ components. Suppose $\halpha_{\cS}$ is the regression coefficient of $Y$ on $(1, \bX_{1}^{\T})^{\T}$ based on a sample from the $P_{\cS}$. Then, $p = p_{X} + 1$ and $q = q_{X} + 1$.
	
	The components of $\bbeta_{\cT}$ were generated independently from $N(0, 0.5/p)$. Let $\etab$ be a $p$-dimensional vector whose components were generated independently from ${\rm Bernoulli}(0.5) \times N(\sqrt{1 / p}, 0.5 / p)$ and $\bbeta_{\cS} = \bbeta_{\cT} + t\etab$, where $\bbeta_{\cS} = (\beta_{0,\cS}, \bbeta_{X, \cS}^{\T})^{\T}$ and $0\leq t\leq 1$ is a parameter that characterizes the difference between $P_{\cT}$ and $P_{\cS}$. For $i,j = 1,\dots, p$, let the $i,j$th element of $\bbSig_{\cT}$ be $0.4^{|i - j|}s_{\cT,i}s_{\cT,j}$ where $s_{\cT ,i}^{2}, s_{\cT,j}^{2}$ were independently generated from  ${\rm Gamma}(2, 0.5)$. Suppose $\bbSig_{\cS} = (1 - t)\bbSig_{\cT} + t \bbSig_{X}$, where $\bbSig_{X}$ is a matrix whose $i,j$th element equals to $0.2^{|i - j|}s_{X,i}s_{X,j}$ where $s_{X ,i}^{2}, s_{X,j}^{2}$ were independent generated from ${\rm Gamma}(2, 0.5)$. For any $\bbeta$, the MSE of the regression function is $E_{\cT}[\{(1 \ \bX^{\T})\bbeta - (1 \ \bX^{\T})\bbeta_{\cT}\}^{2}] = (\bbeta - \bbeta_{\cT})^{\T}E[(1 \ \bX^{\T})^{\T}(1 \ \bX^{\T})](\bbeta - \bbeta_{\cT})$. This motives us to take $\bbH$ as an estimate of $E[(1 \ \bX^{\T})^{\T}(1 \ \bX^{\T})]$ in this example. The target population-based estimator $\hbeta_{\cT}$ is the least squares estimator. The expected quadratic error is approximated by the mean of the quadratic error across $5000$ simulation runs. Figure \ref{fig: sim error} shows the errors of different estimators under different settings.  We varied the number of variables $p_{x}$ and $q_{x}$ and the sample sizes $n_{\cT}$ and $n_{\cS}$ in Figure \ref{fig: sim error}.
	
	\begin{figure}[h]
		\centering
		\subfigure[]{
			\includegraphics[scale = 0.25]{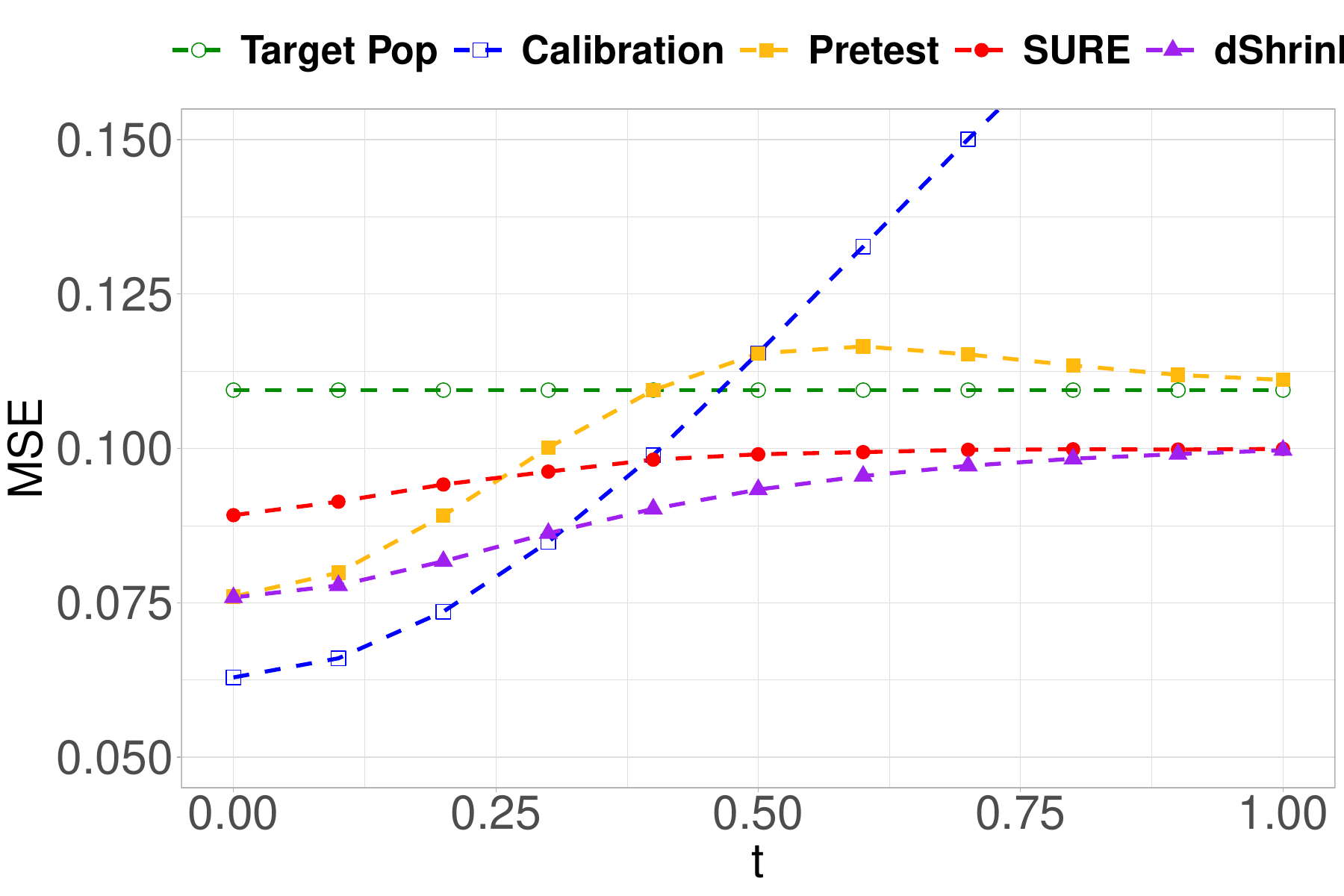}
		}
		\subfigure[]{
			\includegraphics[scale = 0.25]{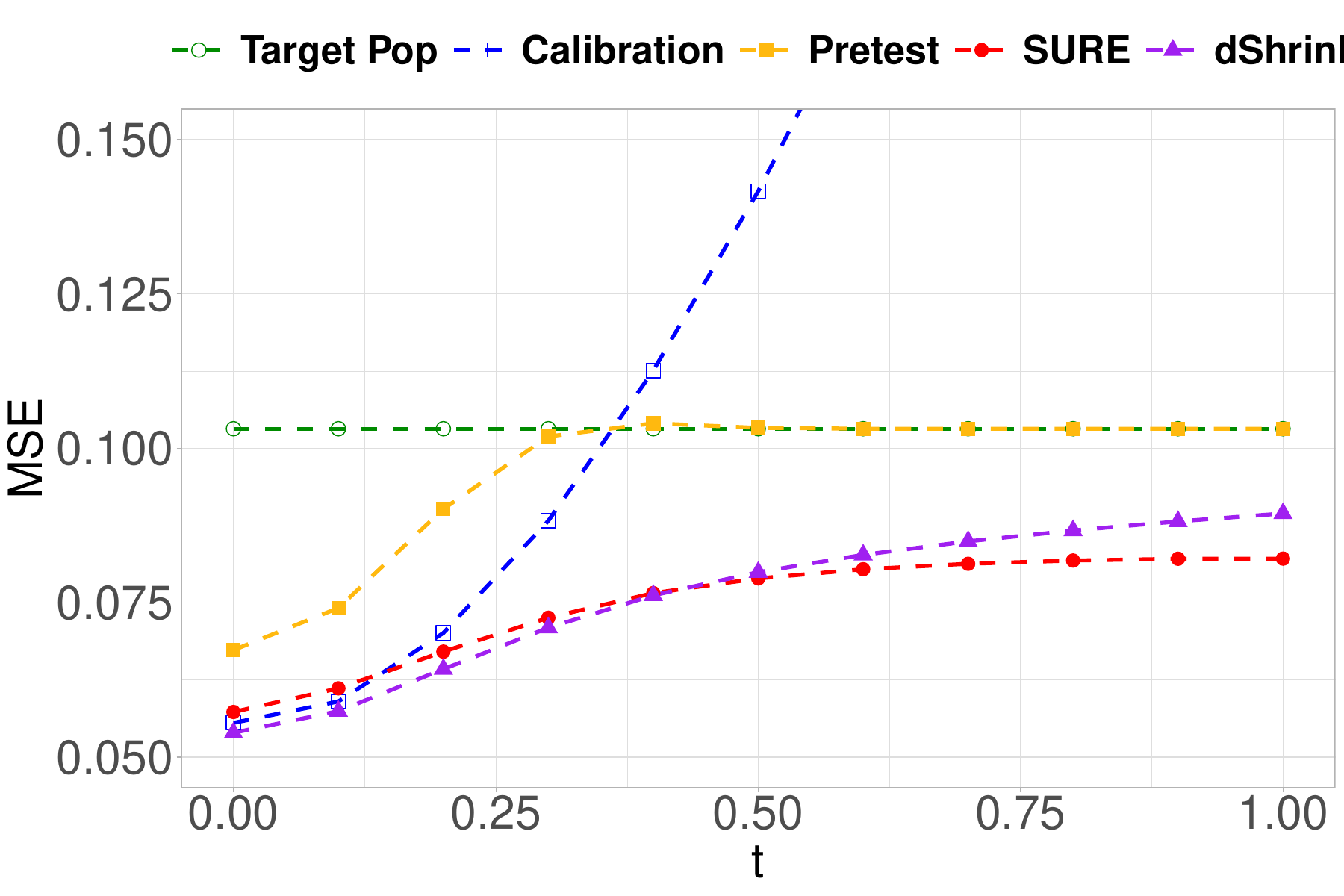}
		}
		\caption{\it MSEs of the regression functions derived from $\hbeta_{\cT}$, $\hbeta_{\cC}$, $\hbeta_{\rm prt}$, $\hbeta_{\SURE}$, and $\hbeta_{\rm ds}$ in simulation studies. (a) $p_{X} = 10$, $q_{X} = 6$, $n_{\cT} = 100$, and $n_{\cS} = 500$; (b)  $p_{X} = 30$, $q_{X} = 20$, $n_{\cT} = 300$, and $n_{\cS} = 1000$.}\label{fig: sim error}
	\end{figure}
	
	 The calibration estimator $\hbeta_{\cC}$ performs well when $t$ is small, i.e., when target and source populations are approximately homogeneous. Its error grows rapidly as $t$ increases, i.e., as the populations become more heterogeneous. The pretest estimator $\hbeta_{\rm prt}$ achieves a smaller error than $\hbeta_{\cC}$ when $t$ is large but has a larger error than $\hbeta_{\cC}$ when $t$ is small. Notably, $\hbeta_{\rm prt}$ has a larger error than $\hbeta_{\cT}$, which confirms the intuitions discussed before Proposition \ref{prop: pretest}. The SURE estimator performs well in the simulation setting considered here. It has a smaller error than the target population-based estimator $\hbeta_{\cT}$ across all $t$'s. However, it still has a larger error than $\hbeta_{\rm ds}$ when $p = 10$ or $t$ is small. The dShrink estimator consistently achieves a smaller error than $\hbeta_{\cT}$ and outperforms its competitors in most scenarios. 
	
	Figure \ref{fig: sim cmpr} compares the dShrink estimator with the constrained maximum likelihood (CML) method \citep{chatterjee2016constrained}, the generalized integration method (GIM) \citep{zhang2020generalized}, the penalized constrained maximum likelihood (PCML) method \citep{zhai2022data}, and the JS type estimator (JS Comb) proposed in \citep{han2024improving} based on $5000$ simulation runs.
	
	\begin{figure}[h]
		\centering
		\subfigure[]{
			\includegraphics[scale = 0.25]{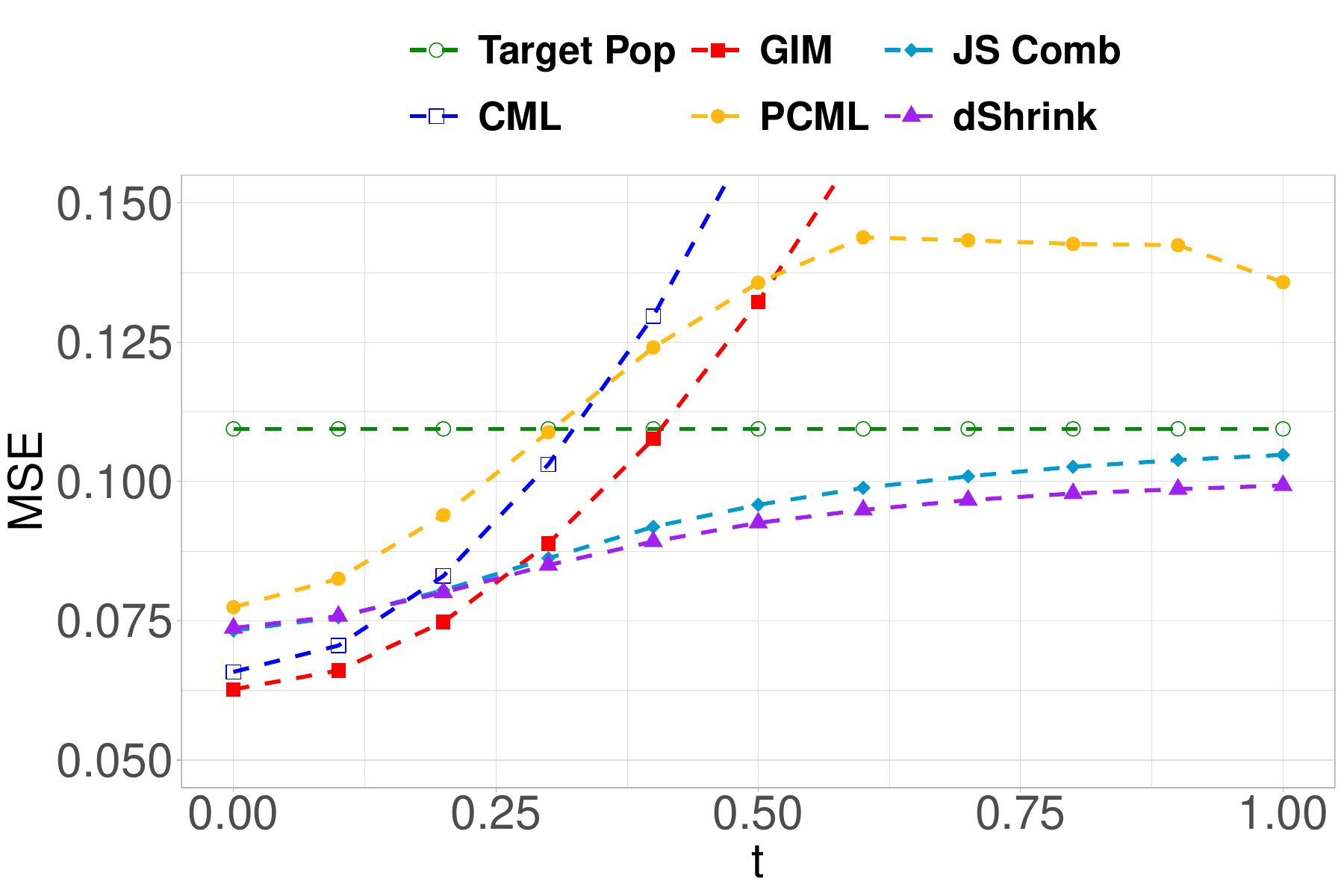}
		}
		\subfigure[]{
			\includegraphics[scale = 0.25]{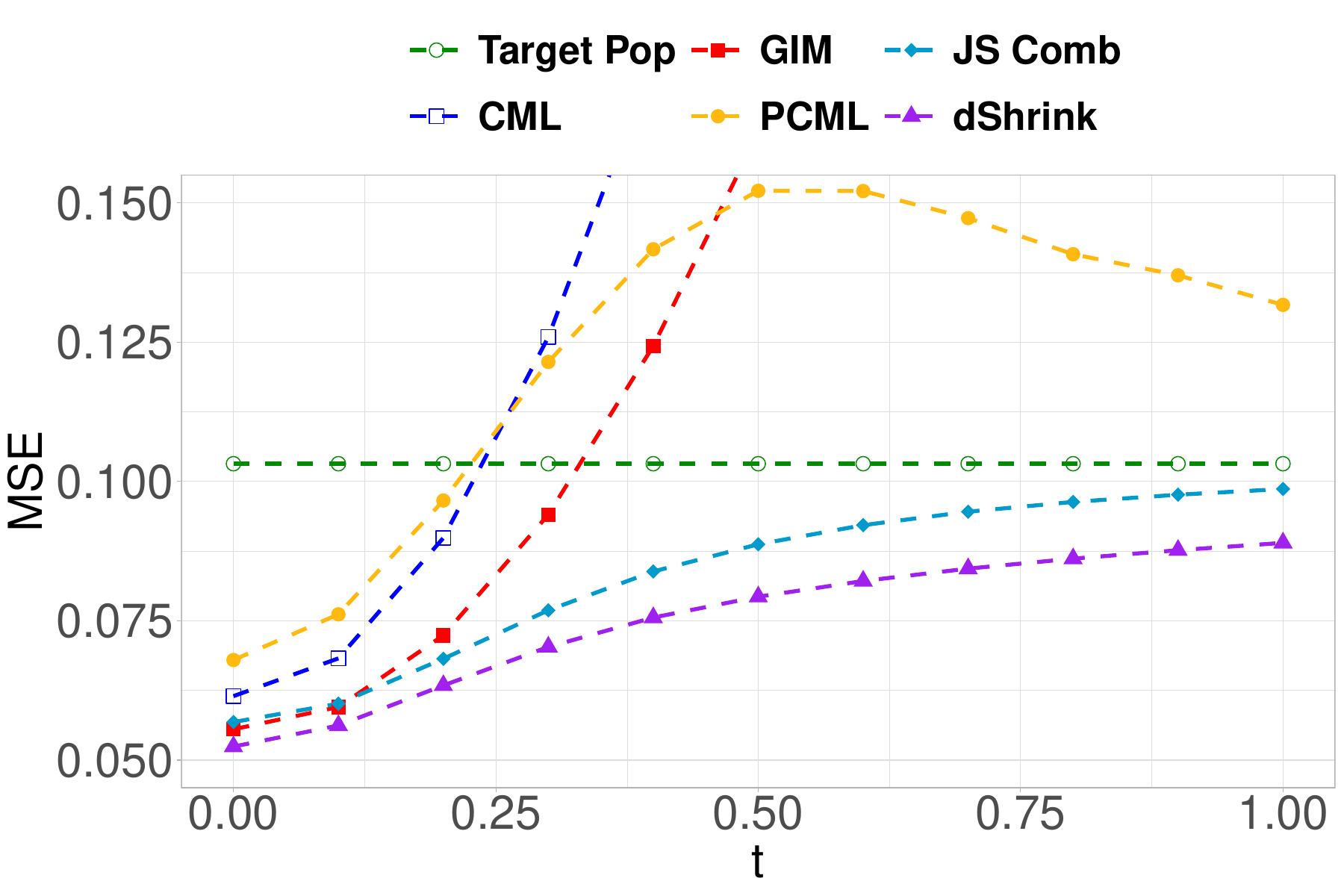}
		}
		\caption{\it MSEs of the regression functions derived from $\hbeta_{\cT}$, $\hbeta_{\rm ds}$, and the existing methods CML, GIM, PCML in simulation studies. (a) $p_{X} = 10$, $q_{X} = 6$, $n_{\cT} = 100$, and $n_{\cS} = 500$; (b)  $p_{X} = 30$, $q_{X} = 20$, $n_{\cT} = 300$, and $n_{\cS} = 1000$.}\label{fig: sim cmpr}
	\end{figure}
	The CML and GIM estimators perform well when $t$ is small but have large errors when $t$ is large. The PCML estimator achieves a smaller error than the CML estimator and the GIM estimator when $t$ is large, while it has a larger error than these two estimators when $t$ is small. The error of the dShrink estimator is smaller than the PCML estimator across all scenarios. The dShrink estimator demonstrates superior performance compared to the CML and GIM estimators when $t$ is not very small.	JS Comb and the dShrink estimator consistently outperform $\hbeta_{\cT}$, supporting our theoretical results that $\hbeta_{\rm ds}$ is safe. Moreover, dShrink achieves a smaller MSE than JS Comb in all cases, confirming the advantages of dShrink discussed in Section~\ref{subsec: JS and data enrich}. In addition, JS Comb \citep{han2024improving} is restricted to linear models, whereas dShrink applies to a broader class of parameter estimation problems. Simulation results in Supplementary Material Section~\ref{app: sim correlated obs} further demonstrate that dShrink performs well in nonlinear models with non-normal, correlated data.
	Simulation results regarding the c-dShrink estimator are provided in Appendix \ref{app: sim m&c}.
	
	Next, we evaluate the performance of the CIs based on $\hbeta_{\rm ds}$ under the simulation setting in Section \ref{subsec: sim}. We set $\bbH$ to be an estimate of $\bbSig_{\beta, \cT}^{-1}$ in the implementation. 
	Figure \ref{fig: sim CI} shows the average coverage rates and average widths over different components of the target population-based standard CIs 
	$[\hat{\beta}_{\cT, j} - \hat{\sigma}_{\cT, j} q_{1 - \tau / 2}, \hat{\beta}_{\cT, j} + \hat{\sigma}_{\cT, j} q_{1 - \tau / 2}]$,
	and the CIs \eqref{eq: CI ds} with $B = 10$, where $\hat{\beta}_{\cT, j}$ is the $j$-th component of $\hbeta_{\cT}$ and  $q_{1 - \tau / 2}$ is the $1 - \tau / 2$ quantile of a standard normal distribution. The CIs \eqref{eq: CI ds} have larger coverage rates and shorter average widths compared to those based on $\hbeta_{\cT}$ in most settings.
	\begin{figure}[h]
		\centering
		\subfigure[]{
			\includegraphics[scale = 0.28]{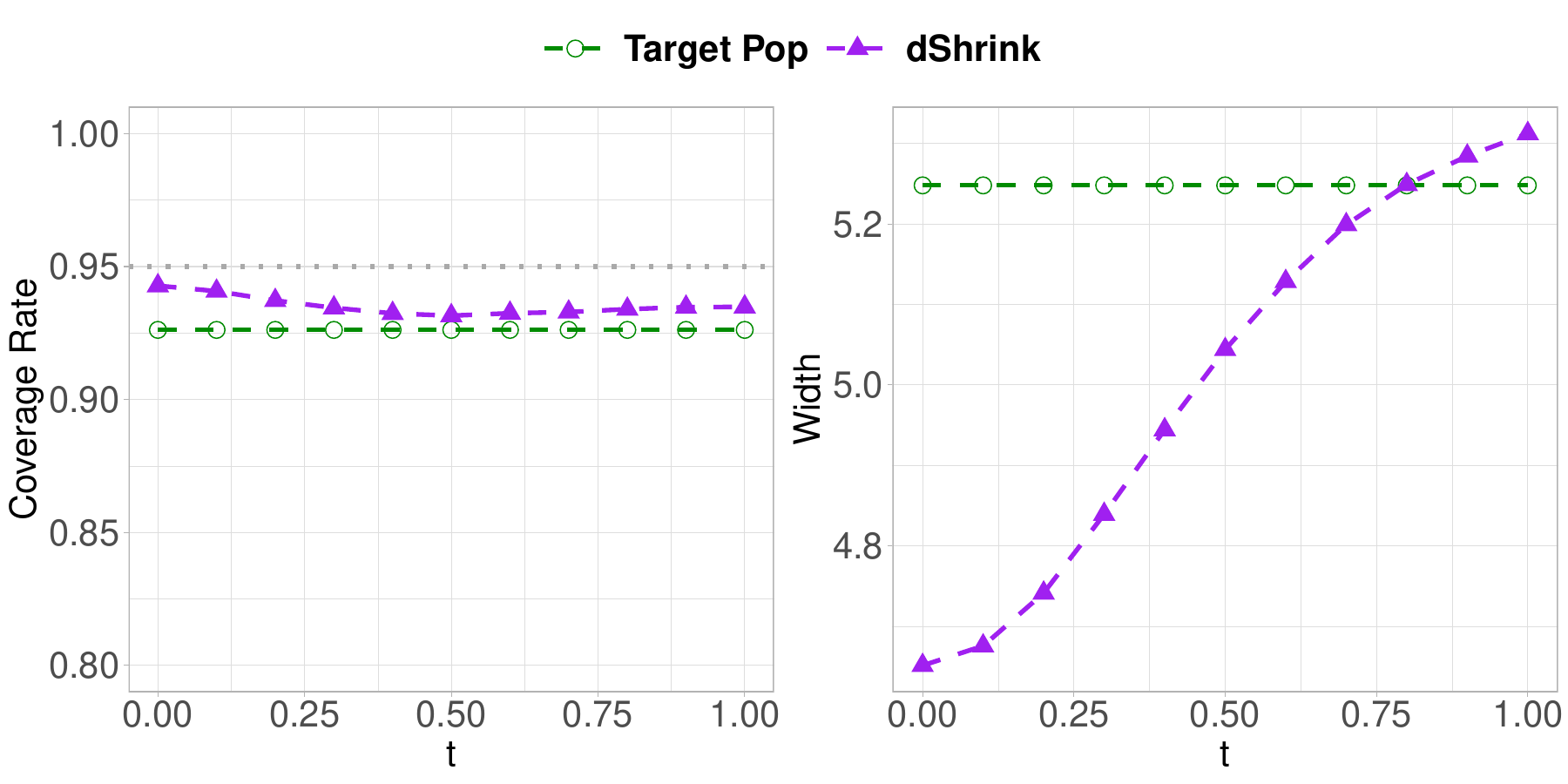}
		}
		\subfigure[]{
			\includegraphics[scale = 0.28]{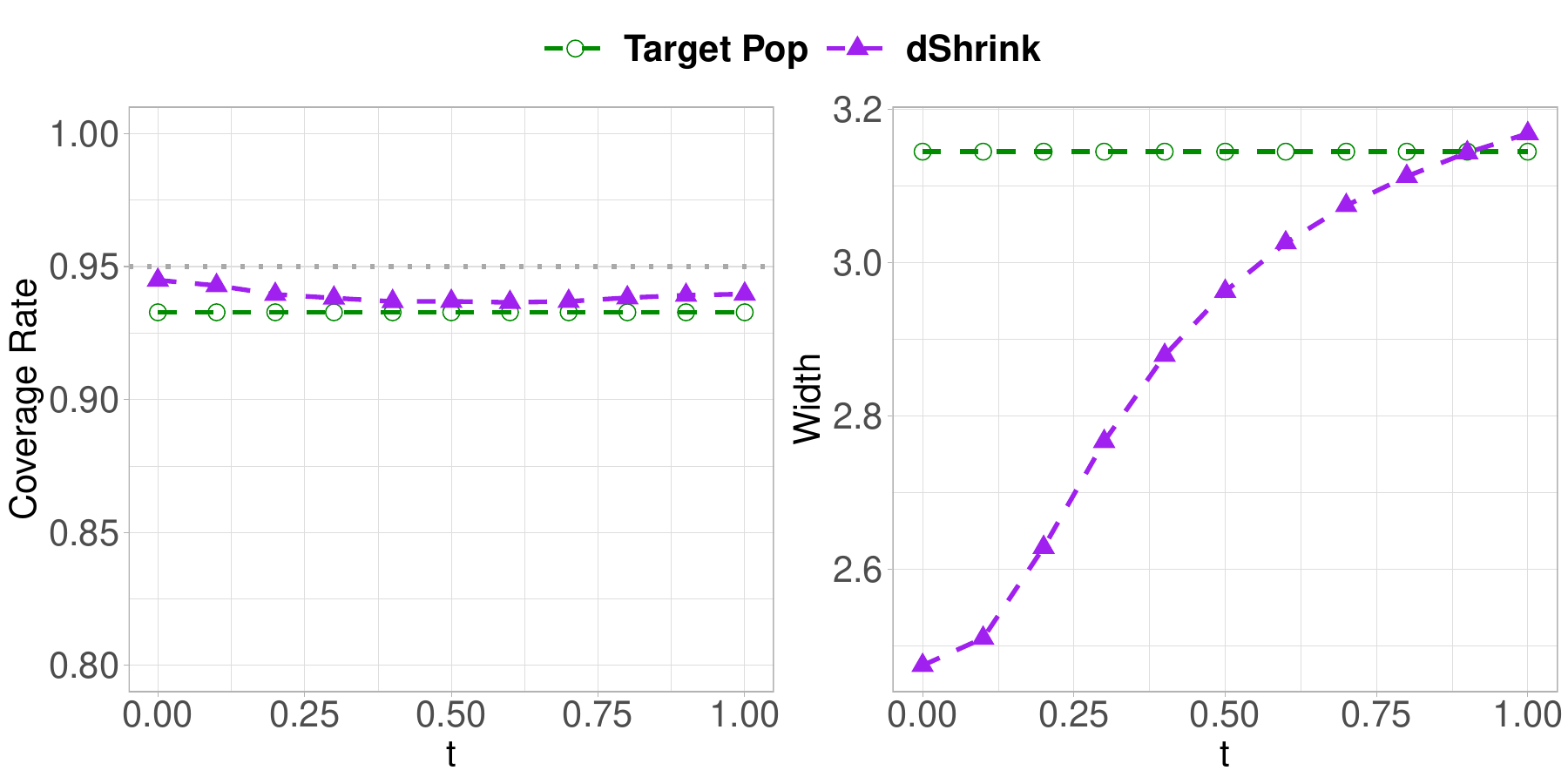}
		}
		\caption{\it Average coverage rates and average widths over different components of the standard CIs based on $\hbeta_{\cT}$ and the CIs \eqref{eq: CI ds} in simulation studies. The average widths are multiplied by ten. (a) $p_{X} = 10$, $q_{X} = 6$, $n_{\cT} = 100$, and $n_{\cS} = 500$; (b)  $p_{X} = 30$, $q_{X} = 20$, $n_{\cT} = 300$, and $n_{\cS} = 1000$.}\label{fig: sim CI}
	\end{figure}

    This section considers a setting in which different sets of variables are observed in the target and source populations, and only summary statistics from the source population are available. Under this scenario, many existing transfer estimation methods—such as data-enriched linear regression \citep{chen2015data}, TransLasso \citep{li2022transfer}, TransGLM \citep{tian2022transfer}, and ISEDI \citep{hector2024turning}—are not directly applicable.
    In Supplementary Material Section~\ref{app: sim same variable}, we also compare the proposed method with the above methods under a setting where the same variables are observed in both populations and individual-level data are available. Even though our method uses only summary statistics, it outperforms the competitors in terms of MSE and coverage rate, and reduces the width of the CIs compared to the target population-based estimator in most cases.
	
	%
	%
 
\subsection{Real data based simulation}\label{subsec: sim infant death}
	In this section, we conduct a real-data-based simulation study based on the 2015 birth cohort 
    data set collected by the U.S. Centers for Disease Control and Prevention, to evaluate the performance of the proposed method in realistic settings. In this simulation, we fit the regression relationship between the infant death $Y$ and various covariates $\bX$. The considered covariates include the mother's age, mother's education, mother's marital status, mother's smoking status, mother's BMI, the length of gestation, the birth weight, plurality, and the sex of the infant. We used the discretized covariates in the original data and converted the categorical variable to dummy variables. This results in a $25$-dimensional covariate vector $\bX$. About $6.6\%$ of observations had missing values and were removed from the analysis. We divided the cohort into subgroups according to the maternal race. The parameter of interest is the regression coefficient in the logistic regression of $Y$ on $\bX$ in the Southeast Asian population, which has the sample size $n_{\cT} \approx 50 {\rm K}$. Moreover, the proportion of infant deaths is very low ($\approx 0.4\%$ in the target population), which makes the effective sample size much smaller than $n_{\cT}$ and necessitates transfer estimation. 
    We use the maximum likelihood estimator for the logistic regression coefficient based on all $25$ covariates from another race as the summary statistics $\halpha_{\cS}$ and apply the dShrink method to borrow information from $\halpha_{\cS}$ to estimate the target parameter. Two source populations, Black and  White, are considered, which have much larger sample sizes $n_{\cS} \approx 2.8 {\rm M}$ and $561 {\rm K}$, respectively. The multi-source dShrink is implemented to incorporate both populations.
    
    We sampled with replacement from both the target and source population data to obtain the simulated data sets and repeated the simulation for $200$ times. To reduce the computational burden while mimicking the real data, we adopt a case-control design to obtain the resampled data set with case numbers matching the original ones and a case-control ratio $1:5$. The target population estimator $\hbeta_{\cT}$ is the maximum likelihood estimator using only the data from Southeast Asian.  
    We then calculated the target population-based estimator and external summary statistics using the resampled data set and applied different methods to estimate the target parameter. We took $\bbH$ to be an estimate of $E_{\cT}[\partial p(\bX^{\T}\bbeta_{\cT})/\partial \bbeta \{\partial p(\bX^{\T}\bbeta_{\cT})/\partial \bbeta\}^{\T}]$ in the implementation of the dShrink method, where $p(\cdot)$ is the logistic function. This choice of $\bbH$ reflects the first-order approximation of the MSE of the regression function. 
    We treated the target population-based estimator obtained from the original data as the true parameter (weights are imposed to account for the case-control design) and calculated the MSEs of the regression function based on different estimators. MSEs (multiplied by $n_{\cT}$) based on $\hbeta_{\cT}$, $\hbeta_{\cC}$, data enriched logistic regression (DELogis) \citep{zheng2025data}, TransGLM, ISEDI, and dShrink are shown in Table \ref{table: infant sim}.  We constructed $0.95$-CI for each component of the target parameter using different methods. The DELogis result is absent because its inference procedure is not available. We took $\bbH$ as an estimate of $\bbSig_{\beta, \cT}^{-1}$ in dShrink when constructing CIs. The average CI widths and coverage are also reported in Table \ref{table: infant sim}. 

    \begin{table}
    \caption{\label{table: infant sim}
    \it  MSE, CI width and coverage of different methods in the simulation based on CDC infant death data. The CI width is reported only when the CI coverage approximately achieves the nominal level ($> 0.9$).}
    \resizebox{\textwidth}{!}{
    \begin{tabular}{l*{6}{c}}
    \toprule
    Method & Target population only & Calibration & DELogis & TransGLM & ISEDI & dShrink \\
    \midrule
    $n_{\cT}\rm MSE$ &
    3.39 & 3.01 & 3.42 & 2.41 & 2.26 & 1.83\\
    CI Width & 1.94 & -- & -- & 2.32 & -- & 1.68\\
    Coverage Rate & 0.93 & 0.25 & -- & 0.92 & 0.85 & 0.92\\
    \bottomrule
    \end{tabular}}
\end{table}
    
Table \ref{table: infant sim} shows that all methods except for the data enriched logistic regression improve the MSE compared to the target population only method, with dShrink making the most significant improvement (more than $45\%$).
The CIs based on the target population method, TransGLM, and dShrink can achieve coverage rates close to the nominal level. Among these methods, dShrink produces CIs with the smallest average width.

    \section{Data Analysis}
    \label{sec: data analysis}
	\subsection{Application to the CDC infant death risk analysis}\label{subsec: infant death}
	In this section, we apply the proposed dShrink method to the analysis of the CDC infant death data set introduced in Section \ref{subsec: sim infant death}, to study the associations between the infant death $Y$ and the covariates $\bX$ described in the same section in the target Southeast Asian population. We use Black and White as the source populations. We exclude the intercept from both $\hbeta_{\cT}$ and $\halpha_{\cS}$ because the association between the $\bY$ and $\bX$ is of primary interest. Table \ref{table: infant death}  presents the point estimates (Est), CIs, and CI widths for the effects of four potential risk factors, i.e., maternal obesity, maternal smoking, multiple births, and male gender of the infant on the risk of infant death. For the dShrink CI estimates, we took $\bbH$ to be an estimate of $\bbSig_{\beta, \cT}^{-1}$ when implementing $\hbeta_{\rm ds}$. We compare the results using the target population, Southeast Asian, only with the dShrink results using the target Southeast Asian population and the source populations, Black and White. We didn't compute the results of TransGLM and ISEDI which also produced (approximately) valid CIs in the simulation in Section \ref{subsec: sim infant death}, because their packages do not seem scalable for the large data set considered here (R console reported more than $50$ TB memory is required when implementing these packages).
    
    Table \ref{table: infant death} shows that the CIs of the dShrink estimates $\hbeta_{\rm ds}$ are much shorter than those calculated based on the target Southeast Asian population $\hbeta_{\cT}$. The target population-based estimator identifies no significant risk factor. Notably, the dShrink method identifies smoking as a significant risk factor in Southeast Asian. This finding aligns with the literature that maternal smoking can increase the risk of infant death \citep{pineles2016systematic}.
     The average width of the CIs across regression coefficients of all the covariates considered in Section \ref{subsec: sim infant death} based on the target population estimator $\hbeta_{\cT}$ and dShrink estimator $\hbeta_{\rm ds}$ are $1.50$ and $1.09$, respectively. This demonstrates the effectiveness of the dShrink method in narrowing CIs.

\begin{table}
    \caption{\label{table: infant death}
    \it  Analysis results of the CDC infant death data to study the risk factors of the risk of infant death in the target Asian population. The point estimates and CIs of the regression coefficients of different risk factors based on the target Asian population only and the dShrink method that uses the target population Asian and the source populations Black and White.}
    \resizebox{\textwidth}{!}{
    \begin{tabular}{l*{6}{c}}
    \toprule
    \multirow{2}{*}{Risk factor} & \multicolumn{3}{c}{Target population only} & \multicolumn{3}{c}{dShrink}\\
    \multirow{2}{*}{ } & \multicolumn{3}{c}{Southeast Asian} & \multicolumn{3}{c}{(Source populations: Black and White)} \\
    \midrule
    \specialrule{0em}{2pt}{2pt}
    & EST & CI & CI Width & EST & CI & CI Width\\
    \midrule
    Obesity& 0.07 & $[-0.39, 0.54]$ & 0.93 & 0.13 & $[-0.18, 0.47]$ & 0.66\\
    Smoking& 0.93 & $[-0.04, 1.89]$ & 1.94 & 0.72 & $[0.09, 1.41]$ & 1.32 \\
    Multiple births& -0.03 & $[-0.43, 0.37]$ & 0.81 & 0.08 & $[-0.24, 0.32]$ & 0.56\\
    Male infant& 0.08 & $[-0.21, 0.37]$ & 0.58 & 0.14 & $[-0.11, 0.29]$ & 0.41\\
    \bottomrule
    \end{tabular}}
\end{table}


    \bigskip
    
	\subsection{Application to estimating treatment effects   of rectal indomethacin}\label{subsec: ITE}
	Rectal indomethacin has been demonstrated to be effective in preventing post-ERCP pancreatitis \citep{elmunzer2020skyrocketing}.  However, the
	price of rectal indomethacin has approximately 20-fold increase since 2012 \citep{elmunzer2020skyrocketing}.
	Considering the cost implications, it is desirable to administer the preventive drug only when it is truly necessary. This requires evaluating the treatment effect of rectal indomethacin for different patients at a more granular level. Available data include individual-level data from a randomized trial in the US with sample size $n_{\cT} = 602$ \citep{elmunzer2012randomized} and summary statistics from a randomized trial in Hungarian with sample size $n_{\cS} = 686$ \citep{dobronte2014rectal} which contains the incidence rate of each subgroup. 
    
    Suppose $P_{\cT}$ is the US population studied by the trial in \cite{elmunzer2012randomized} and $P_{\cS}$ is the Hungarian population studied in \cite{dobronte2014rectal}. The covariates in the randomized trial in \cite{elmunzer2012randomized} consist of age, gender, and $26$ risk factors for post-ERCP pancreatitis. To avoid the collinearity issue, we excluded covariates whose multiple correlation with other covariates is larger than $0.95$. Let $\bX$ be the random vector consisting of the intercept term ``$1$" and the $21$ remaining covariates, $T \in \{0, 1\}$ the treatment indicator which indicates the use of rectal indomethacin, $Y \in \{0, 1\}$ the indicator of post-ERCP pancreatitis. Then $P_{\cT}(T=1) = 1/2$ in the randomized trial. Let $Y_{1}$ and $Y_{0}$ be the potential outcomes under treatment and control, respectively. In addition, let $\halpha_{\cS}$ be the vector of incidence rates of each subgroup in \cite{dobronte2014rectal}. Neither the source population covariance matrix $\bbSig_{\alpha,\cS}$ nor its estimate is available in this example. Thus, we employed the conditional dShrink (c-dShrink) method discussed in Section \ref{subsec: cds},  to incorporate the summary statistics $\halpha_{\cS}$. 

    We first analyze the average treatment effects for the target population and average treatment effects at the subgroup level by stratifying the target population into six age-based subgroups.
    The average treatment effect of rectal indomethacin on post-ERCP pancreatitis estimated using the target population-based difference-in-means estimator is $-0.07$ with a CI of $[-0.13, -0.02]$. The c-dShrink estimator that incorporates $\halpha_{\cS}$ produces an estimate $-0.08$ with a CI of $[-0.13, -0.01]$. These results of the two estimators largely align with each other, and both of them reproduce the findings that rectal indomethacin is effective in preventing post-ERCP pancreatitis \citep{elmunzer2020skyrocketing}. 
     Table \ref{table: post-ERCP}  shows the Est, CIs, and CI widths for the subgroup-level average treatment effects obtained from both the target population-based estimator and the c-dShrink estimator. The results indicate that the CIs derived from the c-dShrink estimator are consistently narrower than those from the target population-based estimator. In addition, the target population only estimator detects no subgroup-level average treatment effect, probably due to the sample size limitation. On the contrary, the c-dShrink estimator suggests a significant effect of rectal indomethacin in the $25$--$35$ age group. 
     
     \begin{table}
    \caption{\label{table: post-ERCP}
    \it  Analysis results of the age subgroup-level average treatment effects of rectal indomethacin on post-ERCP pancreatitis}
    \begin{tabular}{l*{6}{c}}
    \toprule
    Age group & \multicolumn{3}{c}{Target population only} & \multicolumn{3}{c}{c-dShrink}\\
    \midrule
    \specialrule{0em}{2pt}{2pt}
    & EST & CI & Width & EST & CI & Width \\
    \midrule
    $\leq 25$ & -0.13 & $[-0.33, 0.07]$ & 0.40 & -0.15 & $[-0.31, 0.05]$ & 0.36\\
    $(25, 35]$ & -0.09 & $[-0.22, 0.04]$ & 0.26 & -0.15 & $[-0.25, -0.02]$ & 0.23\\
    $(35, 45]$ & -0.05 & $[-0.16, 0.05]$ & 0.21 & -0.05 & $[-0.14, 0.04]$ & 0.18\\
    $(45, 55]$ & -0.05 & $[-0.14, 0.05]$ & 0.19 & -0.04 & $[-0.13, 0.04]$ & 0.17\\
    $(55, 65]$ & -0.14 & $[-0.29, 0.01]$ & 0.30 & -0.06 & $[-0.20, 0.07]$ & 0.27\\
    $> 65$ & 0.00 & $[-0.14, 0.15]$ & 0.29 & 0.01 & $[-0.12, 0.13]$ & 0.25\\
    \bottomrule
    \end{tabular}
\end{table}
    Next, we study the individualized treatment effect $\tau(\bX) = E[Y_{1} - Y_{0}\mid \bX]$, a quantity that is particularly relevant in personalized medicine.
    Define the pseudo-outcome $Y^{\dag} = 2TY - 2(1 - T)Y$. Then $E[Y^{\dag}\mid \bX] = \tau(\bX)$. Suppose $\tau(\bX) = m(\bX^{\T}\bbeta_{\cT})$,
	where $m(v) = 2e^{v}/(1 + e^{v}) - 1$ is a shifted-scaled logistic function to the range of $[-1, 1]$.
	The parameter $\bbeta_{\cT}$ can be estimated by the least squares regression of $Y^{\dag}$ on $\bX$. 
    The resulting estimator is denoted by $\hbeta_{\cT}$. 
	We calculate the point estimates and their CIs for the individualized treatment effects of individual patients based on their covariate values in the data set of \cite{elmunzer2012randomized} using data from the target US population $\hbeta_{\cT}$ and the c-dShrink estimate $\hbeta_{\rm ds}^{\dag}$ that uses the data from both the target US population and the Hungarian population, respectively.  Specifically, the target population-based and the transferred estimate for individuals with covariate $\bx$ is calculated as $m(\bx^{\T}\hbeta_{\cT})$ and $m(\bx^{\T}\hbeta_{\rm ds}^{\dag})$, respectively. The matrix $\bbH$ was taken as the inverse of the estimated covariance matrix of $\hbeta_{\cT}$. The average widths of the resulting CIs based on $\hbeta_{\cT}$ and $\hbeta_{\rm ds}^{\dag}$ overall patients are $0.48$ and $0.27$, respectively. Figure \ref{fig: ITE} plots the point estimates and CIs of the individualized treatment effects for covariate values of $10$ randomly selected patients, which demonstrates the effectiveness of the c-dShrink method in reducing the width of the CIs for each individual.
	
	\begin{figure}
		\centering
		\includegraphics[scale = 0.4]{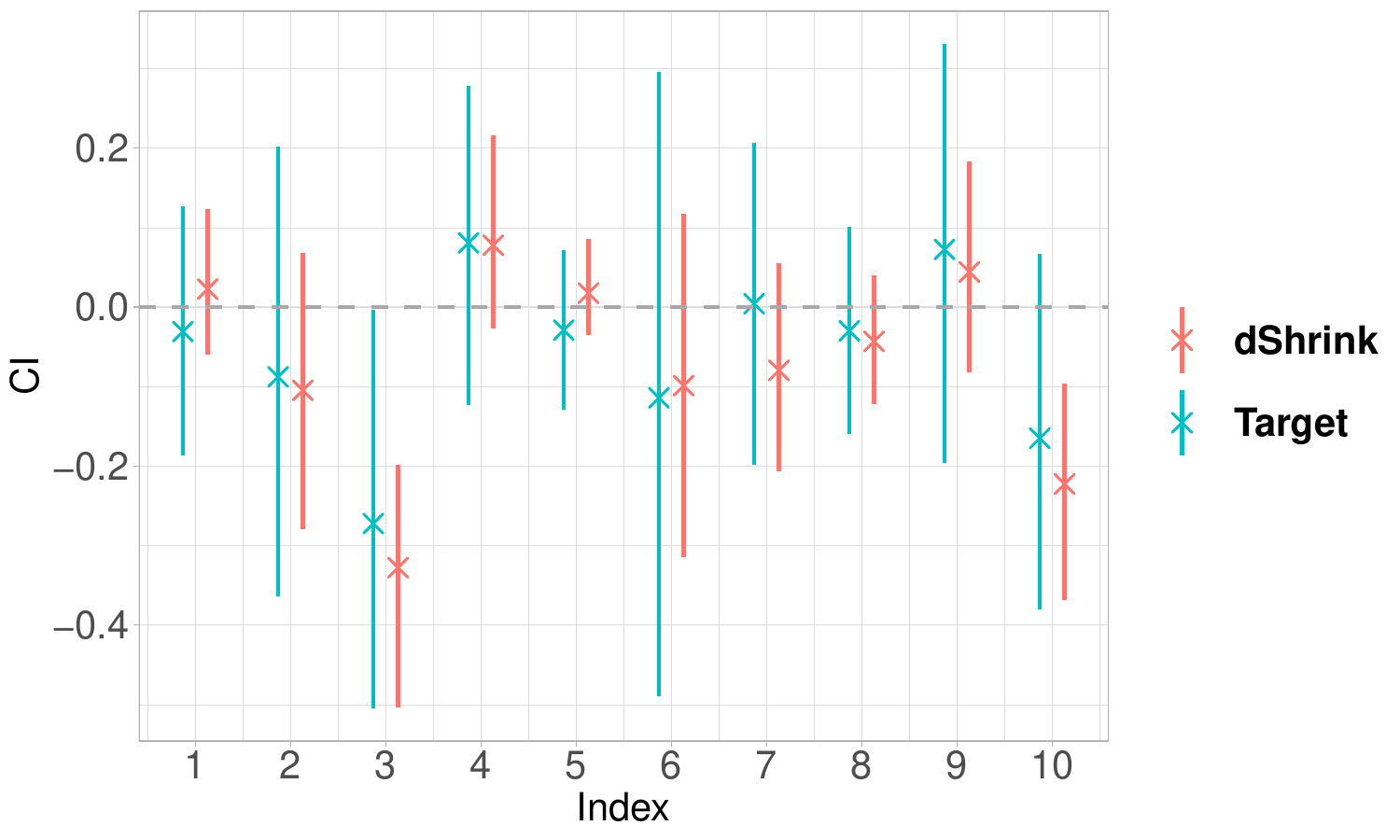}
		\caption{\it Confidence intervals of the individualized treatment effects for the covariate values of $10$ randomly selected patients from the randomized trial \citep{elmunzer2012randomized} constructed using the target population-based estimator and the dShrink estimator}\label{fig: ITE}
	\end{figure}

	\section{Discussion}\label{sec: disscuss}
	We propose in this paper the transfer estimation method dShrink that offers efficient and robust estimators of target population parameters in a closed form using
    summary statistics from both a target population and an external source population, while accounting for population and parameter heterogeneity between the two populations. The dShrink estimator presents a suite of desirable features.
	First, it is robust to population heterogeneity. The dShrink estimator effectively handles the difference between target and source populations with few assumptions on the form or magnitude of heterogeneity. Second, it can accommodate parameter heterogeneity. The dShrink estimator is capable of integrating external summary statistics even when there is a mismatch with the target parameters. Third, the dShrink method does not include any user-specified tuning parameter, which simplifies its usage and makes it accessible for a wide range of applications. Fourth, the dShrink estimator can incorporate additional information, such as study-level information and side information, to improve the estimator's efficiency. Fifth, the dShrink estimator can achieve guaranteed improvement. It is {\it safe} in the sense that it consistently outperforms the target population-based estimator $\hbeta_{\cT}$ in terms of reducing the MSE, ensuring the benefit of positive knowledge transfer.  Furthermore, our extensive simulation studies show that the dShrink estimator also outperforms the existing methods.
	The robust, flexible, and efficient nature of the dShrink estimator makes it stand out as a promising tool for transfer estimation, offering significant advantages in statistical analysis to robustly and efficiently estimate target parameters.

    The dShrink method makes the (asymptotic) normality assumption to estimate the $\cE$-error and the separable surrogate objective function. In high-dimensional problems, penalized methods are often adopted to construct estimators that are not asymptotically normal. The high-dimensional federated learning method proposed by \cite{li2023targeting} enables transfer learning without sharing individual-level data, which can prevent negative transfer in the sense that information sharing does not degrade the convergence rate. However, the guarantee on convergence rate does not imply a reduction in asymptotic estimation error; for example, many estimators share the same $1/\sqrt{n}$ convergence rate in scalar parameter estimation problems while differing substantially in efficiency. A future direction of considerable interest involves extending the dShrink method to high-dimensional settings, ensuring a provable reduction in estimation error.
        
	\bibliographystyle{plainnat}
	\bibliography{dShrink}

	\appendix
	
	\renewcommand{\thesection}{S\arabic{section}}
	\renewcommand{\thecondition}{S\arabic{condition}}
	\renewcommand{\thetheorem}{S\arabic{theorem}}
	\renewcommand{\thetable}{S\arabic{table}}
	\renewcommand{\thefigure}{S\arabic{figure}}
	\renewcommand{\theexample}{S\arabic{example}}
	\renewcommand{\theproposition}{S\arabic{proposition}}
	\renewcommand{\thelemma}{S\arabic{lemma}}
	\renewcommand{\theequation}{S\arabic{equation}}
	
	\newpage
	
	\centerline{\Large\bf Supplementary Material}
	
	\vspace{20pt}
	\section{Component-wise pretest estimator}\label{app: comp test}
	Let $\sA = \{j: (\halpha_{\cT, j} - \halpha_{\cS, j})^{2} / \var(\halpha_{\cT, j} - \halpha_{\cS,j}) \leq c_{j}\}$, where $c_{j}\geq 0$ is the threshold for the $j$th component and $\halpha_{\cT,j}$ and $\halpha_{\cS,j}$ is the $j$th components of $\halpha_{\cT}$ and $\halpha_{\cS}$, respectively. Then the component-wise pretest estimator $\hbeta_{\rm cprt}$ is defined as the estimator that calibrate $\hbeta_{\cT}$ using only the components in $\sA$. Specifically, let $\halpha_{\cS}^{\sA}$ and $\halpha_{\cT}^{\sA}$ be the sub-vectors of $\halpha_{\cS}$ and $\halpha_{\cT}$ that contain only the components in $\sA$. Let $\bbSig_{\alpha,\cS}^{\sA}$, $\bbSig_{\alpha, \cT}^{\sA}$ and $\bbSig_{\beta\alpha, \cT}^{\sA}$ be the corresponding variance and covariance matrices. Then the component-wise pretest estimator is given by
	\[
	\hbeta_{\rm cprt} = \hbeta_{\cT} - \bbQ_{\alpha}^{\sA}(\halpha_{\cT}^{\sA} - \halpha_{\cS}^{\sA}),
	\]
	where $\bbQ_{\alpha}^{\sA} = \bbSig_{\beta\alpha, \cT}^{\sA}(\bbSig_{\alpha, \cT}^{\sA} + \bbSig_{\alpha,\cS}^{\sA})^{-1}$. We have the following result for $\hbeta_{\rm cprt}$.
	\begin{proposition}\label{prop: pretest comp}
		Suppose $\bbH = \bbI_{p}$. In the example stated in \eqref{eq: counter eg} in the main text, for any $n_{\cT}$, $n_{\cS}$, $\boldsymbol{\mu}$, and any $c_{1}, \dots, c_{p}$ that are not all equal to zero, there is some $\bbias$ such that $\cE(\hbeta_{\rm cprt}, \bbeta_{\cT}) > \cE(\hbeta_{\cT}, \bbeta_{\cT})$.
	\end{proposition}
	
	\section{Model-assisted dShrink}\label{app: mds}
	
	The dShrink estimator shrinks $\hbeta_{\cC}$ and $\hdelta_{\cC}$ towards zero. It can be enhanced by incorporating 
	side information \citep{banerjee2020adaptive, luo2023empirical} to guide the shrinkage towards more informed targets.  The 
	side
	information, such as genetic variant functional annotations \citep{li2020dynamic} or disease ontology, 
	can provide model-based predictions for $\hbeta_{\cC}$ and $\hdelta_{\cC}$. Recall that $\bbQ_{\alpha}(\halpha_{\cT} - \halpha_{\cS})$ depend on $\halpha_{\cT}$ and $\halpha_{\cS}$. Suppose the 
	side information for $\beta_{\cT, j}$ can be summarized as a vector $\boldsymbol{a}_{\beta j}$ for $j = 1,\dots, p$, where $\beta_{\cT, j}$ is the $j$th component of $\bbeta_{\cT}$. Similarly, let $\boldsymbol{a}_{d j}$ be the 
	side information for the difference $\alpha_{\cT, j} - \alpha_{\cS, j}$, where $\alpha_{\cT, j}$ and $\alpha_{\cS, j}$ are the $j$th component of $\balpha_{\cT}$ and $\balpha_{\cS}$, respectively, for $j = 1,\dots, q$. We assemble these vectors into matrices 
	\[
	\bbA_{\beta} =
	\begin{pmatrix}
		\boldsymbol{a}_{\beta1}^{\T}\\
		\vdots\\
		\boldsymbol{a}_{\beta p}^{\T}
	\end{pmatrix}, \ 
	\bbA_{d} = 
	\begin{pmatrix}
		\boldsymbol{a}_{d 1}^{\T}\\
		\vdots\\
		\boldsymbol{a}_{d q}^{\T}
	\end{pmatrix}, \ \text{and} \ 
	\bbA_{\cC} = (\bbA_{\beta} \enspace \bbQ_{\alpha}\bbA_{d}),
	\] 
	and apply linear regression models by regressing $\hbeta_{\cC}$ and $\halpha_{\cT} - \halpha_{\cS}$ on their respective 
	side information variable
	matrices $\bbA_{\cC}$ and $\bbA_{d}$ respectively.  Recall that $\bbSig_{\cC} = \var(\hbeta_{\cC})$.
	Define $\bbSig_{d} = \var(\halpha_{\cT} - \halpha_{\cS}) = \bbSig_{\alpha, \cT} + \bbSig_{\alpha,\cS}$,  the weighted least-square hat matrices of these two models $\bbPi_{\cC} = \bbA_{\cC}(\bbA_{\cC}^{\T}\bbSig_{\cC}^{-1}\bbA_{\cC})^{-1}\bbA_{\cC}^{\T}\bbSig_{\cC}^{-1}$,  and $\bbPi_{d} = \bbA_{d}(\bbA_{d}^{\T}\bbSig_{d}^{-1}\bbA_{d})^{-1}\bbA_{d}^{\T}\bbSig_{d}^{-1}$ respectively.
	Then, the weighted least-square model-based estimators for $\bbeta_{\cC}$ and $\bdelta_{\cC}$ are
	$\hbeta_{\rm m} = \bbPi_{\cC}\hbeta_{\cC}$ and $\hdelta_{\rm m} = \bbQ_{\alpha}\bbPi_{d}(\halpha_{\cT} - \halpha_{\cS})$, 
	respectively. 
	
	We refine the dShrink estimator by shrinking $\hbeta_{\cC}$ and $\hdelta_{\cC}$ towards the side-information model-based estimators $\hbeta_{\rm m}$ and $\hdelta_{\rm m}$ instead of zero.   Following the same rationale behind $\widehat{\lambda}_{{\rm s}, 1}$ and $\widehat{\lambda}_{{\rm s}, 2}$ with $\hbeta_{\cC}$ and $\hdelta_{\cC}$ replaced by $\hbeta_{\cC} - \hbeta_{\rm m}$ and $\hdelta_{\cC} - \hdelta_{\rm m}$, we define the shrinkage factors 
	\[
	\begin{aligned}
		\widehat{\lambda}_{{\rm ms}, 1} & = \left\{1 -  \frac{c_{\beta, {\rm m}}}{(\hbeta_{\cC} - \hbeta_{\rm m})^{\T}\bbH(\hbeta_{\cC} - \hbeta_{\rm m})}\right\}_{+},\\ 
		\widehat{\lambda}_{{\rm ms}, 2} & = \left\{1 - \frac{c_{\delta, {\rm m}}}{(\hdelta_{\cC} - \hdelta_{\rm m})^{\T}\bbH(\hdelta_{\cC} - \hdelta_{\rm m})}\right\}_{+},
	\end{aligned}
	\]	
	where 
	\[
	\begin{aligned}
		c_{\beta, {\rm m}} & = (\tr\{\bbH\bbPi_{\cC}\bbSig_{\cC}\bbPi_{\cC}^{\T}\} - 2\sigma_{\rm max}(\bbH\bbPi_{\cC}\bbSig_{\cC}\bbPi_{\cC}^{\T}))_{+}, \\
		c_{\delta, {\rm m}} & = (\tr\{\bbH\bbQ_{\alpha}\bbPi_{d}\bbSig_{d}\bbPi_{d}^{\T}\bbQ_{\alpha}^{\T}\} - 2\sigma_{\rm max}(\bbH\bbQ_{\alpha}\bbPi_{d}\bbSig_{d}\bbPi_{d}^{\T}\bbQ_{\alpha}^{\T}))_{+}.
	\end{aligned}
	\]		
	
	The model-assisted dShrink (m-dShrink) estimator is
	\[
	\hbeta_{\rm mds} = 
	(1 - \widehat{\lambda}_{{\rm ms}, 1})\hbeta_{\rm m} + \widehat{\lambda}_{{\rm ms}, 1}\hbeta_{\cC} + (1-\widehat{\lambda}_{{\rm ms}, 2})\hdelta_{\rm m} + \widehat{\lambda}_{{\rm ms}, 2}\hdelta_{\cC}.
	\]
	Notice that the first two terms in $\hbeta_{\rm mds}$ 
	shrink $\hbeta_{\cC}$ towards the side information model-based estimator $\hbeta_{\rm m}$. The last two terms in $\hbeta_{\rm mds}$ can be interpreted similarly. Thus, the m-dShrink estimator make use of the  side information by shrinking $\hbeta_{\cC}$ and $\hdelta_{\cC}$ towards the side information model-based estimators. Theorem 2 shows that m-dShrink estimator $\hbeta_{\rm mds}$ can also achieve safe transfer estimation.  The derivation of $\hbeta_{\rm mds}$ and the proof of Theorem 2 are given in Supplemental Material Section \ref{app: proof mds}.
	\begin{theorem}\label{thm: dominance mds}
		Suppose $\halpha_{\cS}$, $\halpha_{\cT}$, and $\hbeta_{\cT}$ are normally distributed as in \eqref{eq: normality}. Then, $\cE(\hbeta_{\rm mds},\bbeta_{\cT}) \leq \cE(\hbeta_{\cT}, \bbeta_{\cT})$ for any $\balpha_{\cS}$, $\balpha_{\cT}$, and  $\bbeta_{\cT}$. If $c_{\beta, {\rm m}} > 0$ or $c_{\delta, {\rm m}} > 0$, then $\cE(\hbeta_{\rm mds},\bbeta_{\cT}) < \cE(\hbeta_{\cT}, \bbeta_{\cT})$ for any $\balpha_{\cS}$, $\balpha_{\cT}$, and  $\bbeta_{\cT}$.
	\end{theorem}
	Note that the m-dShrink method constructs model-based estimates but does not assume the side-information based working models are correct. The m-dShrink estimator incorporate 
	side 
	information through the linear model while retaining the robustness to model misspecification and heterogeneity between populations. In addition, it can achieve substantial efficiency gain if the regression models of $\bbeta_{\cC}$ and $\bdelta_{\cC}$ on 
	side information variables  are approximately correct.
	
	\section{Proofs and additional theoretical results}\label{app: proofs}
	\subsection{Proofs of Proposition \ref{prop: pretest} and Proposition \ref{prop: pretest comp}}\label{app: proof prop}
	Before the proof of the propositions, we first prove a useful lemma. For any positive integer $k$ and $a > 0$, we use $\chi_{a, k}^{2}$ to denote a random variable that follows the noncentral chi-square distribution with $k$ degrees of freedom and non-centrality parameter $a$. For a positive definite $p\times p$ matrix $\bbH$, any $p$-dimensional vectors $\boldsymbol{v}_{1}$ and  $\boldsymbol{v}_{2}$, define the inner product $\langle \boldsymbol{v}_{1}, \boldsymbol{v}_{2} \rangle_{\bbH} =  \boldsymbol{v}_{1}^{\T}\bbH\boldsymbol{v}_{2}$, norm $\|\boldsymbol{v}_{1}\|_{\bbH} = \sqrt{\langle \boldsymbol{v}_{1}, \boldsymbol{v}_{1} \rangle_{\bbH}}$. For two sequences of positive numbers $\{c_{1n}\}_{n=1}^{\infty}$ and $\{c_{2n}\}_{n=1}^{\infty}$, we say $c_{1n} \asymp c_{2n}$ or $c_{1n} = \Theta(c_{2n})$ if $C^{-1}c_{1n} \leq c_{2n} \leq C c_{1n}$ for some constant $C > 1$.
	
	\begin{lemma}\label{lem: threshold norm}
		Let $\widehat{\bxi} \sim N(\bxi, \bbI_{p})$ be a $p$-dimensional normal vector. Then, for any positive constant $c$ and $p \times p$ positive definite matrix $\bbH$, 
		\[
		E\left[1\{\|\widehat{\bxi}\|^{2} > c\}\widehat{\bxi}\right] =  P(\chi_{\|\bxi\|^{2}/2, p + 2}^{2} > c)\bxi,
		\]
		and
		\[
		E\left[1\{\|\widehat{\bxi}\|^{2} > c\}\|\widehat{\bxi}\|_{\bbH}^{2}\right] =  P(\chi_{\|\bxi\|^{2}/2, p + 2}^{2} > c)\tr\{\bbH\} + P(\chi_{\|\bxi\|^{2}/2, p + 4}^{2} > c)\|\bxi\|_{\bbH}^{2}.
		\]
	\end{lemma}
	\begin{proof}
		According to Theorem 1 in \cite{bock1973statistical} and the continuity of the normal distribution, we have
		\[
		E\left[1\{\|\widehat{\bxi}\|^{2} \leq c\}\widehat{\bxi}\right] =  P(\chi_{\|\bxi\|^{2}/2, p + 2}^{2} \leq c)\bxi
		\]
		and hence
		\[
		\begin{aligned}
			E\left[1\{\|\widehat{\bxi}\|^{2} > c\}\widehat{\bxi}\right] 
			& = E[\widehat{\bxi}] - E\left[1\{\|\widehat{\bxi}\|^{2} \leq c\}\widehat{\bxi}\right]\\
			& =  P(\chi_{\|\bxi\|^{2}/2, p + 2}^{2} > c)\bxi.
		\end{aligned}
		\]
		Similarly, according to Theorem 2 in \cite{bock1973statistical}, we have
		\[
		E\left[1\{\|\widehat{\bxi}\|^{2} \leq c\}\|\widehat{\bxi}\|_{\bbH}^{2}\right] =  P(\chi_{\|\bxi\|^{2}/2, p + 2}^{2} \leq c)\tr\{\bbH\} + P(\chi_{\|\bxi\|^{2}/2, p + 4}^{2} \leq c)\|\bxi\|_{\bbH}^{2}.
		\]
		and hence
		\[
		\begin{aligned}
			E\left[1\{\|\widehat{\bxi}\|^{2} > c\}\|\widehat{\bxi}\|_{\bbH}^{2}\right] 
			& = E\left[\|\widehat{\bxi}\|_{\bbH}^{2}\right] - E\left[1\{\|\widehat{\bxi}\|^{2} \leq c\}\|\widehat{\bxi}\|_{\bbH}^{2}\right]\\
			& = \tr\{\bbH\} + \|\bxi\|_{\bbH}^{2} -  P(\chi_{\|\bxi\|^{2}/2, p + 2}^{2} \leq c)\tr\{\bbH\} - P(\chi_{\|\bxi\|^{2}/2, p + 4}^{2} \leq c)\|\bxi\|_{\bbH}^{2}\\
			& =  P(\chi_{\|\bxi\|^{2}/2, p + 2}^{2} > c)\tr\{\bbH\} + P(\chi_{\|\bxi\|^{2}/2, p + 4}^{2} > c)\|\bxi\|_{\bbH}^{2},
		\end{aligned}
		\]
		which completes the proof.
	\end{proof}
	
	Next, we get down to the proof of the propositions.
	
	\begin{proof}
		Recall that $\hbeta_{\cT} = \halpha_{\cT} \sim N(\boldsymbol{\mu} + \bbias, 2n_{\cT}^{-1}\bbI_{p}),\  \hbeta_{\cS} = \halpha_{\cS} \sim N(\boldsymbol{\mu} - \bbias, 2n_{\cS}^{-1}\bbI_{p})$.
		Then, in this example, $\hbeta_{\cC} = (n_{\cT}\halpha_{\cT} + n_{\cS}\halpha_{\cS}) / (n_{\cT} + n_{\cS})$, $\hdelta_{\cC} = n_{\cS}(\halpha_{\cT} - \halpha_{\cS}) / (n_{\cT} + n_{\cS})$, $\bbeta_{\cC} = \boldsymbol{\mu} + (n_{\cT} - n_{\cS})\bbias/(n_{\cT} + n_{\cS})$, and $\bdelta_{\cC} = 2n_{\cS}\bbias / (n_{\cT} + n_{\cS})$. Thus, $\var(\hbeta_{\cC}) = 2/(n_{\cT} + n_{\cS})\bbI_{p} \equalscolon v_{\cC}^{2}\bbI_{p}$, $\var(\hdelta_{\cC}) = 2n_{\cS}/\{n_{\cT}(n_{\cT} + n_{\cS})\}\bbI_{p} \equalscolon v_{\delta}^{2}\bbI_{p}$, and $v_{\cC}^{2} + v_{\delta}^{2} = 2n_{\cT}^{-1}$. For $\hbeta_{\rm prt}$, we have $\hbeta_{\rm prt} = \hbeta_{\cC} + 1\{\|\hdelta_{\cC}\|^{2} > c_{\rm g}\}\hdelta_{\cC}$. Then, we have
		\begin{align*}
			\cE(\hbeta_{\rm prt}, \bbeta_{\cT}) &= \cE(\hbeta_{\cC}, \bbeta_{\cC})+ E\left[\|1\{\|\hdelta_{\cC}\|^{2}/v_{\delta}^{2} > c_{\rm g}\}\hdelta_{\cC} - \bdelta_{\cC}\|_{\bbH}^{2}\right]\\
			&= v_{\cC}^{2}\tr\{\bbH\} + E\left[1\{\|\hdelta_{\cC}\|^{2}/v_{\delta}^{2} > c_{\rm g}\}\|\hdelta_{\cC}\|_{\bbH}^{2}\right] + \|\bdelta_{\cC}\|_{\bbH}^{2} - 2\bdelta_{\cC}^{\T}\bbH E\left[1\{\|\hdelta_{\cC}\|^{2}/v_{\delta}^{2} > c_{\rm g}\}\hdelta_{\cC}\right]\\
			&=v_{\cC}^{2}\tr\{\bbH\} + P(\chi_{\|\bdelta_{\cC}\|^{2}/(2v_{\delta}^{2}),p + 2}^{2} > c_{\rm g})v_{\delta}^{2}\tr\{\bbH\} + P(\chi_{\|\bdelta_{\cC}\|^{2}/(2v_{\delta}^{2}),p + 4}^{2} > c_{\rm g})\|\bdelta_{\cC}\|_{\bbH}^{2} + \|\bdelta_{\cC}\|_{\bbH}^{2}\\
			&\quad - 2P(\chi_{\|\bdelta_{\cC}\|^{2}/(2v_{\delta}^{2}),p + 2}^{2} 
			> c_{\rm g})\|\bdelta_{\cC}\|_{\bbH}^{2}\\
			&=(v_{\cC}^{2} + v_{\delta}^{2})\tr\{\bbH\} + v_{\delta}^{2}\tr\{\bbH\}\left\{P(\chi_{\|\bdelta_{\cC}\|^{2}/(2v_{\delta}^{2}),p + 2}^{2} > c_{\rm g}) - 1\right\} \\
			&\quad +
			\|\bdelta_{\cC}\|_{\bbH}^{2}\left\{1 - P(\chi_{\|\bdelta_{\cC}\|^{2}/(2v_{\delta}^{2}),p + 2}^{2} > c_{\rm g})\right\}\\
			&\quad + \|\bdelta_{\cC}\|_{\bbH}^{2}\left\{P(\chi_{\|\bdelta_{\cC}\|^{2}/(2v_{\delta}^{2}),p + 4}^{2} > c_{\rm g}) - P(\chi_{\|\bdelta_{\cC}\|^{2}/(2v_{\delta}^{2}),p + 2}^{2} > c_{\rm g})\right\}\\
			&=(v_{\cC}^{2} + v_{\delta}^{2})\tr\{\bbH\} + (\|\bdelta_{\cC}\|_{\bbH}^{2} -  v_{\delta}^{2}\tr\{\bbH\})P(\chi_{\|\bdelta_{\cC}\|^{2}/(2v_{\delta}^{2}),p + 2}^{2} \leq c_{\rm g})\\
			&\quad + \|\bdelta_{\cC}\|_{\bbH}^{2}\left\{P(\chi_{\|\bdelta_{\cC}\|^{2}/(2v_{\delta}^{2}),p + 4}^{2} > c_{\rm g}) - P(\chi_{\|\bdelta_{\cC}\|^{2}/(2v_{\delta}^{2}),p + 2}^{2} > c_{\rm g})\right\}\\
			&\geq (v_{\cC}^{2} + v_{\delta}^{2})\tr\{\bbH\} + (\|\bdelta_{\cC}\|_{\bbH}^{2} -  v_{\delta}^{2}\tr\{\bbH\})P(\chi_{\|\bdelta_{\cC}\|^{2}/(2v_{\delta}^{2}),p + 2}^{2} \leq c_{\rm g}),
		\end{align*}
		where the third equality is due to Lemma \ref{lem: threshold norm} and the last inequality is because $P(\chi_{\|\bdelta_{\cC}\|^{2}/(2v_{\delta}^{2}),p + 4}^{2} > c_{\rm g}) - P(\chi_{\|\bdelta_{\cC}\|^{2}/(2v_{\delta}^{2}),p + 2}^{2} > c_{\rm g}) > 0$ for any $\bdelta_{\cC}$ and positive $c_{g}$, $v_{\delta}^{2}$.        
		This implies $\cE(\hbeta_{\rm prt}, \bbeta_{\cT}) > (v_{\cC}^{2} + v_{\delta}^{2})\tr\{\bbH\} = 2n_{\cT}^{-1}\tr\{\bbH\} = \cE(\hbeta_{\cT}, \bbeta_{\cT})$ if $\|\bdelta_{\cC}\|_{\bbH}^{2} >  v_{\delta}^{2}\tr\{\bbH\}$.
		
		For $\hbeta_{\rm cprt}$ with $\bbH = \bbI_{p}$, we have $\hat{\beta}_{{\rm cprt},j} = \hat{\beta}_{\cC, j} + 1\{\hat{\delta}_{j}^{2} / v_{\delta}^{2} > c_{j}\} \hat{\delta}_{\cC, j}$ for $j = 1,\dots, p$, where $\hat{\beta}_{\cC, j}$ and $\hat{\delta}_{\cC, j}$ are the $j$th component of $\hbeta_{\cC}$ and $\hdelta_{\cC}$, respectively.
		Recall that $\hbeta_{\cC}\Perp \hdelta_{\cC}$.
		We have
			\begin{align*}
				&\cE(\hbeta_{\rm cprt}, \bbeta_{\cT})\\
				&= \cE(\hbeta_{\cC}, \bbeta_{\cC})+ \sum_{j}E\{(1\{\hat{\delta}_{j}^{2} / v_{\delta}^{2} > c_{j}\} \hat{\delta}_{\cC, j} - \delta_{\cC, j})^{2}\}\\
				&= v_{\cC}^{2}p + \sum_{j}\left(E[1\{\hat{\delta}_{j}^{2} / v_{\delta}^{2} > c_{j}\} \hat{\delta}_{\cC, j}^{2}] + \delta_{\cC, j}^{2} - 2\delta_{\cC, j}E[1\{\hat{\delta}_{j}^{2} / v_{\delta}^{2} > c_{j}\}\hat{\delta}_{\cC, j}]\right)\\
				&=v_{\cC}^{2}p + \sum_{j}\left\{P(\chi_{\delta_{\cC, j}^{2}/(2v_{\delta}^{2}),3}^{2} > c_{j})v_{\delta}^{2} + P(\chi_{\delta_{\cC, j}^{2}/(2v_{\delta}^{2}),5}^{2} > c_{j})\delta_{\cC, j}^{2} + \delta_{\cC, j}^{2} - 2\delta_{\cC, j}^{2}P(\chi_{\delta_{\cC, j}^{2}/(2v_{\delta}^{2}),3}^{2} > c_{j})\right\}\\
				&=(v_{\cC}^{2} + v_{\delta}^{2})p + \sum_{j}(\delta_{\cC, j}^{2} - v_{\delta}^{2})P(\chi_{\delta_{\cC, j}^{2}/(2v_{\delta}^{2}),3}^{2} \leq c_{j})\\
				&\quad + \sum_{j}\delta_{\cC, j}^{2}\left\{P(\chi_{\delta_{\cC, j}^{2}/(2v_{\delta}^{2}),5}^{2} > c_{j}) - P(\chi_{\delta_{\cC, j}^{2}/(2v_{\delta}^{2}),3}^{2} > c_{j})\right\}\\
				&\geq (v_{\cC}^{2} + v_{\delta}^{2})p + \sum_{j}(\delta_{\cC, j}^{2} - v_{\delta}^{2})P(\chi_{\delta_{\cC, j}^{2}/(2v_{\delta}^{2}),3}^{2} \leq c_{j}),
			\end{align*}
		where $\delta_{\cC, j}$ is the $j$th component of $\bdelta_{\cC}$ for $j = 1,\dots,p$, the third equality is due to Lemma \ref{lem: threshold norm}, and the last inequality is because $P(\chi_{\delta_{\cC, j}^{2}/(2v_{\delta}^{2}),5}^{2} > c_{j}) - P(\chi_{\delta_{\cC, j}^{2}/(2v_{\delta}^{2}),3}^{2} > c_{j}) > 0$ for any $\delta_{\cC, j}$ and positive $c_{j}$, $v_{\delta}^{2}$. Thus $\cE(\hbeta_{\rm cprt}, \bbeta_{\cT}) > (v_{\cC}^{2} + v_{\delta}^{2})p = 2n_{\cT}^{-1}p = \cE(\hbeta_{\cT}, \bbeta_{\cT})$ if $\delta_{\cC, j} > v_{\delta}^{2}$ for $j = 1,\dots, p$.
	\end{proof}
	
	\subsection{Derivation of the dShrink estimator}\label{app: derive dShrink}
	The minimizers of $\cE(\lambda_{1}\hbeta_{\cC}, \bbeta_{\cC})$ and $\cE(\lambda_{2}\hdelta_{\cC}, \bdelta_{\cC})$ are 
	\[
	\lambda_{{\rm s}, 1} = 1 - \tr\{\bbH\bbSig_{\cC}\} / [\tr\{\bbH\bbSig_{\cC}\} + \bbeta_{\cC}^{\T}\bbH\bbeta_{\cC}],
	\]
	and
	\[
	\lambda_{{\rm s}, 2} = 1 - \tr\{\bbH\bbSig_{\delta}\} / [\tr\{\bbH\bbSig_{\delta}\} + \bdelta_{\cC}^{\T}\bbH\bdelta_{\cC}],
	\]
	respectively. However, $\lambda_{{\rm s}, 1}, \lambda_{{\rm s}, 2}$ depends on the unknown $\bbeta_{\cC}$ and $\bdelta_{\cC}$ and can not be calculated directly from data. To address this, we invoke the Stein's formula \citep{stein1981estimation} which gives the unbiased estimates $\tr\{\bbH\bbSig_{\cC}\} + (\lambda_{1} - 1)^{2}\hbeta_{\cC}^{\T}\bbH\hbeta_{\cC} - 2(\lambda_{1} - 1)\tr\{\bbH\bbSig_{\cC}\}$ and $\tr\{\bbH\bbSig_{\delta}\} + (\lambda_{2} - 1)^{2}\hdelta_{\cC}^{\T}\bbH\hdelta_{\cC} - 2(\lambda_{2} - 1)\tr\{\bbH\bbSig_{\delta}\}$ for $\cE(\lambda_{1}\hbeta_{\cC}, \bbeta_{\cC})$ and $\cE(\lambda_{2}\hdelta_{\cC}, \bdelta_{\cC})$, respectively.  Notice that $\lambda_{{\rm s}, 1}, \lambda_{{\rm s}, 2} \in [0, 1]$.
	Minimizing these unbiased estimates with respect to $\lambda_{1}\in [0, 1]$ and $\lambda_{2}\in [0, 1]$ yields the minimizers
	\[
	\left(1 - \frac{\tr\{\bbH\bbSig_{\cC}\}}{\hbeta_{\cC}^{\T}\bbH\hbeta_{\cC}}\right)_{+}\ \text{and}\ 
	\left(1 - \frac{\tr\{\bbH\bbSig_{\delta}\}}{\hdelta_{\cC}^{\T}\bbH\hdelta_{\cC}}\right)_{+},
	\]
	where $(u)_{+} = \max\{0, u\}$ for any $u$.
	In the theoretical analysis, we find it necessary to slightly adjust the constants $\tr\{\bbH\bbSig_{\cC}\}$ and $\tr\{\bbH\bbSig_{\delta}\}$ in order to properly account for the estimation error. This gives 
	the estimator $(\widehat{\lambda}_{{\rm s}, 1}, \widehat{\lambda}_{{\rm s}, 2})$ in the main text.
	\subsection{Proof of \eqref{eq: pop dominance} in the main text}\label{app: proof pop dominance}
	\begin{proof}
		For any $\lambda_{1}$ and $\lambda_{2}$, we have
		\begin{equation}\label{eq: E upper bound}
			\begin{aligned}
				\cE(\lambda_{1}\hbeta_{\cC} + \lambda_{2}\hdelta_{\cS}, \bbeta_{\cT}) 
				& = \cE(\lambda_{1}\hbeta_{\cC} + \lambda_{2}\hdelta_{\cS}, \bbeta_{\cC} + \bdelta_{\cC})\\
				& = 
				\cE(\lambda_{1}\hbeta_{\cC}, \bbeta_{\cC}) + \cE(\lambda_{2}\hdelta_{\cC}, \bdelta_{\cC}) + 2(1 - \lambda_{1})(1 - \lambda_{2})\bbeta_{\cC}^{\T}\bbH\bdelta_{\cC}\\
				&\leq \cE(\lambda_{1}\hbeta_{\cC}, \bbeta_{\cC}) + \cE(\lambda_{2}\hdelta_{\cC}, \bdelta_{\cC}) + (1 - \lambda_{1})^{2}\bbeta_{\cC}^{\T}\bbH\bbeta_{\cC} + (1 - \lambda_{2})^{2}\bdelta_{\cC}^{\T}\bbH\bdelta_{\cC}\\
				&= \lambda_{1}^{2}\tr\{\bbH\bbSig_{\cC}\} + 2(1 - \lambda_{1})^{2}\bbeta_{\cC}^{\T}\bbH\bbeta_{\cC} + \lambda_{2}^{2}\tr\{\bbH\bbSig_{\delta}\} + 2(1 - \lambda_{2})^{2}\bdelta_{\cC}^{\T}\bbH\bdelta_{\cC}.
			\end{aligned}
		\end{equation}
		Plug $\lambda_{{\rm s}, 1}$ and $\lambda_{{\rm s}, 2}$ into the upper bound in \eqref{eq: E upper bound}.
		For the terms regarding $\lambda_{{\rm s}, 1}$ in \eqref{eq: E upper bound}, we have
		\[
		\begin{aligned}
			\lambda_{{\rm s}, 1}^{2}\tr\{\bbH\bbSig_{\cC}\} + 2(1 - \lambda_{{\rm s}, 1})^{2}\bbeta_{\cC}^{\T}\bbH\bbeta_{\cC}
			& = \frac{(\bbeta_{\cC}^{\T}\bbH\bbeta_{\cC})^{2}\tr\{\bbH\bbSig_{\cC}\}}{(\bbeta_{\cC}^{\T}\bbH\bbeta_{\cC} + \tr\{\bbH\bbSig_{\cC}\})^{2}} 
			+ \frac{2\bbeta_{\cC}^{\T}\bbH\bbeta_{\cC}\tr\{\bbH\bbSig_{\cC}\}^{2}}{(\bbeta_{\cC}^{\T}\bbH\bbeta_{\cC} + \tr\{\bbH\bbSig_{\cC}\})^{2}}\\
			& = \frac{\bbeta_{\cC}^{\T}\bbH\bbeta_{\cC}\tr\{\bbH\bbSig_{\cC}\}}{\bbeta_{\cC}^{\T}\bbH\bbeta_{\cC} + \tr\{\bbH\bbSig_{\cC}\}} 
			+ \frac{\bbeta_{\cC}^{\T}\bbH\bbeta_{\cC}\tr\{\bbH\bbSig_{\cC}\}^{2}}{(\bbeta_{\cC}^{\T}\bbH\bbeta_{\cC} + \tr\{\bbH\bbSig_{\cC}\})^{2}}\\
			& < \frac{\bbeta_{\cC}^{\T}\bbH\bbeta_{\cC}\tr\{\bbH\bbSig_{\cC}\}}{\bbeta_{\cC}^{\T}\bbH\bbeta_{\cC} + \tr\{\bbH\bbSig_{\cC}\}} 
			+ \frac{\tr\{\bbH\bbSig_{\cC}\}^{2}}{\bbeta_{\cC}^{\T}\bbH\bbeta_{\cC} + \tr\{\bbH\bbSig_{\cC}\}}\\
			& = \tr\{\bbH\bbSig_{\cC}\}.
		\end{aligned}
		\]
		The similar result holds for the term regarding $\lambda_{{\rm s}, 2}$ in \eqref{eq: E upper bound}. Then
		\[
		\cE(\lambda_{{\rm s}, 1}\hbeta_{\cC} + \lambda_{{\rm s}, 2}\hdelta_{\cC}, \bbeta_{\cT}) < \tr\{\bbH\bbSig_{\cC}\} + \tr\{\bbH\bbSig_{\delta}\} = \cE(\hbeta_{\cT}, \bbeta_{\cT}),
		\]
		which completes the proof.
	\end{proof}
	
	\subsection{Proof of Theorem \ref{thm: dominance ds}}\label{app: proof of dominance}
	\begin{proof}
		The $\cE$-error of the dShrink estimator $\hbeta_{\rm ds}$ can be decomposed as  
		\begin{align}	
			\cE(\hbeta_{\rm ds}, \bbeta_{\cT}) &= \cE(\widehat{\lambda}_{{\rm s}, 1}\hbeta_{\cC} +  \widehat{\lambda}_{{\rm s}, 2}\hdelta_{\cC} ,\bbeta_{\cT})\notag\\
			&=E\left\{(\widehat{\lambda}_{{\rm s}, 1}\hbeta_{\cC} - \bbeta_{\cC})^{\T}\bbH(\widehat{\lambda}_{{\rm s}, 1}\hbeta_{\cC} - \bbeta_{\cC})\right\} + E\left\{(\widehat{\lambda}_{{\rm s}, 2}\hdelta_{\cC} - \bdelta_{\cC})^{\T}\bbH(\widehat{\lambda}_{{\rm s}, 2}\hdelta_{\cC} - \bdelta_{\cC})\right\}\notag\\
			& + 2E\left\{(\widehat{\lambda}_{{\rm s}, 1}\hbeta_{\cC} - \bbeta_{\cC})^{\T}\bbH(\widehat{\lambda}_{{\rm s}, 2}\hdelta_{\cC} - \bdelta_{\cC})\right\}\notag\\
			&=\cE(\widehat{\lambda}_{{\rm s}, 1}\hbeta_{\cC}, \bbeta_{\cC}) + \cE(\widehat{\lambda}_{{\rm s}, 2}\hdelta_{\cC}, \bdelta_{\cC}) + 2E\left\{(\widehat{\lambda}_{{\rm s}, 1}\hbeta_{\cC} - \bbeta_{\cC})^{\T}\bbH(\widehat{\lambda}_{{\rm s}, 2}\hdelta_{\cC} - \bdelta_{\cC})\right\}.\label{eq: decom1} 
		\end{align}
		Note that $\widehat{\lambda}_{{\rm s}, 1}\hbeta_{\cC} = \hbeta_{\cC} - \min\{1, c_{\beta}/\hbeta_{\cC}^{\T}\bbH\hbeta_{\cC}\} \hbeta_{\cC}$ and $\widehat{\lambda}_{{\rm s}, 2}\hdelta_{\cC} = \hdelta_{\cC} - \min\{1, c_{\delta}/\hdelta_{\cC}^{\T}\bbH\hdelta_{\cC}\} \hdelta_{\cC}$.
		According to the proof of Theorem 1 in \cite{tan2015improved}, it holds that
		\begin{equation}\label{eq: bound excess risk beta}
			\begin{aligned}
				\cE(\widehat{\lambda}_{{\rm s}, 1}\hbeta_{\cC}, \bbeta_{\cC}) 
				&\leq \cE(\hbeta_{\cC}, \bbeta_{\cC})  + E\left[\min\{\hbeta_{\cC}^{\T}\bbH\hbeta_{\cC},c_{\beta}\}\left(\min\{\hbeta_{\cC}^{\T}\bbH\hbeta_{\cC},c_{\beta}\} - 2c_{\beta}\right)\frac{1}{\hbeta_{\cC}^{\T}\bbH\hbeta_{\cC}}\right]\\
				&\leq \cE(\hbeta_{\cC}, \bbeta_{\cC})  - E\left[\min\{\hbeta_{\cC}^{\T}\bbH\hbeta_{\cC},c_{\beta}\}^{2}\frac{1}{\hbeta_{\cC}^{\T}\bbH\hbeta_{\cC}}\right],
			\end{aligned}
		\end{equation}
		and
		\begin{equation}\label{eq: bound excess risk delta}
			\begin{aligned}
				\cE(\widehat{\lambda}_{{\rm s}, 2}\hdelta_{\cC}, \bdelta_{\cC}) 
				&\leq \cE(\hdelta_{\cC}, \bdelta_{\cC})  + E\left[\min\{\hdelta_{\cC}^{\T}\bbH\hdelta_{\cC},c_{\delta}\}\left(\min\{\hdelta_{\cC}^{\T}\bbH\hdelta_{\cC},c_{\delta}\} - 2c_{\delta}\right)\frac{1}{\hdelta_{\cC}^{\T}\bbH\hdelta_{\cC}}\right]\\
				&\leq \cE(\hdelta_{\cC}, \bdelta_{\cC})  - E\left[\min\{\hdelta_{\cC}^{\T}\bbH\hdelta_{\cC},c_{\delta}\}^{2}\frac{1}{\hdelta_{\cC}^{\T}\bbH\hdelta_{\cC}}\right],
			\end{aligned}
		\end{equation}
		where the first and second strict inequalities hold if $c_{\beta} > 0$ and $c_{\delta} > 0$, respectively.
		Because $\hbeta_{\cC} \Perp \hdelta_{\cC}$ and $E[\hbeta_{\cC} - \bbeta_{\cC}] = E[\hdelta_{\cC} - \bdelta_{\cC}] = 0$,
		we have
		\begin{align*}
			& 2E\left\{(\widehat{\lambda}_{{\rm s}, 1}\hbeta_{\cC} - \bbeta_{\cC})^{\T}\bbH(\widehat{\lambda}_{{\rm s}, 2}\hdelta_{\cC} - \bdelta_{\cC})\right\}\\
			&= 2E\left[\min\left\{1, \frac{c_{\beta}}{\hbeta_{\cC}^{\T}\bbH\hbeta_{\cC}}\right\}\min\left\{1, \frac{c_{\delta}}{\hdelta_{\cC}^{\T}\bbH\hdelta_{\cC}}\right\}\hbeta_{\cC}^{\T}\bbH\hdelta_{\cC}\right]\\
			&\leq E\left[\min\left\{1, \frac{c_{\beta}}{\hbeta_{\cC}^{\T}\bbH\hbeta_{\cC}}\right\}^2\hbeta_{\cC}^{\T}\bbH\hbeta_{\cC}\right]  + E\left[\min\left\{1, \frac{c_{\delta}}{\hdelta_{\cC}^{\T}\bbH\hdelta_{\cC}}\right\}^{2}\hdelta_{\cC}^{\T}\bbH\hdelta_{\cC}\right] \\
			&= E\left[\min\left\{\hbeta_{\cC}^{\T}\bbH\hbeta_{\cC}, c_{\beta}\right\}^2\frac{1}{\hbeta_{\cC}^{\T}\bbH\hbeta_{\cC}}\right]  + E\left[\min\left\{\hdelta_{\cC}^{\T}\bbH\hdelta_{\cC}, c_{\delta}\right\}^{2}\frac{1}{\hdelta_{\cC}^{\T}\bbH\hdelta_{\cC}}\right].
		\end{align*}
		Combining this with \eqref{eq: decom1}, \eqref{eq: bound excess risk beta} and \eqref{eq: bound excess risk delta}, we conclude that 
		\[
		\cE(\hbeta_{\rm ds}, \bbeta_{\cT}) \leq \cE(\hbeta_{\cC}, \bbeta_{\cC}) + \cE(\hdelta_{\cC}, \bdelta_{\cC}) = \cE(\hbeta_{\cT}, \bbeta_{\cT}),
		\]
		where the strict inequality holds when $c_{\beta} > 0$ or $c_{\delta} > 0$.
	\end{proof}

		\subsection{Proof of the results in Remark \ref{remark: linear}}\label{app: remark 3}
		In this section, we show that the SURE estimator and the dShrink estimator can be obtained from the penalized least squares problem
		\begin{equation}\label{eq: penal reg}
			\frac{1}{\sigma_{\cT}^{2}}\|\bY_{\cT} - \bbX_{\cT}\bbeta\|^{2} + \lambda_{0}\|\bbX_{\cT}\bbeta\|^{2} + \lambda_{\rm pool}\|\bbX_{\cT}(\bbeta - \hbeta_{\rm pool})\|^{2},
		\end{equation}
		where $\hbeta_{\rm pool}$ is the minimizer of 
		\[
		\frac{1}{\sigma_{\cT}^{2}}\|\bY_{\cT} - \bbX_{\cT}\bbeta\|^{2} 
		+
		\frac{1}{\sigma_{\cS}^{2}}\|\bY_{\cS} - \bbX_{\cS}\bbeta\|^{2}.
		\]
		In linear regression, we have $\bbSig_{\beta, \cT} = \bbSig_{\beta\alpha, \cT} = \bbSig_{\alpha, \cT} = \sigma_{\cT}^{2}(\bX_{\cT}^{\T}\bX_{\cT})^{-1}$ and $\bbSig_{\alpha, \cS} = \sigma_{\cS}^{2}(\bX_{\cS}^{\T}\bX_{\cS})^{-1}$. Note that 
		\[
		\begin{aligned}
			\hbeta_{\rm pool} 
			& = \left(\sigma_{\cT}^{-2} \bX_{\cT}^{\T}\bX_{\cT} + 
			\sigma_{\cS}^{-2} \bX_{\cS}^{\T}\bX_{\cS}\right)^{-1}\left(\sigma_{\cT}^{-2}\bX_{\cT}^{\T}\bY_{\cT} + \sigma_{\cS}^{-2}\bX_{\cS}^{\T}\bY_{\cS}\right)\\
			& = \left(\bbSig_{\beta,\cT}^{-1} + \bbSig_{\alpha,\cS}^{-1}\right)^{-1}\left(\bbSig_{\beta,\cT}^{-1}\hbeta_{\cT} + \bbSig_{\alpha,\cS}^{-1}\halpha_{\cS}\right) = \hbeta_{\cC}.
		\end{aligned}
		\]
		Straightforward calculation can show that the minimizer of the objective function \eqref{eq: penal reg} is
		\[
		\begin{aligned}
			(\sigma_{\cT}^{-2} + \lambda_{0} + \lambda_{\rm pool})^{-1}(\sigma_{\cT}^{-2}\hbeta_{\cT} + \lambda_{\rm pool}\hbeta_{\cC})
			&=
			\frac{\sigma_{\cT}^{-2} + \lambda_{\rm pool}}{\sigma_{\cT}^{-2} + \lambda_{0} + \lambda_{\rm pool}}\hbeta_{\cC} + 
			\frac{\sigma_{\cT}^{-2}}{\sigma_{\cT}^{-2} + \lambda_{0} + \lambda_{\rm pool}}\hdelta_{\cC}\\
			&\eqcolon \bar{\lambda}_{1}\hbeta_{\cC} + \bar{\lambda}_{2}\hdelta_{\cC}.    \end{aligned}
		\]
		Choosing $\lambda_{0}$ and $\lambda_{\rm pool}$ (may be negative or infinite) so that $(\bar{\lambda}_{1}, \bar{\lambda}_{2}) = (\widehat{\lambda}_{\SURE, 1}, \widehat{\lambda}_{\SURE, 2})$ or $(\widehat{\lambda}_{{\rm s}, 1}, \widehat{\lambda}_{{\rm s}, 2})$ leads to the SURE estimator or dShrink estimator, respectively.
		
		\subsection{Proof of Theorem \ref{thm: dominance ds asy}}\label{app: dominance ds asy}
		\begin{proof}
			Under the conditions of Theorem \ref{thm: dominance ds asy}, the limits $c_{\beta, \infty} = \lim_{n_{\cT}\to \infty}r_{\cT}c_{\beta}$ and $c_{\delta, \infty} = \lim_{n_{\cT}\to \infty}r_{\cT}c_{\delta}$ exists, and are finite and positive. 
			Define the functions
			\[
			\begin{aligned}
				\bg_{\beta, n_{\cT}}(\ba_{1}, \ba_{2}) 
				& = \ba_{1} - \min\{1, r_{\cT}c_{\beta}\|\ba_{1} + \ba_{2}\|_{\bbH}^{-2}\}(\ba_{1} + \ba_{2}),\\
				\bg_{\delta, n_{\cT}}(\ba_{1}, \ba_{2})
				& = \ba_{1} - \min\{1, r_{\cT}c_{\delta}\|\ba_{1} + \ba_{2}\|_{\bbH}^{-2}\}(\ba_{1} + \ba_{2}),\\
				\bg_{\beta, \infty}(\ba_{1}, \ba_{2}) 
				& = \ba_{1} - \min\{1, c_{\beta, \infty}\|\ba_{1} + \ba_{2}\|_{\bbH}^{-2}\}(\ba_{1} + \ba_{2}),\\
				\bg_{\delta, \infty}(\ba_{1}, \ba_{2})
				& = \ba_{1} - \min\{1, c_{\delta, \infty}\|\ba_{1} + \ba_{2}\|_{\bbH}^{-2}\}(\ba_{1} + \ba_{2}).
			\end{aligned}
			\]
			where we allow $\ba_{2}$ to take infinite values and all the above functions are defined to be $\ba_{1}$ if $\ba_{2}$ contains infinite components. Define the distance of two numbers $a_{1}$, $a_{2}$ on the extended real line as $|a_{1}/(1 + |a_{1}|) - a_{2}/(1 + |a_{2}|)|$ where we adopt the conventions $\infty / (1 + \infty) = 1$ and $-\infty / (1 + \infty) = -1$. Then, $\bg_{\beta, n_{\cT}}$, $\bg_{\delta, n_{\cT}}$, $\bg_{\beta, \infty}$, and $\bg_{\delta, \infty}$ are continuous on $(-\infty, \infty)^{p}\times [-\infty, \infty]^{p}$ and $\bg_{\beta, n_{\cT}} \to \bg_{\beta, \infty}$, $\bg_{\delta, n_{\cT}} \to \bg_{\delta, \infty}$ uniformly as $n_{\cT} \to \infty$.  Note that 
			\begin{equation}\label{eq: decompose ds}
				\begin{aligned}
					\sqrt{r_{\cT}}(\hbeta_{\rm ds} - \bbeta_{\cT})
					& = \sqrt{r_{\cT}}(\hbeta_{\cC} - \bbeta_{\cC}) - \min\{1, c_{\beta}\|\hbeta_{\cC}\|_{\bbH}^{-2}\}  \sqrt{r_{\cT}}\hbeta_{\cC} + \sqrt{r_{\cT}}(\hdelta_{\cC} - \bdelta_{\cC}) - \min\{1, c_{\delta}\|\hdelta_{\cC}\|_{\bbH}^{-2}\} \sqrt{r_{\cT}}\hdelta_{\cC}\\
					& = \sqrt{r_{\cT}}(\hbeta_{\cC} - \bbeta_{\cC}) - \min\{1, r_{\cT}c_{\beta}\|\sqrt{r_{\cT}}(\hbeta_{\cC} - \bbeta_{\cC}) + \sqrt{r_{\cT}}\bbeta_{\cC}\|_{\bbH}^{-2}\}  \left\{\sqrt{r_{\cT}}(\hbeta_{\cC} - \bbeta_{\cC}) + \sqrt{r_{\cT}}\bbeta_{\cC}\right\}\\
					&\quad + \sqrt{r_{\cT}}(\hdelta_{\cC} - \bdelta_{\cC}) - \min\{1, r_{\cT}c_{\delta}\|\sqrt{r_{\cT}}(\hdelta_{\cC} - \bdelta_{\cC}) + \sqrt{r_{\cT}}\bdelta_{\cC}\|_{\bbH}^{-2}\} \left\{\sqrt{r_{\cT}}(\hdelta_{\cC} - \bdelta_{\cC}) + \sqrt{r_{\cT}}\bdelta_{\cC}\right\}\\
					&= \bg_{\beta, n_{\cT}}(\sqrt{r_{\cT}}(\hbeta_{\cC} - \bbeta_{\cC}), \sqrt{r_{\cT}}\bbeta_{\cC}) + \bg_{\delta, n_{\cT}}(\sqrt{r_{\cT}}(\hdelta_{\cC} - \bdelta_{\cC}), \sqrt{r_{\cT}}\bdelta_{\cC}).
				\end{aligned}
			\end{equation}    	
			Define $\bbV_{\cC}$ and $\bbV_{\delta}$ as the limits of $r_{\cT}\bbSig_{\cC}$ and $r_{\cT}\bbSig_{\delta}$ as $n_{\cT} \to \infty$. Let $\bvareps_{\cC}$ and $\bvareps_{\delta}$ be independent normal vectors that satisfy
			$\bvareps_{\cC} \sim N(\bzero, \bbV_{\cC})$ and $\bvareps_{\delta} \sim N(\bzero, \bbV_{\delta})$. Then, under the assumption of this theorem, we have $\sqrt{r_{\cT}}(\hbeta_{\cC} - \bbeta_{\cC}, \bbeta_{\cC}, \hdelta_{\cC} - \bdelta_{\cC}, \bdelta_{\cC}) \to (\bvareps_{\cC}, \bxi_{\cC}, \bvareps_{\delta}, \bxi_{\delta})$ in distribution. Then, according to Theorem 18.11 of \cite{van2000asymptotic}, we have 
			\begin{equation}\label{eq: asy dist}
				\sqrt{r_{\cT}}(\hbeta_{\rm ds} - \bbeta_{\cT}) \to \bg_{\beta, \infty}(\bvareps_{\cC}, \bxi_{\cC}) + \bg_{\delta, \infty}(\bvareps_{\delta}, \bxi_{\delta})
			\end{equation}
			in distribution. Notice that
			\[
			\begin{aligned}
				\|\sqrt{r_{\cT}}(\hbeta_{\rm ds} - \bbeta_{\cT})\|_{\bbH}^{2} 
				& \leq 3r_{\cT}\|\hbeta_{\cT} - \bbeta_{\cT}\|_{\bbH}^{2} + 3\min\{1, c_{\beta}\|\hbeta_{\cC}\|_{\bbH}^{-2}\}^{2}r_{\cT} \|\hbeta_{\cC}\|_{\bbH}^{2} + 3\min\{1, c_{\delta}\|\hdelta_{\cC}\|_{\bbH}^{-2}\}^{2}r_{\cT}\|\hdelta_{\cC}\|_{\bbH}^{2}\\
				&\leq 3r_{\cT}\|\hbeta_{\cT} - \bbeta_{\cT}\|_{\bbH}^{2} + 3\min\{1, c_{\beta}\|\hbeta_{\cC}\|_{\bbH}^{-2}\} r_{\cT}\|\hbeta_{\cC}\|_{\bbH}^{2} + 3\min\{1, c_{\delta}\|\hdelta_{\cC}\|_{\bbH}^{-2}\}r_{\cT}\|\hdelta_{\cC}\|_{\bbH}^{2}\\
				&\leq 3r_{\cT}\|\hbeta_{\cT} - \bbeta_{\cT}\|_{\bbH}^{2} + 3r_{\cT}c_{\beta} + 3r_{\cT} c_{\delta},
			\end{aligned}
			\]
			which implies 
			\begin{equation}\label{eq: ds upper bound}
				\begin{aligned}
					\|\sqrt{r_{\cT}}(\hbeta_{\rm ds} - \bbeta_{\cT})\|_{\bbH}^{2+\tau}
					& \leq 3^{1+\tau}
					\left\{\|\sqrt{r_{\cT}}(\hbeta_{\cT} - \bbeta_{\cT})\|_{\bbH}^{2 + \tau} + (r_{\cT}c_{\beta})^{1 + \frac{\tau}{2}} + (r_{\cT} c_{\delta})^{1 + \frac{\tau}{2}}\right\}\\
					& \leq 3^{1+\tau}
					\left\{\|\bbH\|^{2 + \tau}\|\sqrt{r_{\cT}}(\hbeta_{\cT} - \bbeta_{\cT})\|^{2 + \tau} + (r_{\cT}c_{\beta})^{1 + \frac{\tau}{2}} + (r_{\cT} c_{\delta})^{1 + \frac{\tau}{2}}\right\}
				\end{aligned}
			\end{equation}
			according to Jensen's inequality. Hence, we have
			$E\{\|\sqrt{r_{\cT}}(\hbeta_{\rm ds} - \bbeta_{\cT})\|_{\bbH}^{2+\tau}\} \leq C$ for some constant $C$ according to Condition \ref{cond: asy normal} (iii). Combining this with \eqref{eq: asy dist}, we have $r_{\cT}\cE(\hbeta_{\rm ds}, \bbeta_{\cT}) \to E(\|\bg_{\beta, \infty}(\bvareps_{\cC}, \bxi_{\cC}) + \bg_{\delta, \infty}(\bvareps_{\delta}, \bxi_{\delta})\|_{\bbH}^{2})$ by Example 2.21 of \cite{van2000asymptotic}. Similarly, we have $r_{\cT}\cE(\hbeta_{\cT}, \bbeta_{\cT}) \to E(\|\bvareps_{\cC} + \bvareps_{\delta}\|_{\bbH}^{2})$ by Condition \ref{cond: asy normal} (iii). Results of this Theorem follow by applying Theorem \ref{thm: dominance ds} to $\bvareps_{\cC} + \bxi_{\cC}$ and $\bvareps_{\delta} + \bxi_{\delta}$ if $ \bxi_{\cC}$ and $\bxi_{\delta}$ are both finite. Note that $\bg_{\beta, \infty}(\bvareps_{\cC}, \bxi_{\cC}) = \bvareps_{\cC}$ ($\bg_{\delta, \infty}(\bvareps_{\delta}, \bxi_{\delta}) = \bvareps_{\delta}$) when $\bxi_{\cC}$ ($\bxi_{\delta}$) contains infinite components. The case where  $\bxi_{\cC}$ or $\bxi_{\delta}$ contains infinite components can be proved using similar arguments as in the proof of Theorem \ref{thm: dominance ds}.
		\end{proof}
		
		\subsection{Proof of Theorem \ref{thm: error unknown matrix}}\label{app: estimate matrix}
		\begin{proof}
			Noting that $\lim_{n_{\cT}\to \infty}r_{\cT}c_{\beta} > 0$ and $\lim_{n_{\cT}\to \infty}r_{\cT}c_{\delta} > 0$, we have
			$\hat{c}_{\beta} - c_{\beta} = o_{P}(c_{\beta})$, $\hat{c}_{\delta} - c_{\delta} = o_{P}(c_{\delta})$, $\|\tbeta_{\cC} - \hbeta_{\cC}\| = o_{P}(\|\hbeta_{\cC}\|)$, $\|\tdelta_{\cC} - \hdelta_{\cC}\| = o_{P}(\|\hdelta_{\cC}\|)$, $\|\tbeta_{\cC}\|_{\hH}^{2} - \|\hbeta_{\cC}\|_{\bbH}^{2} = o_{P}(\|\hbeta_{\cC}\|^{2})$, and $\|\tdelta_{\cC}\|_{\hH}^{2} - \|\hdelta_{\cC}\|_{\bbH}^{2} = o_{P}(\|\hdelta_{\cC}\|^{2})$ according to Condition \ref{cond: asy normal} (i) (ii) and Condition \ref{cond: matrices error} (i). By the same decomposition as \eqref{eq: decompose ds},
			we have
			\begin{equation}\label{eq: decompos eds}
				\begin{aligned}
					\sqrt{r_{\cT}}(\tbeta_{\rm ds} - \bbeta_{\cT})
					& = \sqrt{r_{\cT}}(\tbeta_{\cC} - \bbeta_{\cC}) - \min\{1, \hat{c}_{\beta}\|\tbeta_{\cC}\|_{\hH}^{-2}\}  \sqrt{r_{\cT}}\tbeta_{\cC} + \sqrt{r_{\cT}}(\tdelta_{\cC} - \bdelta_{\cC}) - \min\{1, \hat{c}_{\delta}\|\tdelta_{\cC}\|_{\hH}^{-2}\} \sqrt{r_{\cT}}\tdelta_{\cC}\\
					& = \sqrt{r_{\cT}}(\hbeta_{\cT} - \bbeta_{\cT}) - \min\{1, \hat{c}_{\beta}\|\tbeta_{\cC}\|_{\hH}^{-2}\}  \sqrt{r_{\cT}}\tbeta_{\cC} - \min\{1, \hat{c}_{\delta}\|\tdelta_{\cC}\|_{\hH}^{-2}\} \sqrt{r_{\cT}}\tdelta_{\cC}.
				\end{aligned}
			\end{equation}    	
			Thus,
			\begin{align*}
				&\|\tbeta_{\rm ds} - \hbeta_{\rm ds}\|\\
				& \leq \left\|\min\{1, \hat{c}_{\beta}\|\tbeta_{\cC}\|_{\hH}^{-2}\}\tbeta_{\cC} - \min\{1, c_{\beta}\|\hbeta_{\cC}\|_{\bbH}^{-2}\}  \hbeta_{\cC}\right\| +
				\left\|\min\{1, \hat{c}_{\delta}\|\tdelta_{\cC}\|_{\hH}^{-2}\tdelta_{\cC} - \min\{1, c_{\delta}\|\hdelta_{\cC}\|_{\bbH}^{-2}\}\hdelta_{\cC}\right\|\\
				& = \left\|\min\{1, (c_{\beta} + o_{P}(c_{\beta}))(\|\hbeta_{\cC}\|_{\bbH} + o_{P}(\|\hbeta_{\cC}\|))^{-2}\}(\hbeta_{\cC} + o_{P}(\|\hbeta_{\cC}\|)) - \min\{1, c_{\beta}\|\hbeta_{\cC}\|_{\bbH}^{-2}\}  \hbeta_{\cC}\right\|\\
				&\quad +
				\left\|\min\{1, (c_{\delta} + o_{P}(c_{\delta}))(\|\hdelta_{\cC}\|_{\bbH} + o_{P}(\|\hdelta_{\cC}\|))^{-2}(\hdelta_{\cC} + o_{P}(\|\hdelta_{\cC}\|)) - \min\{1, c_{\delta}\|\hdelta_{\cC}\|_{\bbH}^{-2}\}\hdelta_{\cC}\right\|\\
				& = o_{P}\left\{\frac{c_{\beta}}{\|\hbeta_{\cC}\|_{\bbH}} + \frac{c_{\delta}}{\|\hdelta_{\cC}\|_{\bbH}}\right\}\\
				& = o_{P}\left\{\frac{\tr\{\bbSig_{\cC}\}}{\|\hbeta_{\cC}\|} + \frac{\tr\{\bbSig_{\delta}\}}{\|\hdelta_{\cC}\|}\right\}\\
				& = o_{P}(\sqrt{r_{\cT}}).
			\end{align*}
			This proves $\|\tbeta_{\rm ds} - \hbeta_{\rm ds}\| = o_{P}(\sqrt{r_{\cT}})$ and implies that $\sqrt{r_{\cT}}(\tbeta_{\rm ds} - \bbeta_{\cT})$ converges to the same limiting distribution as that of $\sqrt{r_{\cT}}(\hbeta_{\rm ds} - \bbeta_{\cT})$. Then, similarly to \eqref{eq: ds upper bound}, we have
			\begin{equation}
				\begin{aligned}
					\|\sqrt{r_{\cT}}(\tbeta_{\rm ds} - \bbeta_{\cT})\|_{\bbH}^{2+\tau}
					& \leq 3^{1+\tau}
					\left\{\|\sqrt{r_{\cT}}(\hbeta_{\cT} - \bbeta_{\cT})\|_{\bbH}^{2 + \tau} + (r_{\cT}\hat{c}_{\beta})^{1 + \frac{\tau}{2}} + (r_{\cT} \hat{c}_{\delta})^{1 + \frac{\tau}{2}}\right\}\\
					& \leq 3^{1+\tau}
					\left\{\|\bbH\|^{2 + \tau}\|\sqrt{r_{\cT}}(\hbeta_{\cT} - \bbeta_{\cT})\|^{2 + \tau} + (r_{\cT}\hat{c}_{\beta})^{1 + \frac{\tau}{2}} + (r_{\cT} \hat{c}_{\delta})^{1 + \frac{\tau}{2}}\right\}
				\end{aligned}
			\end{equation}
			Therefore, we have
			$E\{\|\sqrt{r_{\cT}}(\tbeta_{\rm ds} - \bbeta_{\cT})\|_{\bbH}^{2+\tau}\} \leq C$ for some constant $C$ according to Conditions \ref{cond: asy normal} (iii) and \ref{cond: matrices error} (ii). The conclusions of this theorem follow from Example 2.21 of \cite{van2000asymptotic} and Theorem \ref{thm: dominance ds asy}.
		\end{proof}
	
	\subsection{Proof of Theorem \ref{thm: dominance mds}}\label{app: proof mds}
	Straightforward calculations can show that 
	\[
	(\hbeta_{\cC} - \hbeta_{\rm m}, \hdelta_{\cC} - \hdelta_{\rm m}) \Perp (\hbeta_{\rm m}, \hdelta_{\rm m}).
	\]
	Then, we have 
	\[
	\begin{aligned}
		\cE(\hbeta_{\cT}, \bbeta_{\cT}) 
		& = \cE(\hbeta_{\cC}, \bbeta_{\cC}) + \cE(\hdelta_{\cC}, \bdelta_{\cC})\\
		& = \cE(\hbeta_{\rm m}, \bbPi_{\cC}\bbeta_{\cC}) + \cE(\hdelta_{\rm m}, \bbQ_{\alpha}\bbPi_{d}(\balpha_{\cT} - \balpha_{\cS})) + \cE(\hbeta_{\cC} - \hbeta_{\rm m}, (\bbI_{p} - \bbPi_{\cC})\bbeta_{\cC})\\
		&\quad + \cE(\hdelta_{\cC} - \hdelta_{\rm m}, \bbQ_{\alpha}(\bbI_{q} - \bbPi_{d})(\balpha_{\cT} - \balpha_{\cS})).
	\end{aligned}
	\] 
	Note that $\hbeta_{\rm mds} = \hbeta_{\rm m} + \widehat{\lambda}_{{\rm ms}, 1}(\hbeta_{\cC} - \hbeta_{\rm m}) + \hdelta_{\rm m} + \widehat{\lambda}_{{\rm ms}, 2}(\hdelta_{\cC} - \hdelta_{\rm m})$. We have $\cE(\hbeta_{\rm mds}, \bbeta_{\cT}) = \cE(\hbeta_{\rm m}, \bbPi_{\cC}\bbeta_{\cC}) + \cE(\hdelta_{\rm m}, \bbQ_{\alpha}\bbPi_{d}(\balpha_{\cT} - \balpha_{\cS})) + \cE(\widehat{\lambda}_{{\rm ms}, 1}(\hbeta_{\cC} - \hbeta_{\rm m}) + \widehat{\lambda}_{{\rm ms}, 2}(\hdelta_{\cC} - \hdelta_{\rm m}), (\bbI_{p} - \bbPi_{\cC})\bbeta_{\cC} + \bbQ_{\alpha}(\bbI_{q} - \bbPi_{d})(\balpha_{\cT} - \balpha_{\cS}))$. Applying the same arguments in the proof of Theorem \ref{thm: dominance ds} to $\widehat{\lambda}_{{\rm ms}, 1}(\hbeta_{\cC} - \hbeta_{\rm m}) + \widehat{\lambda}_{{\rm ms}, 2}(\hdelta_{\cC} - \hdelta_{\rm m})$, we have $\cE(\widehat{\lambda}_{{\rm ms}, 1}(\hbeta_{\cC} - \hbeta_{\rm m}) + \widehat{\lambda}_{{\rm ms}, 2}(\hdelta_{\cC} - \hdelta_{\rm m}), (\bbI_{p} - \bbPi_{\cC})\bbeta_{\cC} + \bbQ_{\alpha}(\bbI_{q} - \bbPi_{d})(\balpha_{\cT} - \balpha_{\cS})) \leq \cE(\hbeta_{\cC} - \hbeta_{\rm m}, (\bbI_{p} - \bbPi_{\cC})\bbeta_{\cC}) + \cE(\hdelta_{\cC} - \hdelta_{\rm m}, \bbQ_{\alpha}(\bbI_{q} - \bbPi_{d})(\balpha_{\cT} - \balpha_{\cS}))$ where the strict inequality holds if $c_{\beta, {\rm m}} > 0$ or $c_{\delta, {\rm m}} > 0$. Combining this with the expressions of $\cE(\hbeta_{\cT}, \bbeta_{\cT})$ and $\cE(\hbeta_{\rm mds}, \bbeta_{\cT})$ derived above, the proof of
	Theorem \ref{thm: dominance mds} is completed. 
	
	\subsection{Proof of Theorem \ref{thm: dominance cds}}\label{app: proof cds}
	Note that $\halpha_{\cS} \Perp (\hbeta_{\cT}, \halpha_{\cT})$.
	Let
	\[
	\cE_{\halpha_{\cS}}(\widehat{\bxi}, \bxi) = E\left\{(\widehat{\bxi} - \bxi)^{\T}\bbH(\widehat{\bxi} - \bxi)\mid \halpha_{\cS}\right\}
	\]
	for any $p$-dimensional random vector $\widehat{\bxi}$ and deterministic vector $\bxi$.
	Similar arguments as those in the proof of Theorem \ref{thm: dominance ds} can show the conditional dominance result $\cE_{\halpha_{\cS}}(\hbeta_{\rm ds}^{\dag},\bbeta_{\cT}) \leq \cE(\hbeta_{\cT}, \bbeta_{\cT})$ for any $\halpha_{\cS}$, which implies the first result in Theorem \ref{thm: dominance cds} according to the law of iterated expectations. If $c_{\beta}^{\dag} > 0$ or $c_{\delta}^{\dag} > 0$, we have $\cE_{\halpha_{\cS}}(\hbeta_{\rm ds}^{\dag},\bbeta_{\cT}) < \cE(\hbeta_{\cT}, \bbeta_{\cT})$ and hence $\cE(\hbeta_{\rm ds}^{\dag},\bbeta_{\cT}) < \cE(\hbeta_{\cT}, \bbeta_{\cT})$ according to arguments similar to the proof of Theorem \ref{thm: dominance ds}, which completes the proof of Theorem \ref{thm: dominance cds}.

		\subsection{Proof of \eqref{eq: compr ds&cds} in the main text}\label{app: compr ds&cds}
		\begin{proof}
			Recall that $\bbH = \bbI_{p}$, $\hbeta_{\cT} = \halpha_{\cT} \sim N(\bzero, 2n_{\cT}^{-1}\bbI_{p}),\  \hbeta_{\cS} = \halpha_{\cS} \sim N(\bzero, 2n_{\cS}^{-1}\bbI_{p})$ in this example.
			Then, $\hbeta_{\cC} = (n_{\cT}\halpha_{\cT} + n_{\cS}\halpha_{\cS}) / (n_{\cT} + n_{\cS})$, $\hdelta_{\cC} = n_{\cS}(\halpha_{\cT} - \halpha_{\cS}) / (n_{\cT} + n_{\cS})$,  $\bbeta_{\cC} = \bzero$, and $\bdelta_{\cC} = \bzero$. Thus, $\var(\hbeta_{\cC}) = 2/(n_{\cT} + n_{\cS})\bbI_{p} = v_{\cC}^{2}\bbI_{p}$, $\var(\hdelta_{\cC}) = 2n_{\cS}/\{n_{\cT}(n_{\cT} + n_{\cS})\}\bbI_{p} = v_{\delta}^{2}\bbI_{p}$, $v_{\cC}^{2} + v_{\delta}^{2} = 2n_{\cT}^{-1}$, $c_{\beta} = (p - 2)v_{\cC}^{2}$, and $c_{\delta} = (p - 2)v_{\delta}^{2}$. Then, the $\cE$-error of the dShrink estimator $\hbeta_{\rm ds}$ can be decomposed as  
			\begin{align*}	
				\cE(\hbeta_{\rm ds}, \bbeta_{\cT}) &= \cE(\widehat{\lambda}_{{\rm s}, 1}\hbeta_{\cC} +  \widehat{\lambda}_{{\rm s}, 2}\hdelta_{\cC}, \bzero)\notag\\
				&=E\left(\widehat{\lambda}_{{\rm s}, 1}^{2}\|\hbeta_{\cC}\|^{2}\right) + E\left(\widehat{\lambda}_{{\rm s}, 2}^{2}\|\hdelta_{\cC}\|^{2}\right) + 2E\left\{(\widehat{\lambda}_{{\rm s}, 1}\hbeta_{\cC})^{\T}(\widehat{\lambda}_{{\rm s}, 2}\hdelta_{\cC})\right\}\notag\\
				&=E\left(\widehat{\lambda}_{{\rm s}, 1}^{2}\|\hbeta_{\cC}\|^{2}\right) + E\left(\widehat{\lambda}_{{\rm s}, 2}^{2}\|\hdelta_{\cC}\|^{2}\right)\\
				&=E\left(\|\hbeta_{\cC}\|^{2}\right) -  2E\left[\min\{1, (p - 2)v_{\cC}^{2}\|\hbeta_{\cC}\|^{-2}\}\|\hbeta_{\cC}\|^{2}\right] + E\left[\min\{1, (p - 2)v_{\cC}^{2}\|\hbeta_{\cC}\|^{-2}\}^{2}\|\hbeta_{\cC}\|^{2}\right]\\
				&\quad + E\left(\|\hdelta_{\cC}\|^{2}\right) -  2E\left[\min\{1, (p - 2)v_{\delta}^{2}\|\hdelta_{\cC}\|^{-2}\}\|\hdelta_{\cC}\|^{2}\right] + E\left[\min\{1, (p - 2)v_{\delta}^{2}\|\hdelta_{\cC}\|^{-2}\}^{2}\|\hdelta_{\cC}\|^{2}\right]\\
				& = 2n_{\cT}^{-1}p - E\left[1\left\{\|\hbeta_{\cC}\|^{2}\geq (p - 2)v_{\cC}^{2}\right\}(p - 2)v_{\cC}^{2}\{2 - (p - 2)v_{\cC}^{2}\|\hbeta_{\cC}\|^{-2}\}\right]\\
				&\quad - E\left[1\left\{\|\hbeta_{\cC}\|^{2}< (p - 2)v_{\cC}^{2}\right\}\|\hbeta_{\cC}\|^{2}\right]\\
				&\quad - E\left[1\left\{\|\hdelta_{\cC}\|^{2}\geq (p - 2)v_{\delta}^{2}\right\}(p - 2)v_{\delta}^{2}\{2 - (p - 2)v_{\delta}^{2}\|\hdelta_{\cC}\|^{-2}\}\right]\\
				&\quad - E\left[1\left\{\|\hdelta_{\cC}\|^{2}< (p - 2)v_{\delta}^{2}\right\}\|\hdelta_{\cC}\|^{2}\right].
			\end{align*}
			Let $\pi_{\cS} = (n_{\cT} + n_{\cS})^{-1}n_{\cS}$. The distributions of $\hbeta_{\cC}$ and $\sqrt{(1 - \pi_{\cS})/\pi_{\cS}}\hdelta_{\cC}$ are identical. Thus,
			\begin{align*}	
				\cE(\hbeta_{\rm ds}, \bbeta_{\cT}) 
				& = 2n_{\cT}^{-1}p - E\left[1\left\{\|\hdelta_{\cC}\|^{2}\geq (p - 2)v_{\delta}^{2}\right\}(p - 2)(\pi_{\cS}^{-1} - 1)v_{\delta}^{2}\{2 - (p - 2)v_{\delta}^{2}\|\hdelta_{\cC}\|^{-2}\}\right]\\
				&\quad -(\pi_{\cS}^{-1} - 1) E\left[1\left\{\|\hdelta_{\cC}\|^{2}< (p - 2)v_{\delta}^{2}\right\}\|\hdelta_{\cC}\|^{2}\right]\\
				&\quad - E\left[1\left\{\|\hdelta_{\cC}\|^{2}\geq (p - 2)v_{\delta}^{2}\right\}(p - 2)v_{\delta}^{2}\{2 - (p - 2)v_{\delta}^{2}\|\hdelta_{\cC}\|^{-2}\}\right]\\
				&\quad - E\left[1\left\{\|\hdelta_{\cC}\|^{2}< (p - 2)v_{\delta}^{2}\right\}\|\hdelta_{\cC}\|^{2}\right]\\
				& = 2n_{\cT}^{-1}p - \pi_{\cS}^{-1}E\left[1\left\{\|\hdelta_{\cC}\|^{2}\geq (p - 2)v_{\delta}^{2}\right\}(p - 2)v_{\delta}^{2}\{2 - (p - 2)v_{\delta}^{2}\|\hdelta_{\cC}\|^{-2}\}\right]\\
				&\quad -\pi_{\cS}^{-1}E\left[1\left\{\|\hdelta_{\cC}\|^{2}< (p - 2)v_{\delta}^{2}\right\}\|\hdelta_{\cC}\|^{2}\right].
			\end{align*}
			On the other hand, we have $\hbeta_{\cC}^{\dag} = \halpha_{\cS} = \hbeta_{\cC} - (\pi_{\cS}^{-1} - 1)\hdelta_{\cC}$ and $\hdelta_{\cC}^{\dag} = \halpha_{\cT} - \halpha_{\cS} = \pi_{\cS}^{-1}\hdelta_{\cC}$. Thus, $\var(\hbeta_{\cC}^{\dag}\mid \halpha_{\cS}) = \mathbf{0}$, $\var(\hdelta_{\cC}^{\dag}\mid \halpha_{\cS}^{\dag}) = 2n_{\cT}^{-1}\bbI_{p}$, $c_{\beta}^{\dag} = 0$, and $c_{\delta}^{\dag} = 2(p - 2)n_{\cT}^{-1} = (p - 2)\pi_{\cS}^{-1}v_{\delta}^{2}$. Then, similar to the above decomposition for the dShrink estimator, the $\cE$-error of the c-dShrink estimator $\hbeta_{\rm ds}^{\dag}$ can be decomposed as
			\begin{align*}	
				\cE(\hbeta_{\rm ds}^{\dag}, \bbeta_{\cT})
				& = 2n_{\cT}^{-1}p - E\left[1\left\{\|\hdelta_{\cC}^{\dag}\|^{2}\geq c_{\delta}^{\dag}\right\}c_{\delta}^{\dag}\{2 - c_{\delta}^{\dag}\|\hdelta_{\cC}^{\dag}\|^{-2}\}\right] - E\left[1\left\{\|\hdelta_{\cC}^{\dag}\|^{2}< c_{\delta}^{\dag}\right\}\|\hdelta_{\cC}^{\dag}\|^{2}\right]\\
				&\quad - 2E\left[1\left\{\|\hdelta_{\cC}^{\dag}\|^{2}\geq c_{\delta}^{\dag}\right\}c_{\delta}^{\dag}\|\hdelta_{\cC}^{\dag}\|^{-2}\hdelta_{\cC}^{\dag\T}\hbeta_{\cC}^{\dag}\right] - 2E\left[1\left\{\|\hdelta_{\cC}^{\dag}\|^{2}< c_{\delta}^{\dag}\right\}\hdelta_{\cC}^{\dag\T}\hbeta_{\cC}^{\dag}\right]\\
				& = 2n_{\cT}^{-1}p - E\left[1\left\{\|\hdelta_{\cC}\|^{2}\geq (p - 2)\pi_{\cS}v_{\delta}^{2}\right\}(p - 2)\pi_{\cS}^{-1}v_{\delta}^{2}\{2 - (p - 2)\pi_{\cS}v_{\delta}^{2}\|\hdelta_{\cC}\|^{-2}\}\right]\\
				&\quad - E\left[1\left\{\|\hdelta_{\cC}\|^{2}< (p - 2)\pi_{\cS}v_{\delta}^{2}\right\}\pi_{\cS}^{-2}\|\hdelta_{\cC}\|^{2}\right]\\
				&\quad + 2(\pi_{\cS}^{-1} - 1)E\left[1\left\{\|\hdelta_{\cC}\|^{2}\geq (p - 2)\pi_{\cS}v_{\delta}^{2}\right\}(p - 2)v_{\delta}^{2}\right]\\
				&\quad + 2(\pi_{\cS}^{-2} - \pi_{\cS}^{-1})E\left[1\left\{\|\hdelta_{\cC}\|^{2}< (p - 2)\pi_{\cS}v_{\delta}^{2}\right\}\|\hdelta_{\cC}\|^{2}\right]\\
				& = 2n_{\cT}^{-1}p - E\left[1\left\{\|\hdelta_{\cC}\|^{2}\geq (p - 2)\pi_{\cS}v_{\delta}^{2}\right\}(p - 2)v_{\delta}^{2}\{2 - (p - 2)v_{\delta}^{2}\|\hdelta_{\cC}\|^{-2}\}\right]\\
				&\quad + (\pi_{\cS}^{-2} - 2\pi_{\cS}^{-1})E\left[1\left\{\|\hdelta_{\cC}\|^{2}< (p - 2)\pi_{\cS}v_{\delta}^{2}\right\}\|\hdelta_{\cC}\|^{2}\right]\\
				& = 2n_{\cT}^{-1}p - E\left[1\left\{\|\hdelta_{\cC}\|^{2}\geq (p - 2)\pi_{\cS}v_{\delta}^{2}\right\}(p - 2)v_{\delta}^{2}\{2 - (p - 2)v_{\delta}^{2}\|\hdelta_{\cC}\|^{-2}\}\right]\\
				&\quad - E\left[1\left\{\|\hdelta_{\cC}\|^{2}< (p - 2)\pi_{\cS}v_{\delta}^{2}\right\}\|\hdelta_{\cC}\|^{2}\right] + (\pi_{\cS}^{-1} - 1)^{2}E\left[1\left\{\|\hdelta_{\cC}\|^{2}< (p - 2)\pi_{\cS}v_{\delta}^{2}\right\}\|\hdelta_{\cC}\|^{2}\right].
			\end{align*}
			Thus,
			\begin{align*}	
				&\cE(\hbeta_{\rm ds}^{\dag}, \bbeta_{\cT}) - \cE(\hbeta_{\rm ds}, \bbeta_{\cT})\\
				& = (\pi_{\cS}^{-1} - 1)^{2}E\left[1\left\{\|\hdelta_{\cC}\|^{2}< (p - 2)\pi_{\cS}v_{\delta}^{2}\right\}\|\hdelta_{\cC}\|^{2}\right]\\
				&\quad - E\left[1\left\{(p - 2)v_{\delta}^{2}\|\hdelta_{\cC}\|^{-2}\in (1, \pi_{\cS}^{-1}]\right\}(p - 2)v_{\delta}^{2}\{2 - (p - 2)v_{\delta}^{2}\|\hdelta_{\cC}\|^{-2}\}\right]\\
				& \quad + E\left[1\left\{(p - 2)v_{\delta}^{2}\|\hdelta_{\cC}\|^{-2}\in (1, \pi_{\cS}^{-1}]\right\}\|\hdelta_{\cC}\|^{2}\right]\\
				&\quad + (\pi_{\cS}^{-1} - 1)E\left[1\left\{\|\hdelta_{\cC}\|^{2}\geq (p - 2)v_{\delta}^{2}\right\}(p - 2)v_{\delta}^{2}\{2 - (p - 2)v_{\delta}^{2}\|\hdelta_{\cC}\|^{-2}\}\right]\\
				&\quad
				+ (\pi_{\cS}^{-1} - 1)E\left[1\left\{\|\hdelta_{\cC}\|^{2}< (p - 2)v_{\delta}^{2}\right\}\|\hdelta_{\cC}\|^{2}\right]\\
				& = (\pi_{\cS}^{-1} - 1)^{2}E\left[1\left\{(p - 2)v_{\delta}^{2}\|\hdelta_{\cC}\|^{-2} > \pi_{\cS}^{-1}\right\}\|\hdelta_{\cC}\|^{2}\right]\\
				& \quad + E\left[1\left\{(p - 2)v_{\delta}^{2}\|\hdelta_{\cC}\|^{-2}\in (1, \pi_{\cS}^{-1}]\right\}\|\hdelta_{\cC}\|^{2}\left\{(p - 2)v_{\delta}^{2}\|\hdelta_{\cC}\|^{-2} - 1\right\}^{2}\right]\\
				&\quad + (\pi_{\cS}^{-1} - 1)E\left[1\left\{(p - 2)v_{\delta}^{2}\|\hdelta_{\cC}\|^{-2}\leq 1\right\}(p - 2)v_{\delta}^{2}\{2 - (p - 2)v_{\delta}^{2}\|\hdelta_{\cC}\|^{-2}\}\right]\\
				&\quad
				+ (\pi_{\cS}^{-1} - 1)E\left[1\left\{(p - 2)v_{\delta}^{2}\|\hdelta_{\cC}\|^{-2} > 1\right\}\|\hdelta_{\cC}\|^{2}\right]\\
				& = E\left(1\left\{(p - 2)v_{\delta}^{2}\|\hdelta_{\cC}\|^{-2} > 1\right\}\|\hdelta_{\cC}\|^{2}\left[\min\{\pi_{\cS}^{-1}, (p - 2)v_{\delta}^{2}\|\hdelta_{\cC}\|^{-2}\} - 1\right]^{2}\right)\\
				&\quad + (\pi_{\cS}^{-1} - 1)E\left[1\left\{(p - 2)v_{\delta}^{2}\|\hdelta_{\cC}\|^{-2}\leq 1\right\}(p - 2)v_{\delta}^{2}\{2 - (p - 2)v_{\delta}^{2}\|\hdelta_{\cC}\|^{-2}\}\right]\\
				&\quad
				+ (\pi_{\cS}^{-1} - 1)E\left[1\left\{(p - 2)v_{\delta}^{2}\|\hdelta_{\cC}\|^{-2} > 1\right\}\|\hdelta_{\cC}\|^{2}\right].
			\end{align*}
		\end{proof}
		\subsection{Hardness of estimating $(\lambda_{{\rm opt},1}, \lambda_{{\rm opt},2})$ and $(\lambda_{{\rm s},1}, \lambda_{{\rm s},2})$}\label{app: hardness}
		The mean square consistent estimation of $(\lambda_{{\rm opt},1}, \lambda_{{\rm opt},2})$ and $(\lambda_{{\rm s},1}, \lambda_{{\rm s},2})$ are not possible in general. For illustration, we show that the minimax lower bound does not converge to zero as $n_{\cT}, n_{\cS} \to \infty$ even under the simple example \eqref{eq: counter eg} in the main text. For simplicity, we assume $n_{\cT} = n_{\cS}$. Then, we have the following proposition.
		\begin{proposition}\label{prop: hardness of coefficient estimation}
			Under the example \eqref{eq: counter eg} in the main text with $n_{\cT} = n_{\cS}$. Let $\Theta = \{(\boldsymbol{\mu}, \bbias): \|\boldsymbol{\mu}\|^{2} \leq C, \|\bbias\|^{2} \leq C\}$ for some $C > 0$ be the parameter space. Suppose $\bbH = \bbI_{p}$. Then, we have
			\[
			\inf_{(\check{\lambda}_{1}, \check{\lambda}_{2})}\sup_{\boldsymbol{\mu},\bbias}E\left[\left\|(\check{\lambda}_{1}, \check{\lambda}_{2}) - (\lambda_{{\rm opt},1}, \lambda_{{\rm opt},2})\right\|^{2}\right] \geq
			\frac{1}{8}\left(\frac{c}{c + p}\right)^{2}
			\]
			and
			\[
			\inf_{(\check{\lambda}_{1}, \check{\lambda}_{2})}\sup_{\boldsymbol{\mu},\bbias}E\left[\left\|(\check{\lambda}_{1}, \check{\lambda}_{2}) - (\lambda_{{\rm s},1}, \lambda_{{\rm s},2})\right\|^{2}\right] \geq
			\frac{1}{8}\left(\frac{c}{c + p}\right)^{2},
			\]
			where $c = \min\{C, 1/2\}$ and the infimums are taken over all possible estimators.
		\end{proposition}
		\begin{proof}
			For any $\boldsymbol{\mu}$, $\bbias$, straightforward calculations can show that
			$\bbeta_{\cC} = \boldsymbol{\mu}$, $\bdelta_{\cC} = \bbias$, $\tr\{\bbH\bbSig_{\cC}\} = \tr\{\bbH\bbSig_{\delta}\} = pn_{\cT}^{-1}$,
			\[
			\begin{aligned}
				(\lambda_{{\rm opt},1}, \lambda_{{\rm opt},2}) 
				& = (1, 1) - \{\|\boldsymbol{\mu}\|^{2}\|\bbias\|^{2} + pn_{\cT}^{-1}(\|\boldsymbol{\mu}\|^{2} + \|\bbias\|^{2}) + p^{2}n_{\cT}^{-2} - (\boldsymbol{\mu}^{\T}\bbias)^{2}\}^{-1}\\
				&\quad\times(pn_{\cT}^{-1}(\|\bbias\|^{2} + pn_{\cT}^{-1} - \boldsymbol{\mu}^{\T}\bbias), pn_{\cT}^{-1}(\|\boldsymbol{\mu}\|^{2} + pn_{\cT}^{-1} - \boldsymbol{\mu}^{\T}\bbias)),
			\end{aligned}
			\]
			and
			\[
			\begin{aligned}
				(\lambda_{{\rm s},1}, \lambda_{{\rm s},2}) 
				& = (1, 1) - \left(\frac{pn_{\cT}^{-1}}{pn_{\cT}^{-1} + \|\boldsymbol{\mu}\|^{2}}, \frac{pn_{\cT}^{-1}}{pn_{\cT}^{-1} + \|\bbias\|^{2}}\right).
			\end{aligned}
			\]
			Take a arbitrary $p$-dimensional vector $\ba$ such that $\|\ba\|^{2} = n_{\cT}^{-1}c$. Let $\boldsymbol{\mu}_{0} = \ba$, $\bbias_{0} = \bzero$, $\boldsymbol{\mu}_{1} = \bzero$, and $\bbias_{1} = \ba$. Let $P_{0}$ and $P_{1}$ be the joint distributions of $(\hbeta_{\cT}, \hbeta_{\cS})$ under the parameters $(\boldsymbol{\mu}_{0}, \bbias_{0})$ and $(\boldsymbol{\mu}_{1}, \bbias_{1})$, respectively. For $t = 0, 1$, define $(\lambda_{{\rm opt},1}^{(t)}, \lambda_{{\rm opt},2}^{(t)})$ and $(\lambda_{{\rm s},1}^{(t)}, \lambda_{{\rm s},2}^{(t)})$ as the value of $(\lambda_{{\rm opt},1}, \lambda_{{\rm opt},2})$ and $(\lambda_{{\rm s},1}, \lambda_{{\rm s},2})$ under $P_{t}$, respectively. Then, we have
			\[
			\big\|(\lambda_{{\rm opt},1}^{(0)}, \lambda_{{\rm opt},2}^{(0)}) - (\lambda_{{\rm opt},1}^{(1)}, \lambda_{{\rm opt},2}^{(1)})\big\|^{2} = 2c^{2}(c + p)^{-2}
			\]
			and
			\[
			\big\|(\lambda_{{\rm s},1}^{(0)}, \lambda_{{\rm s},2}^{(0)}) - (\lambda_{{\rm s},1}^{(1)}, \lambda_{{\rm s},2}^{(1)})\big\|^{2} = 2c^{2}(c + p)^{-2}.
			\]
			Moreover, we have $d_{\rm TV}(P_{0}, P_{1}) \leq \sqrt{c / 2} \leq 1/2$ according to Pinsker's inequality and the expression for the KL-divergence between multivariate normal distributions, where $d_{\rm TV}(\cdot, \cdot)$ is the total variation distance between two distributions. Then, the conclusion of this proposition by applying equation (15.14) in \cite{wainwright2019high} with $\delta = 2^{-1/2}c/(c + p)$, $\rho(\ba_{1}, \ba_{2}) = \|\ba_{1} - \ba_{2}\|$, and $\Phi(a) = a^{2}$ for any two-dimensional vectors $\ba_{1}$, $\ba_{2}$, and scalar $a$.
		\end{proof}
	
	\section{Additional simulation results}
		\subsection{MSEs with respect to $(\lambda_{{\rm opt}, 1}, \lambda_{{\rm opt}, 2})$ of combination coefficient estimators in example  \eqref{eq: counter eg} in the main text}
		Figure \ref{fig: eg-weight2} presents the MSEs with respect to $(\lambda_{{\rm opt}, 1}, \lambda_{{\rm opt}, 2})$ of combination coefficient estimators in example  \eqref{eq: counter eg} in the main text. In Figure \ref{fig: eg-weight2}, adopting the surrogate loss can reduce the MSE for estimating $(\lambda_{{\rm opt}, 1}, \lambda_{{\rm opt}, 2})$ in most cases.
		\begin{figure}
			\centering
			\includegraphics[scale = 0.4]{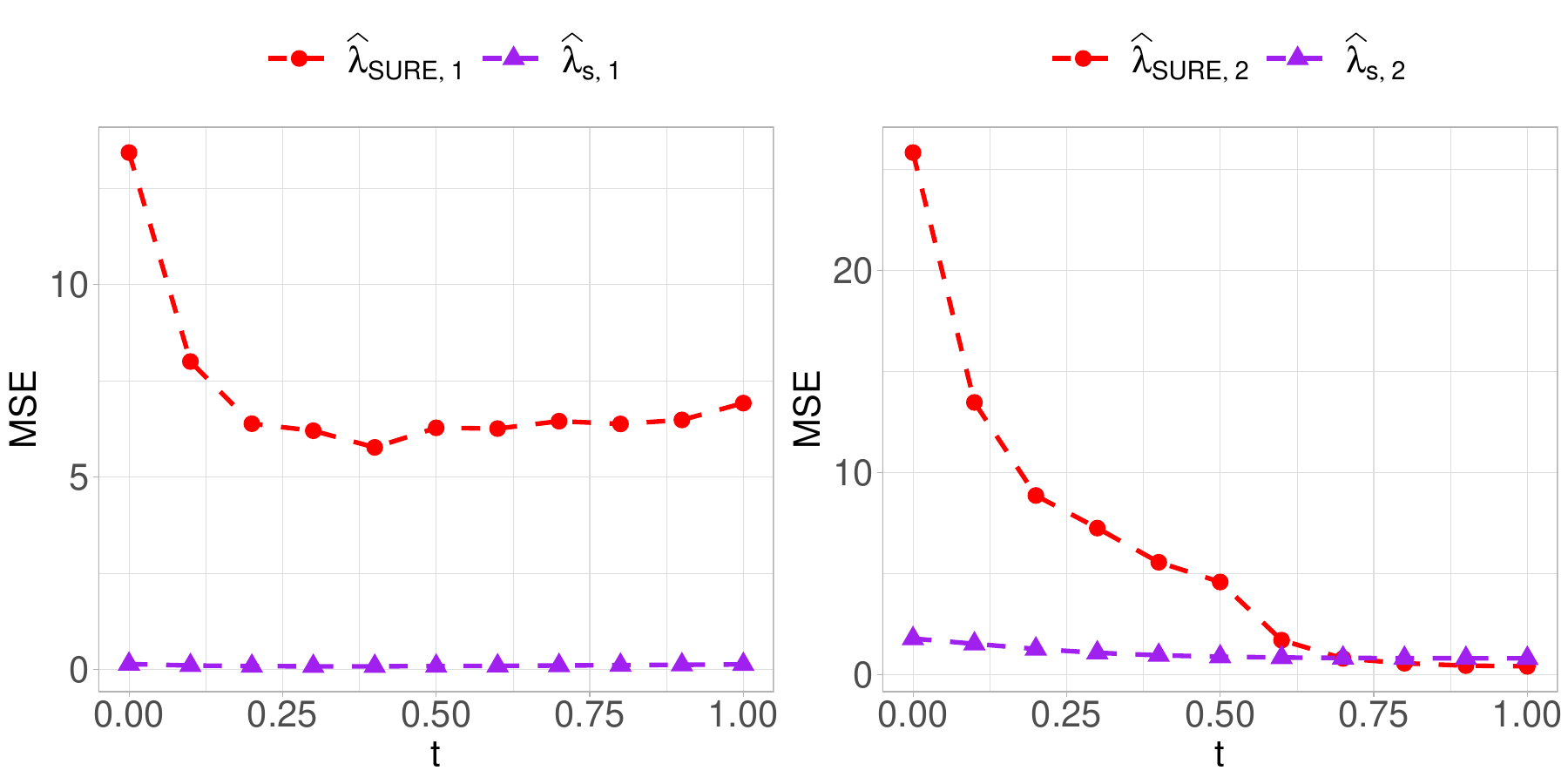}
			\caption{\it MSEs of $\widehat{\lambda}_{{\rm s}, 1}$, $\widehat{\lambda}_{{\rm s}, 2}$ $\widehat{\lambda}_{{\SURE}, 1}$, and $\widehat{\lambda}_{{\SURE}, 2}$ defined in \eqref{eq: lambda ds} and \eqref{eq: lambda sure} with respect to $\lambda_{{\rm opt}, 1}$ and $\lambda_{{\rm opt}, 2}$  in the simulation under the setting of \eqref{eq: counter eg} with $\bbH = \bbI_{p}$, $p = 5$, $n_{\cT} = n_{\cS} = 300$,  $\bbeta_{\cT} = \balpha_{\cT} = \boldsymbol{\mu} + \bbias,$  $\balpha_{\cS} = \boldsymbol{\mu} - \bbias,$ $\boldsymbol{\mu} = (0.05, 0.02, 0.1, 0.1, 0.1) ^{\T} / \sqrt{5}$, 
				$\bbias = t\etab$, $\etab = (0.2, 0.3, 0.3, 0.3, 0.3)^{\T} / \sqrt{5}$ and different $t$'s.}\label{fig: eg-weight2}
		\end{figure}
	
	\subsection{Numerical results regarding the ridge-SURE estimator}\label{app: sim ridge SURE}
	In this section, we evaluate the ridge-SURE estimator $\hbeta_{\SURE}^{(r)}$ introduced in Section \ref{subsec: SURE}  under various tuning parameters $r$.  We conduct a simulation study following the setting described in \eqref{eq: counter eg} in the main text. Figure \ref{fig: eg-ridge-SURE} presents the MSEs of different estimators under different heterogeneity levels based on $5000$ simulation runs.
	\begin{figure}
		\centering
		\includegraphics[scale = 0.25]{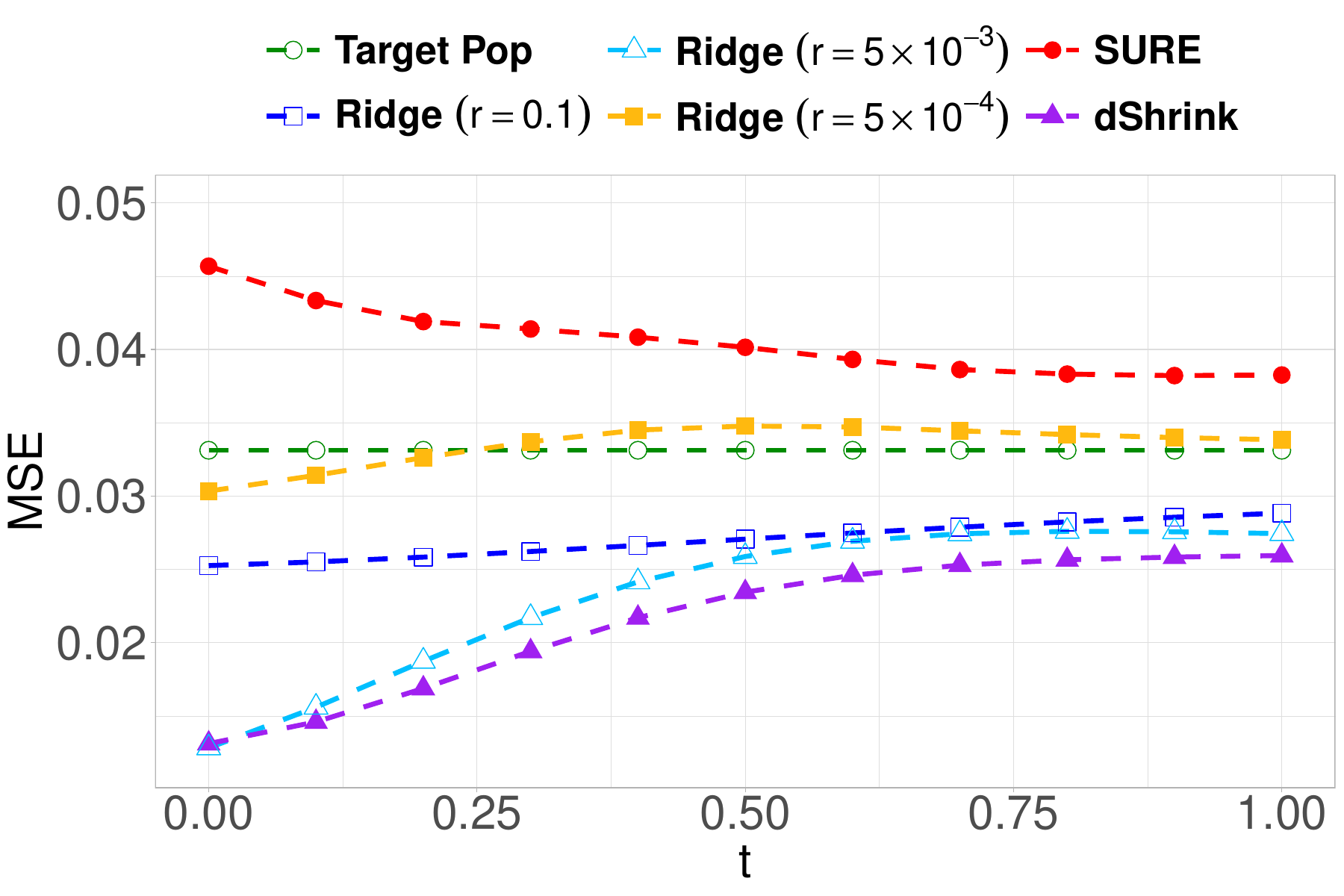}
		\caption{\it MSEs of $\hbeta_{\cT}$, $\hbeta_{\rm ds}$, $\hbeta_{\SURE}$, and $\hbeta_{\SURE}^{(r)}$ with different $r$'s in the simulation under the setting of \eqref{eq: counter eg} with $\bbH = \bbI_{p}$, $p = 5$, $n_{\cT} = n_{\cS} = 300$, $\boldsymbol{\mu} = (0.05, 0.02, 0.1, 0.1, 0.1) ^{\T} / \sqrt{5}$, $\bbias = t\etab$, $\etab = (0.2, 0.3, 0.3, 0.3, 0.3)^{\T} / \sqrt{5}$ and different $t$'s.}\label{fig: eg-ridge-SURE}
	\end{figure}
	As shown in Figure \ref{fig: eg-ridge-SURE}, the performance of $\hbeta_{\SURE}^{(r)}$ is sensitive to the choice of the tuning parameter $r$. In particular, $\hbeta_{\SURE}^{(r)}$  is closely match $\hbeta_{\rm ds}$ when $r = 5\times 10^{-3}$ but performs worse than $\hbeta_{\rm ds}$ when $r = 0.1$ or $5\times 10^{-4}$. Notably, when $r = 5 \times 10^{-4}$, $\hbeta_{\SURE}^{(r)}$ may yield a larger error than $\hbeta_{\cT}$, indicating a negative transfer problem.

		\subsection{Bias and variance comparison in the simulation in Section \ref{subsec: sim} in the main text}
		We present in Figure \ref{fig: sim bias var} the bias and variance of the dShrink estimator with $\bbH$ being an estimate of $E[(1 \ \bX^{\T})^{\T}(1 \ \bX^{\T})]$ under the simulation setting in Section \ref{subsec: sim} in the main text. Bias and variance of some other estimators ($\hbeta_{\cT}$, $\hbeta_{\cC}$, $\hbeta_{\rm prt}$, $\hbeta_{\SURE}$) are plotted for reference. As illustrated in Figure \ref{fig: sim bias var}, the target population-based estimator is nearly unbiased, dShrink achieves a marked reduction in variance relative to the target population-based estimator at the cost of a mild increase in bias.
		\begin{figure}[h]
			\centering
			\subfigure[Squared sum of the bias across different parameter components. $p_{X} = 10$, $n_{\cT} = 100$, $n_{\cS} = 500$]{
				\includegraphics[scale = 0.25]{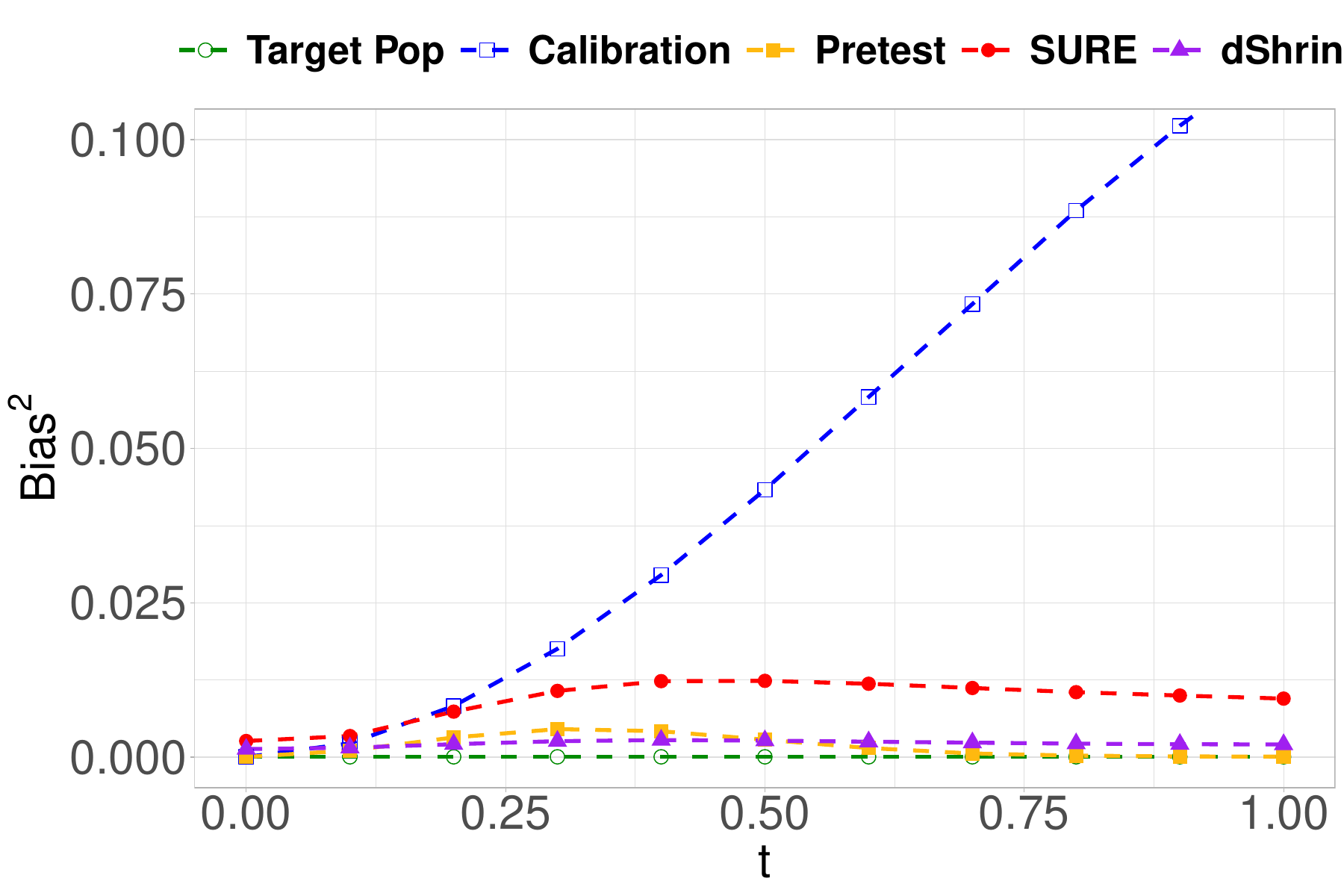}
			}
			\subfigure[Sum of the variance across different parameter components, $p_{X} = 10$, $n_{\cT} = 100$, $n_{\cS} = 500$]{
				\includegraphics[scale = 0.25]{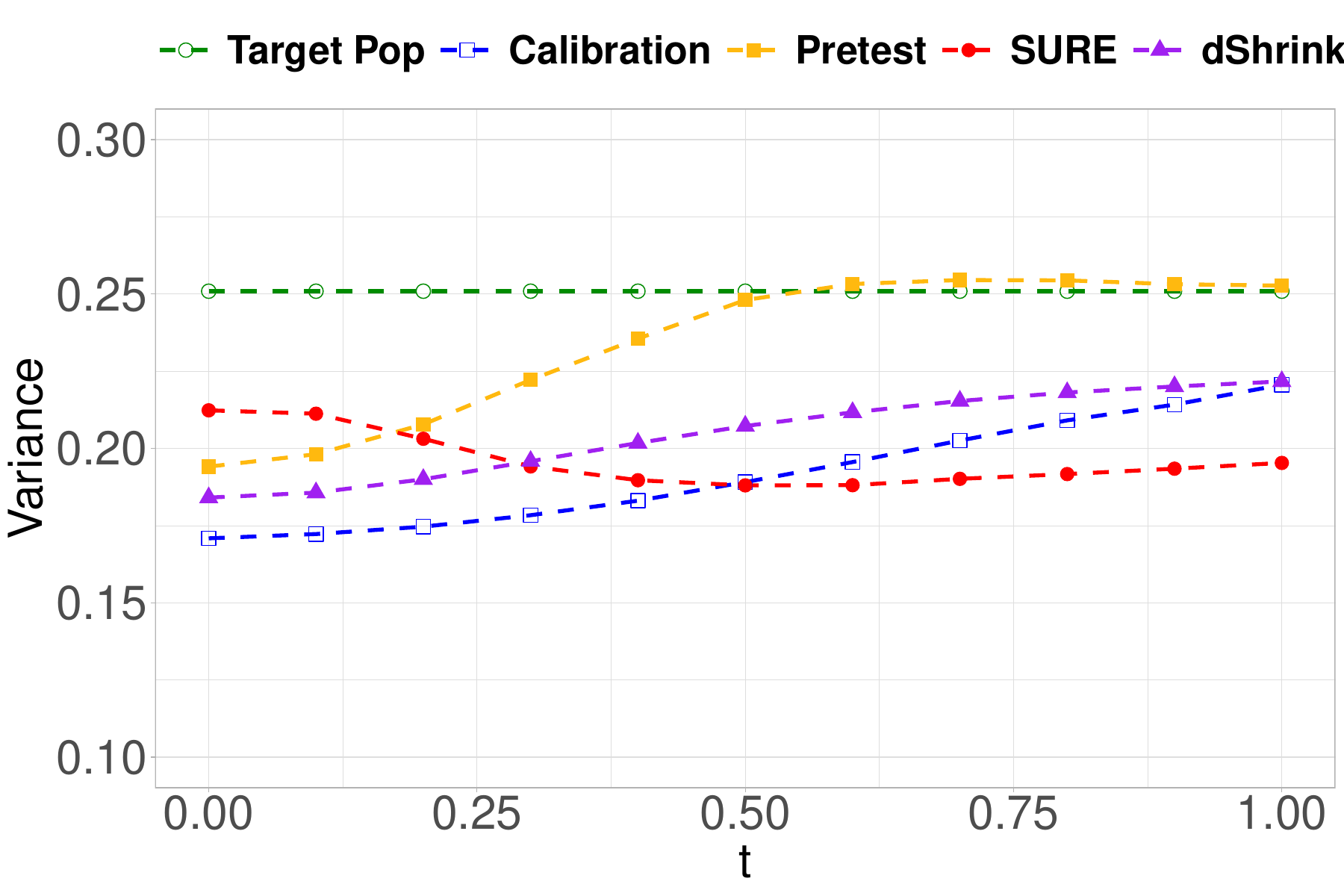}
			}
			\subfigure[Squared sum of the bias across different parameter components. $p_{X} = 30$, $n_{\cT} = 300$, $n_{\cS} = 1000$]{
				\includegraphics[scale = 0.25]{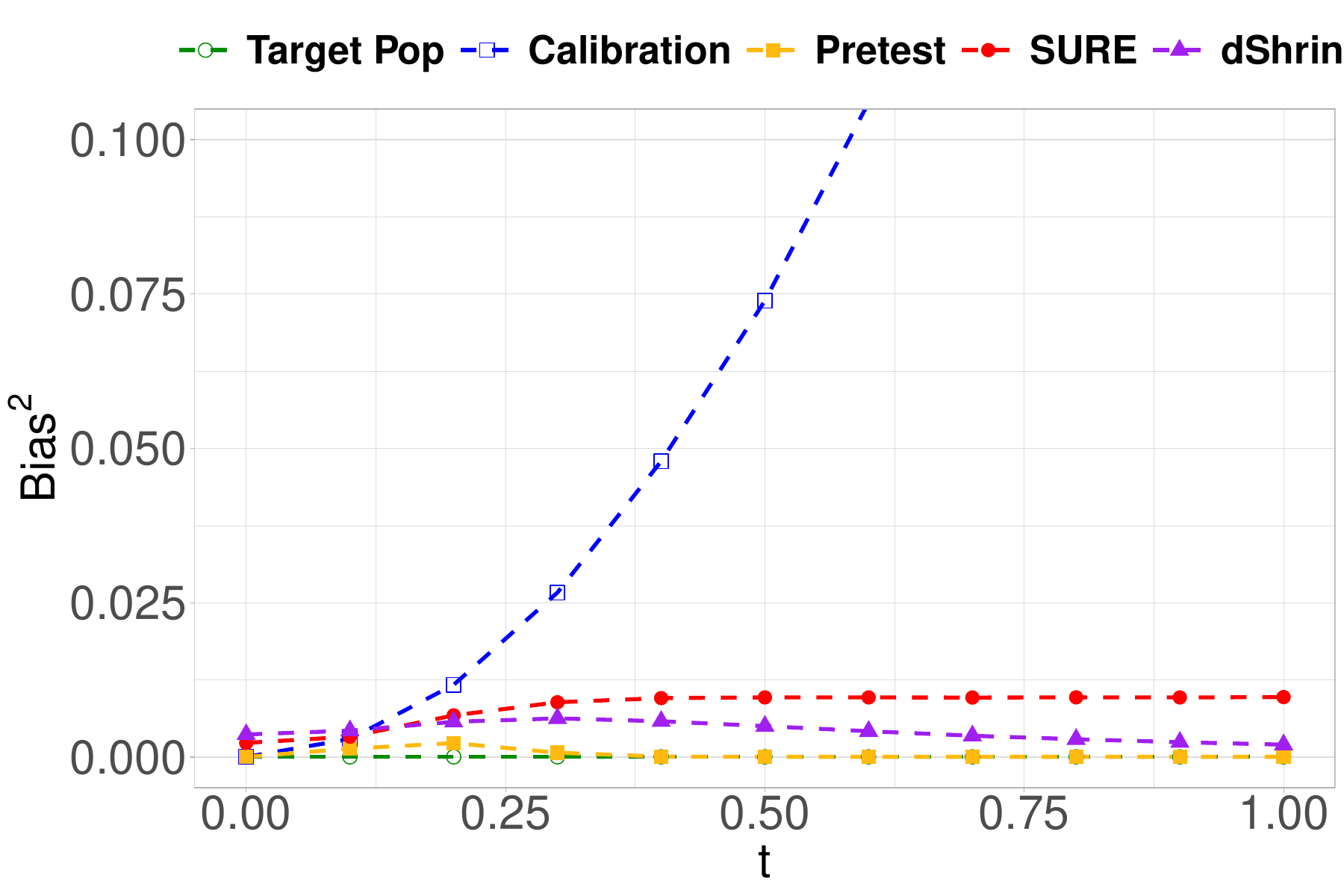}
			}
			\subfigure[Sum of the variance across different parameter components, $p_{X} = 30$, $n_{\cT} = 300$, $n_{\cS} = 1000$]{
				\includegraphics[scale = 0.25]{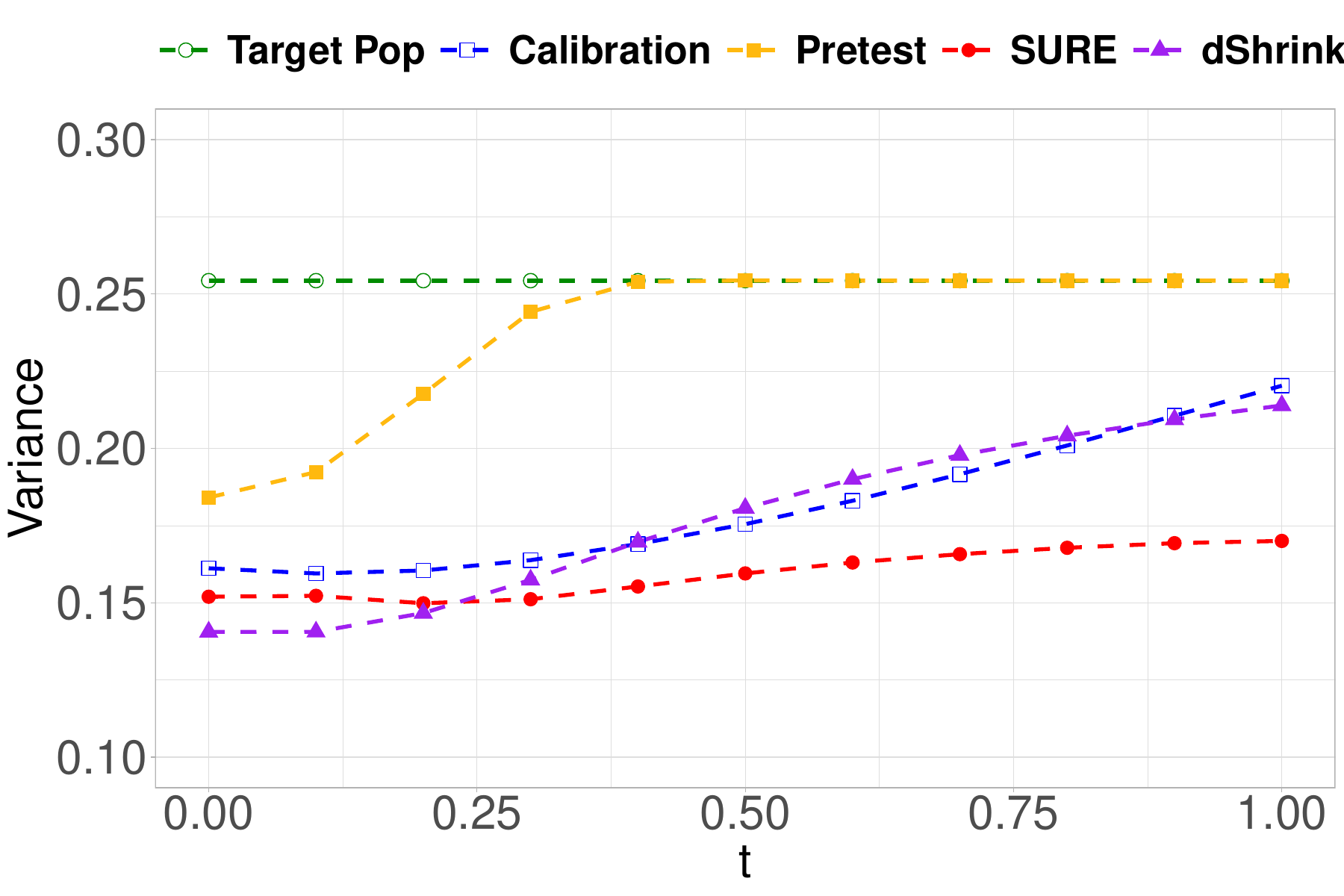}
			}
			\caption{\it Squared sum of the bias and sum of the variance across different parameter components of $\hbeta_{\cT}$, $\hbeta_{\cC}$, $\hbeta_{\rm prt}$, $\hbeta_{\SURE}$, and $\hbeta_{\rm ds}$ in the simulation studies in Section \ref{subsec: sim} in the main text.}\label{fig: sim bias var}
		\end{figure}
		
		\subsection{Simulation results based on the real data in Section \ref{subsec: ITE} in the main text}
		We conducted simulation studies based on the real data in Section \ref{subsec: ITE} by viewing the estimation in the real data as the true value and sampled with replacement from the target population data to obtain the simulated data sets. The simulation was repeated for $5000$ times.
		The target population-based estimator $\hbeta_{\cT}$ and the external summary statistics $\halpha_{\cS}$ were defined in the same way as those in Section \ref{subsec: ITE}.
		We applied different methods to estimate the target parameter. We took $\bbH$ to be an estimate of $E_{\cT}[\partial m(\bX^{\T}\bbeta_{\cT})/\partial \bbeta \{\partial m(\bX^{\T}\bbeta_{\cT})/\partial \bbeta\}^{\T}]$ in the implementation of the dShrink method, where $m(\cdot)$ is the logistic function. This choice of $\bbH$ reflects the first-order approximation of the MSE of the individualized treatment effect function. 
		We calculated the MSEs of the individualized treatment effect function based on different estimators. MSEs (multiplied by $n_{\cT}$) based on $\hbeta_{\cT}$ and the c-dShrink method are $11.09$ and $6.83$, respectively. The c-dShrink method achieves a significant reduction (more than $38\%$) in MSE compared to $\hbeta_{\cT}$. We constructed $0.95$-CI for the individualized treatment effect for each patients.  We took $\bbH$ as an estimate of $\bbSig_{\beta, \cT}^{-1}$ in c-dShrink when constructing CIs. The average width of CIs based on $\hbeta_{\cT}$ and c-dShrink are $0.48$ and $0.38$, respectively, while the corresponding average coverage rates are $0.95$ and $0.94$, respectively. These results demonstrate that c-dShrink achieves a substantial reduction in CI width compared to $\hbeta_{\cT}$ while maintaining near-nominal coverage. Figure \ref{fig: ITE bias var} presents the boxplot of the estimated individualized treatment effects for different individuals (c-dShrink implemented with $\bbH$ being an estimate of $\bbSig_{\beta, \cT}^{-1}$). Figure \ref{fig: ITE bias var} shows that the target population-based estimator is nearly unbiased, and c-dShrink achieves a significant reduction in variance relative to the target population-based estimator at the cost of a mild increase in bias.
		
		\begin{figure}[h]
			\centering
			\includegraphics[scale = 0.3]{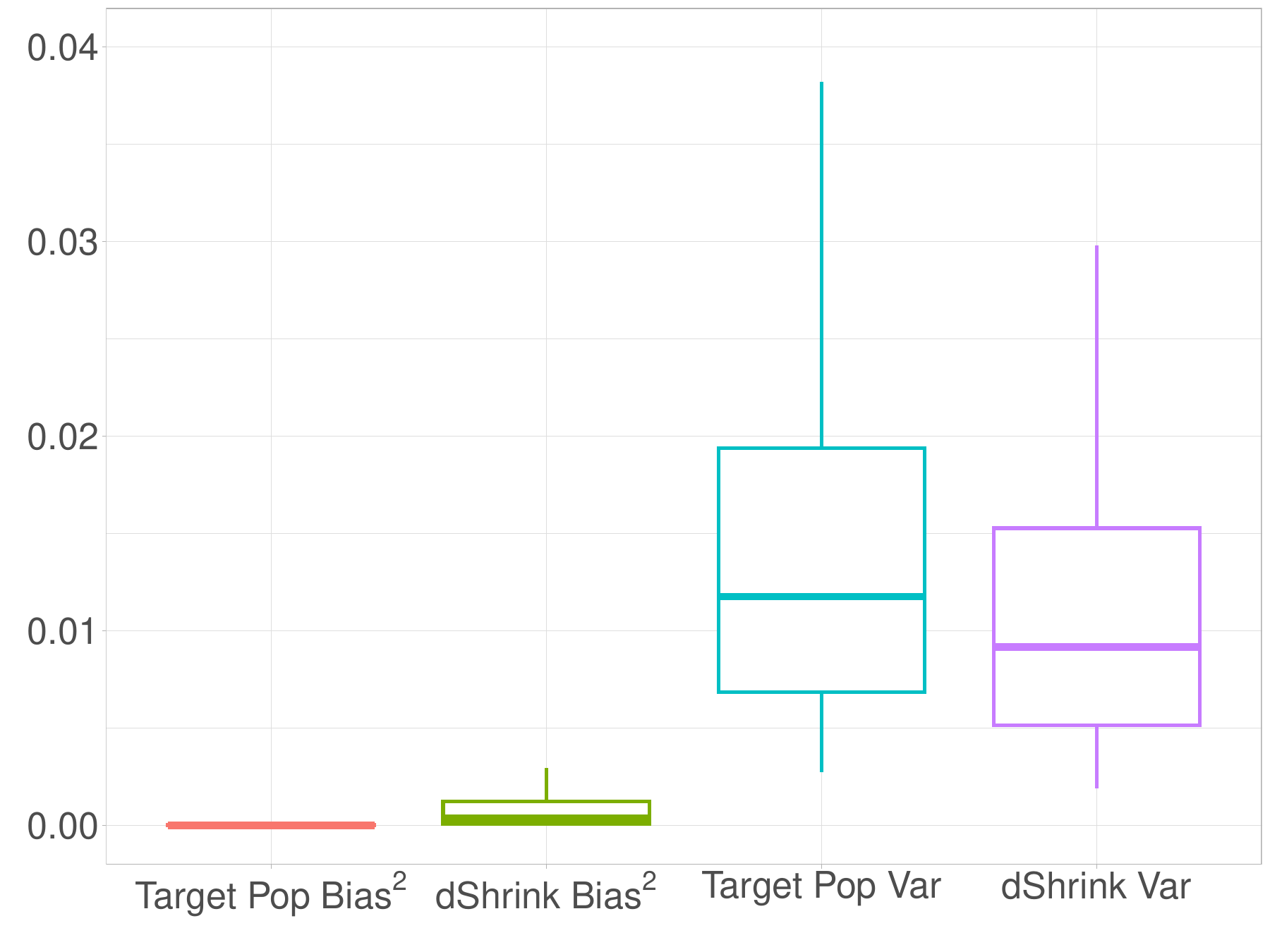}
			\caption{\it Squared sum of the bias and sum of the variance across the estimation for individualized treatment effects of different patients based on the target population-based estimator and the c-dShrink estimator.}\label{fig: ITE bias var}
		\end{figure}
	
	\subsection{Performances of m-dShrink and c-dShrink estimators in the simulation studies}\label{app: sim m&c}
	In this section, we evaluate the m-dShrink estimator and c-dShrink estimator under the simulation setting of Section \ref{subsec: sim} in the main text.
	Figure \ref{fig: sim ds-mds} compares the dShrink estimator and the m-dShrink estimator. $\bbA_{\beta}$ and $\bbA_{d}$ are taken as a vector of $1$'s in implementing the m-dShrink estimator.
	\begin{figure}[h]
		\centering
		\subfigure[]{
			\includegraphics[scale = 0.25]{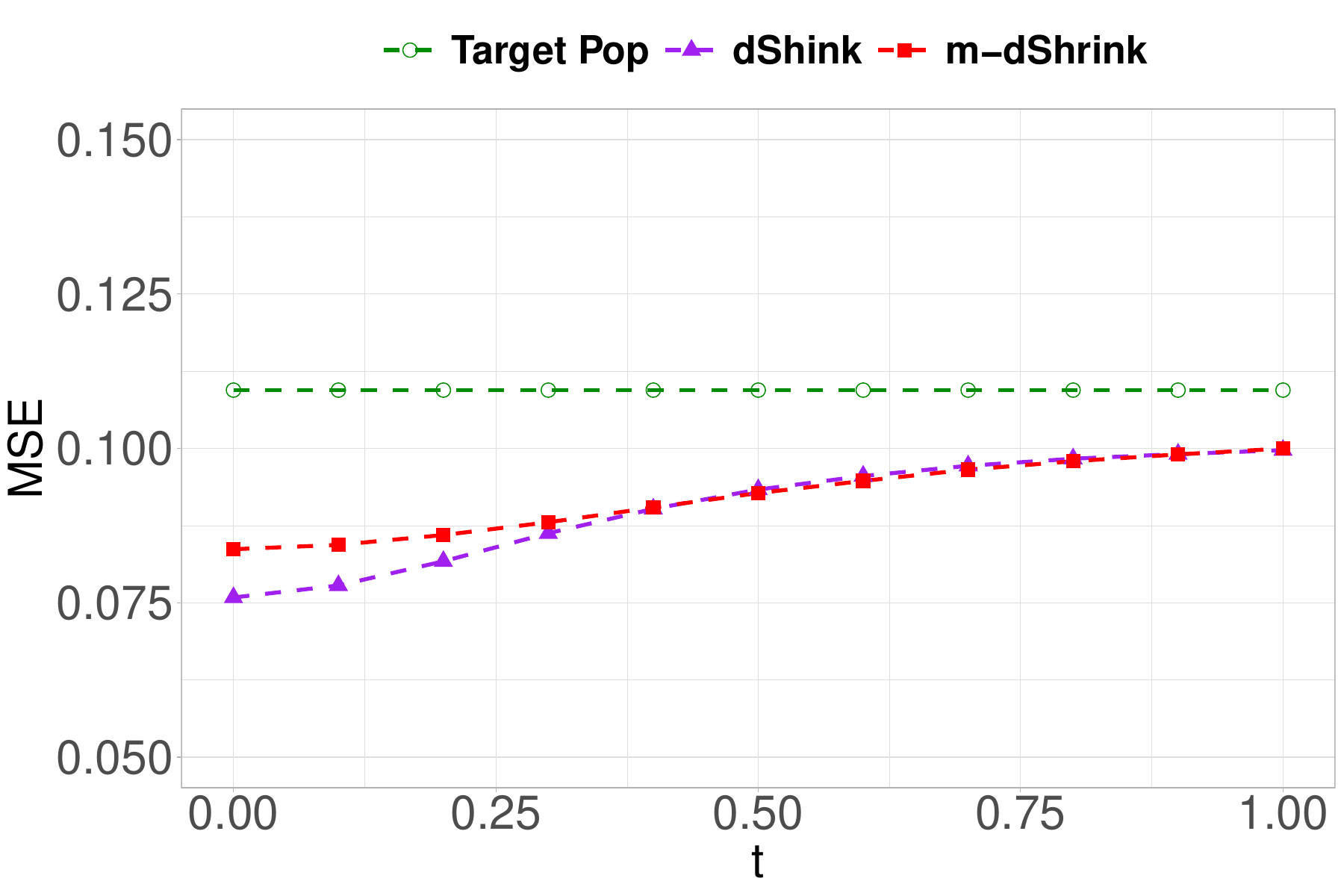}
		}
		\subfigure[]{
			\includegraphics[scale = 0.25]{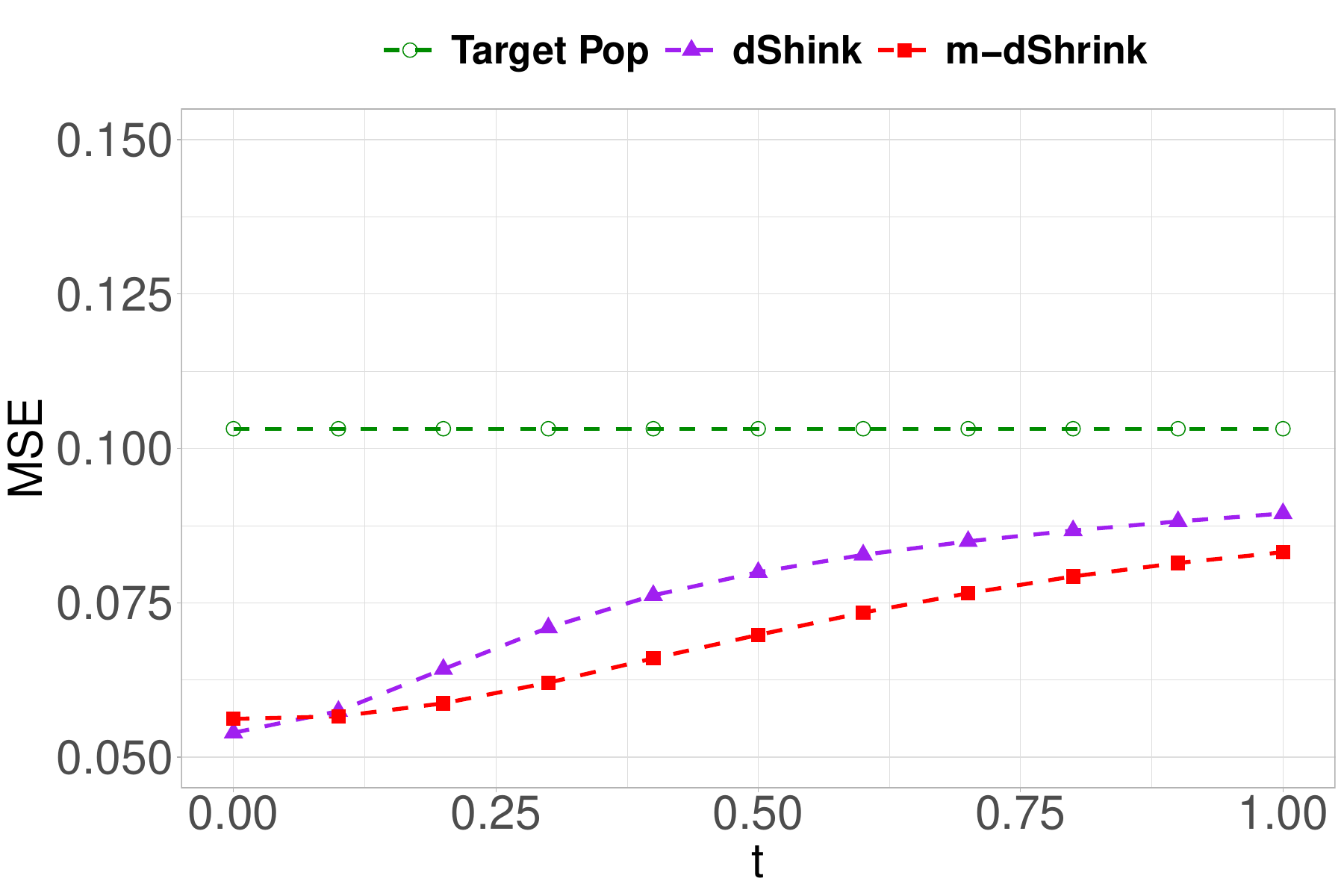}
		}
		\caption{\it MSEs of the regression functions derived from $\hbeta_{\cT}$, $\hbeta_{\rm ds}$, and $\hbeta_{\rm mds}$. (a) $p_{X} = 10$, $q_{X} = 6$, $n_{\cT} = 100$, and $n_{\cS} = 500$; (b)  $p_{X} = 30$, $q_{X} = 20$, $n_{\cT} = 300$, and $n_{\cS} = 1000$.}\label{fig: sim ds-mds}
	\end{figure}
	
	The m-dShrink estimator performs worse than or similar to the dShrink estimator in setting (a).  Note that the m-dShrink estimator involves linear regression models, which are fitted based on $\hbeta_{\cC}$ and $\halpha_{\cT} - \halpha_{\cS}$. $p$ and $q$ can be viewed as ``sample sizes" in the two regressions. Thus, the m-dShrink estimator does not perform well in setting (a) where $p$ and $q$ are small. In setting (b), $p$ and $q$ are larger and the m-dShrink estimator outperforms the dShrink estimator for most $t$'s. 
	
	Figure \ref{fig: sim ds-cds} compares the dShrink estimator and the c-dShrink estimator.
	\begin{figure}[h]
		\centering
		\subfigure[]{
			\includegraphics[scale = 0.25]{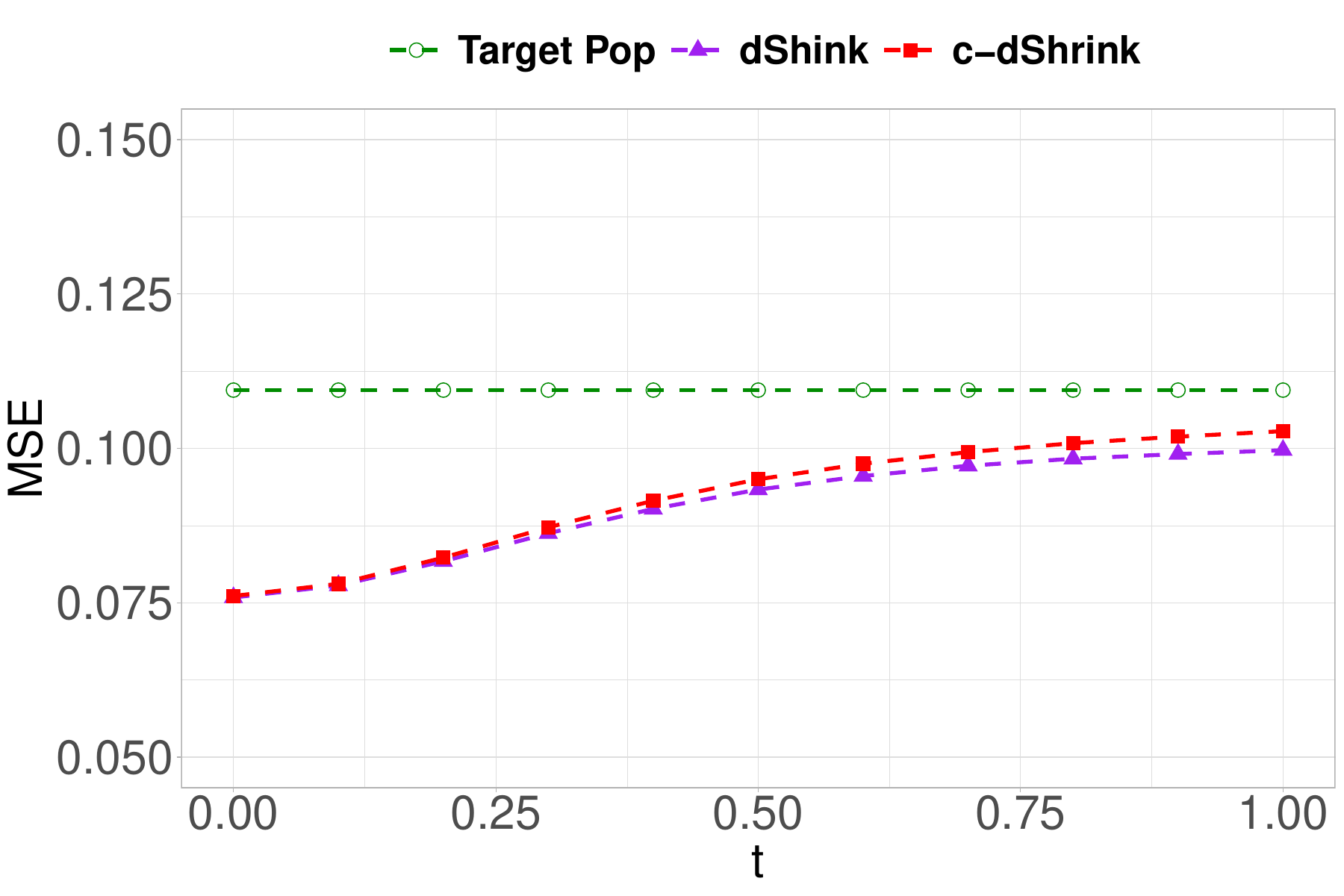}
		}
		\subfigure[]{
			\includegraphics[scale = 0.25]{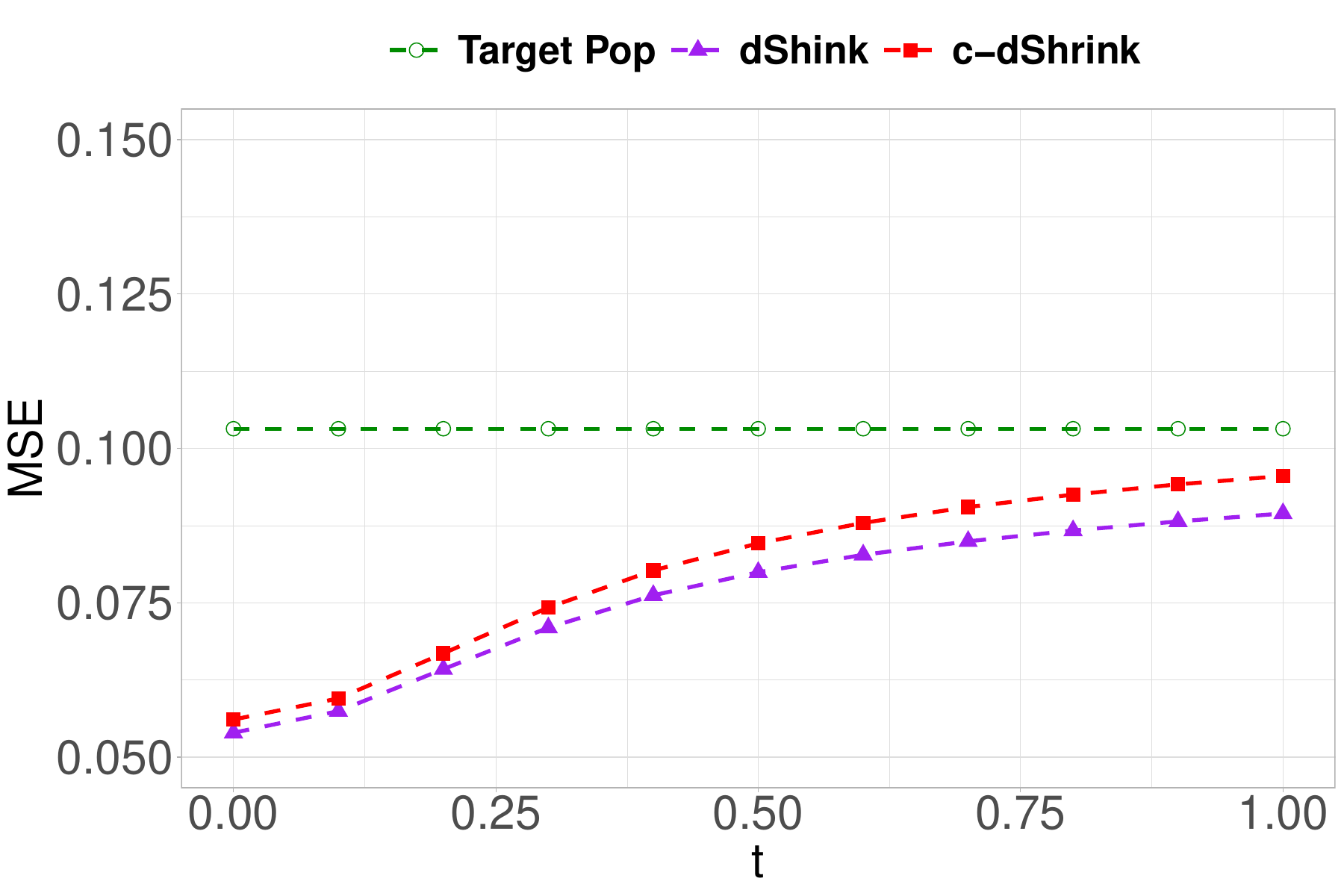}
		}
		\caption{\it MSEs of the regression functions derived from $\hbeta_{\cT}$, $\hbeta_{\rm ds}$, and $\hbeta_{\rm cds}^{\dag}$. (a) $p_{X} = 10$, $q_{X} = 6$, $n_{\cT} = 100$, and $n_{\cS} = 500$; (b)  $p_{X} = 30$, $q_{X} = 20$, $n_{\cT} = 300$, and $n_{\cS} = 1000$.}\label{fig: sim ds-cds}
	\end{figure}
	The c-dShrink estimator incurs larger errors than the dShrink estimator as it does not incorporate the covariance information $\bbSig_{\alpha,\cS}$ of $\halpha_{\cS}$. Despite this, the increase in error is minimal, and c-dShrink maintains a smaller error than $\hbeta_{\cT}$ across all evaluated scenarios.
	
		\subsection{Simulations with different error variances in target and source populations}\label{app: sim hVar}
		This section considers the same setting as that in Section \ref{subsec: sim} in the main text, while the error has different variances in target and source populations.
		Specifically, we set $\var(\epsilon) = 1$ and $2$ in target and source populations, respectively.   
		Figure \ref{fig: sim cmpr hVar} shows the errors of CML, GIM, PCML, and dShrink under different settings.  We varied the number of variables $p_{x}$ and $q_{x}$ and the sample sizes $n_{\cT}$ and $n_{\cS}$ in Figure \ref{fig: sim cmpr hVar}.
		
		\begin{figure}[h]
			\centering
			\subfigure[]{
				\includegraphics[scale = 0.25]{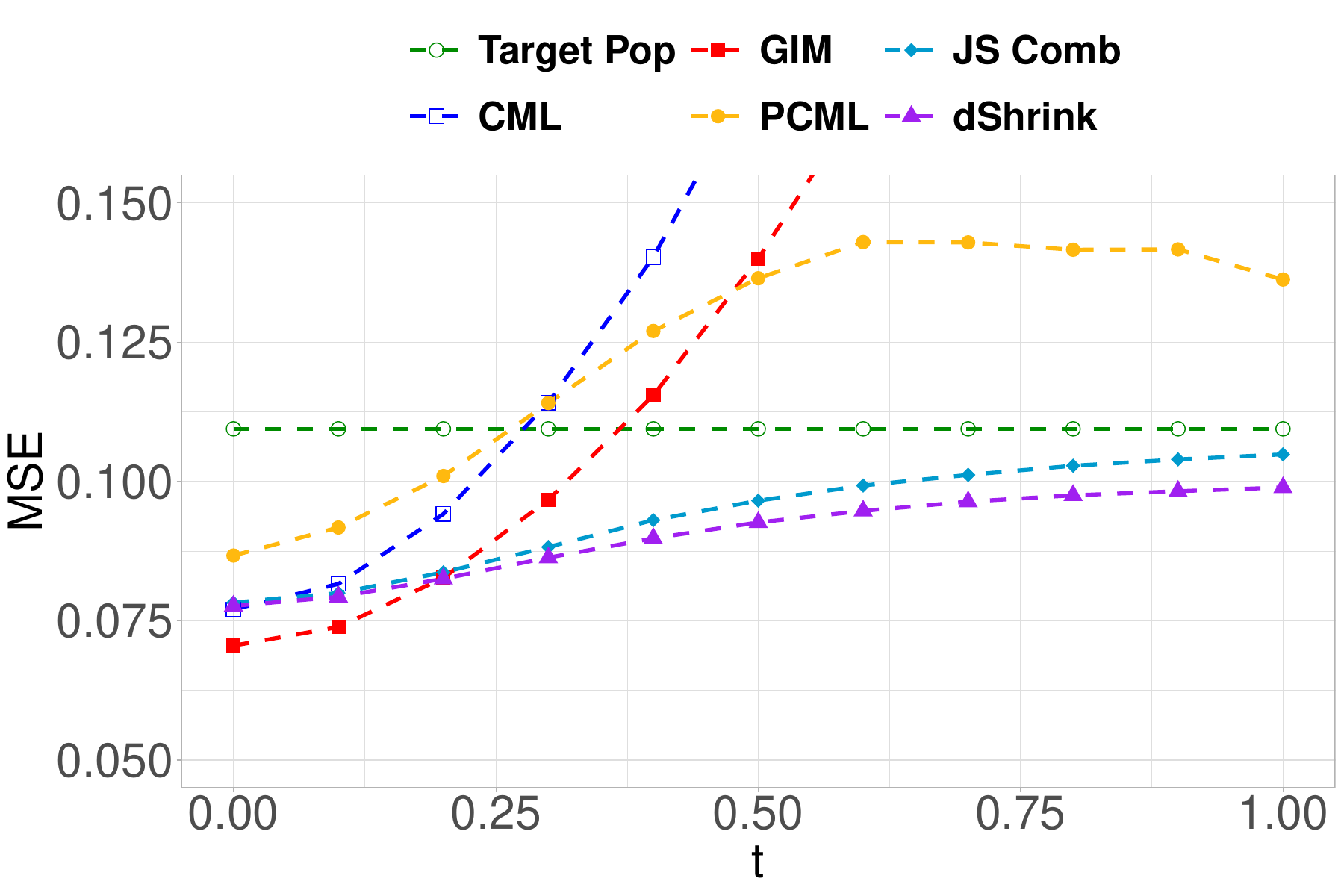}
			}
			\subfigure[]{
				\includegraphics[scale = 0.25]{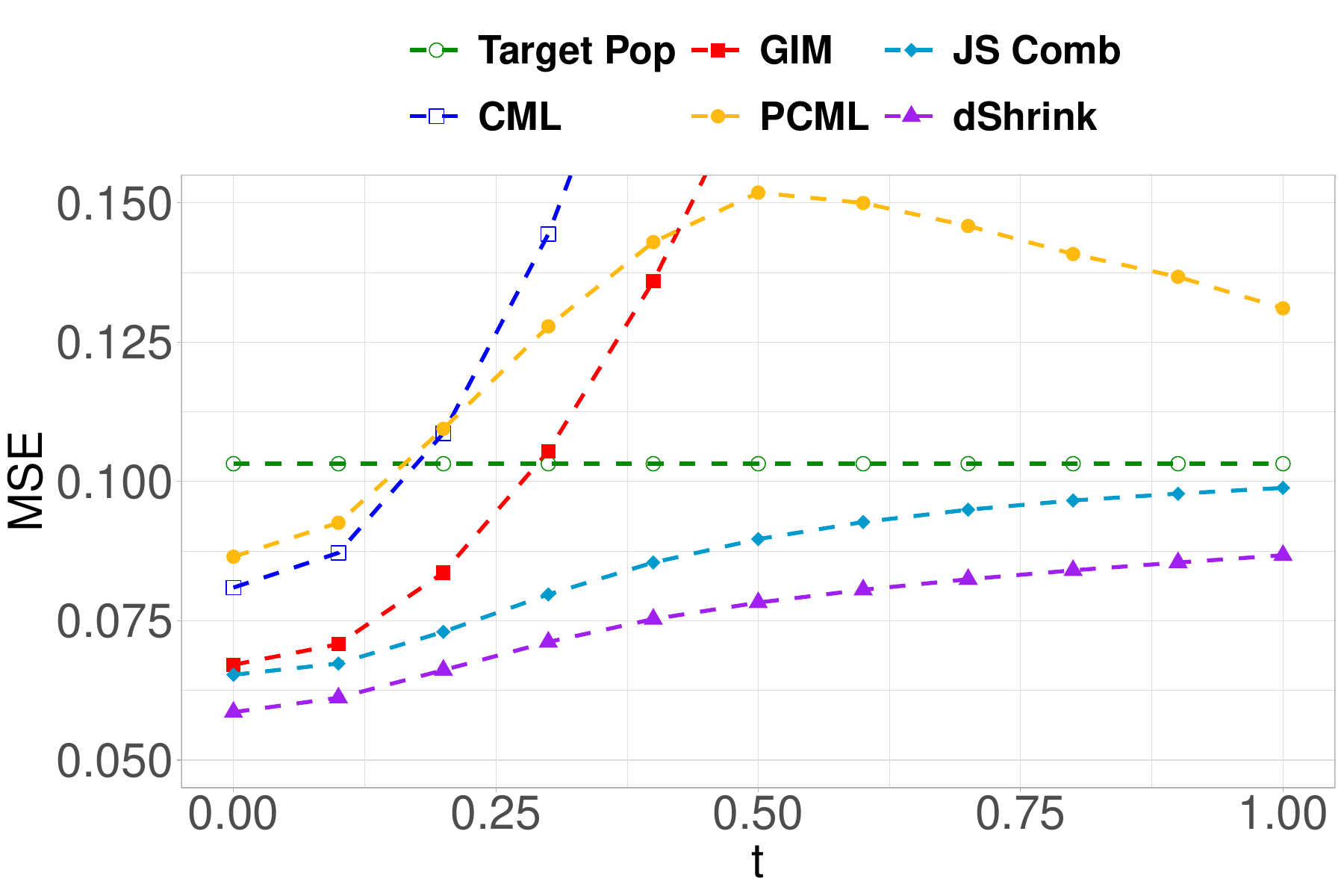}
			}
			\caption{\it MSEs of the regression functions derived from $\hbeta_{\cT}$, $\hbeta_{\rm ds}$, and the existing methods CML, GIM, PCML in simulation studies with different error variances in target and source populations. (a) $p_{X} = 10$, $q_{X} = 6$, $n_{\cT} = 100$, and $n_{\cS} = 500$; (b)  $p_{X} = 30$, $q_{X} = 20$, $n_{\cT} = 300$, and $n_{\cS} = 1000$.}\label{fig: sim cmpr hVar}
		\end{figure}
		
		Figure \ref{fig: sim CI hVar} shows the average coverage rates and average widths over CIs of different parameter components constructed using the $\hbeta_{\cT}$ and $\hbeta_{\rm ds}$.
		We set $\bbH$ to be an estimate of $\bbSig_{\beta, \cT}^{-1}$ in the implementation of dShrink. 
		
		\begin{figure}[h]
			\centering
			\subfigure[]{
				\includegraphics[scale = 0.28]{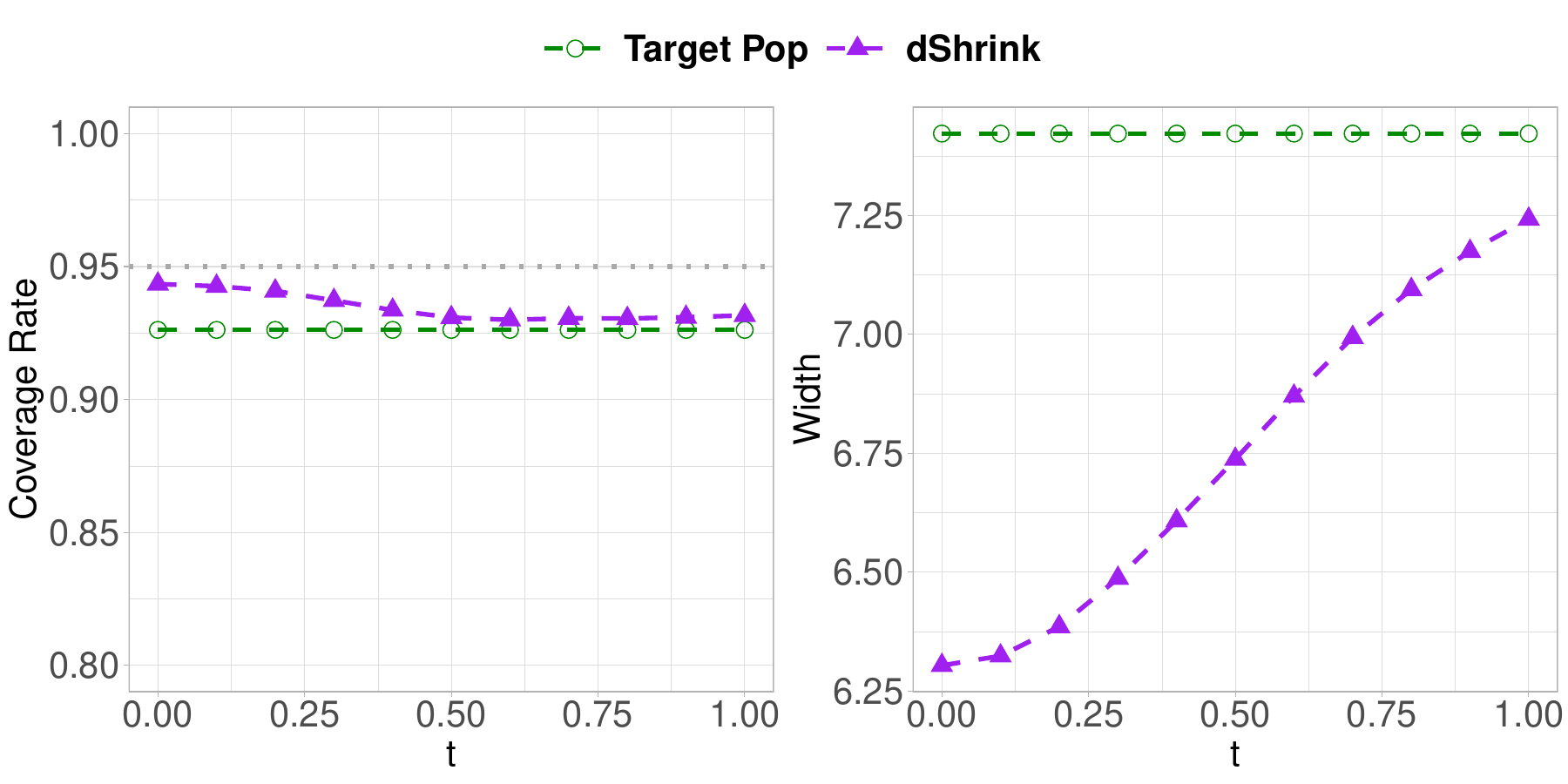}
			}
			\subfigure[]{
				\includegraphics[scale = 0.28]{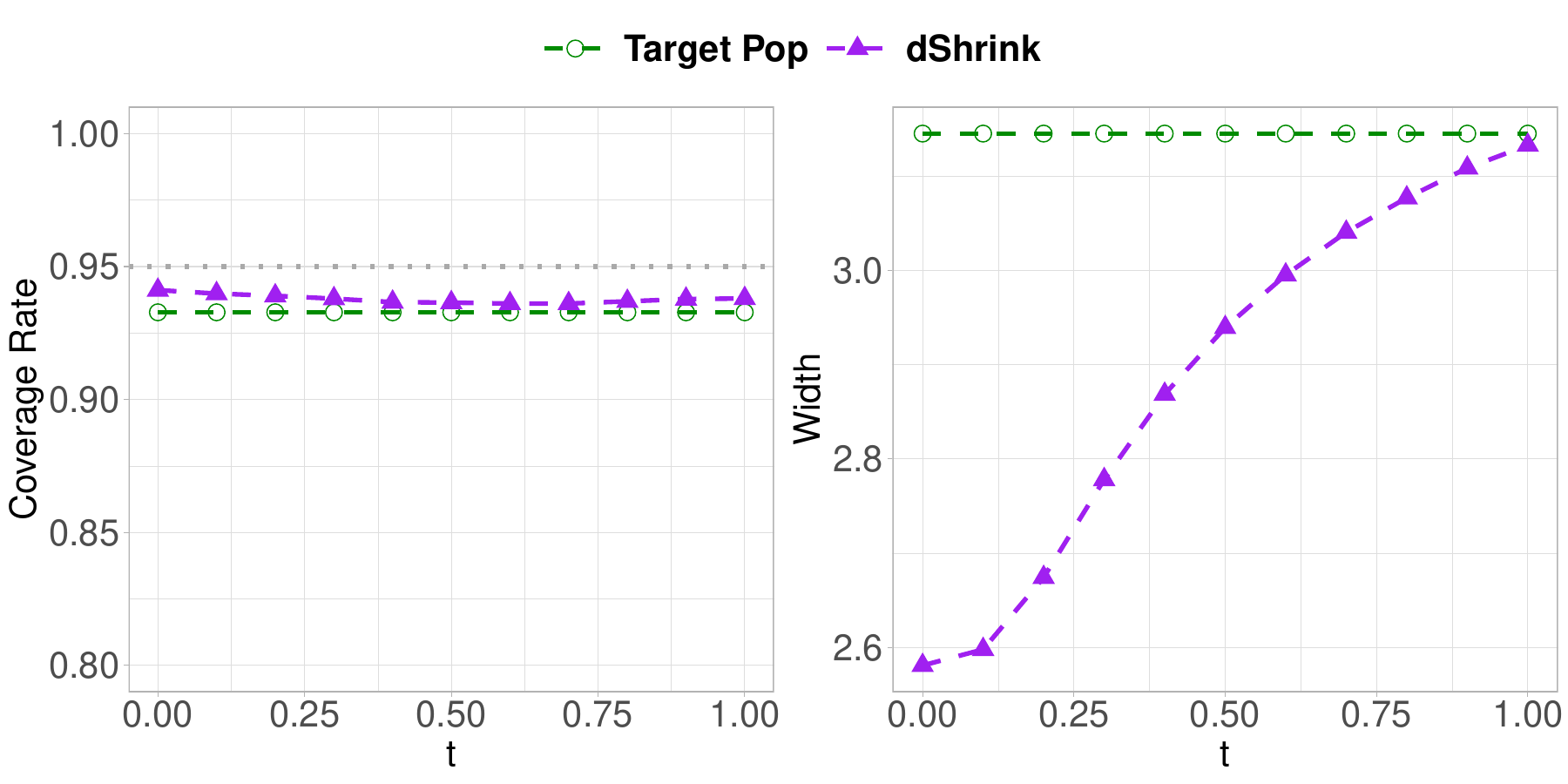}
			}
			\caption{\it Average coverage rates and average widths over different components of the standard CIs based on $\hbeta_{\cT}$ and $\hbeta_{\rm ds}$. The average widths are multiplied by ten. (a) $p_{X} = 10$, $q_{X} = 6$, $n_{\cT} = 100$, and $n_{\cS} = 500$; (b)  $p_{X} = 30$, $q_{X} = 20$, $n_{\cT} = 300$, and $n_{\cS} = 1000$.}\label{fig: sim CI hVar}
		\end{figure}
		The results in this section are similar to those in Section \ref{subsec: sim}, which demonstrates the robustness of dShrink against the heterogeneity of error variance in target and source populations.
		
		\subsection{Simulations with correlated observations}\label{app: sim correlated obs}
		To demonstrate the generality of the proposed method, we consider a nonlinear regression problem with non-normal correlated errors. In the target population, let $\bX$ be the $p_{X}$-dimensional covariate vector of an individual with i.i.d. standard normal components. For each individual, three repeated measurements for the outcome were observed. Let 
		$\bY = (Y_{1}, Y_{2}, Y_{3})$
		be the $3$-dimensional outcome vector of an individual with 
		\[
		Y_{j} = p(\bX^{\T}\beta_{\cT}) + \frac{\epsilon_{c} + \epsilon_{j}}{2},
		\]
		where $p(\cdot)$ is the logistic function, $\epsilon_{c}$ is a common error term across $j = 1, 2, 3$ that follows a t-distribution with $10$ degrees of freedom, $\epsilon_{j}$ is an error term specific to the $j$-th observation that follows a Laplace distribution with mean zero and unit variance, $(\epsilon_{c}, \epsilon_{1}, \epsilon_{2}, \epsilon_{3})\Perp \bX$, and $\{\epsilon_{j}\}_{j = 1}^{3}$ are independent of each other.
		
		The source population data are generated in the same way as the target population except that the coefficient $\bbeta_{\cT}$ is replaced by $\bbeta_{\cS}$. To specify $\bbeta_{\cT}$ and $\bbeta_{\cS}$, we generate $p_{X}$-dimensional vectors $\boldsymbol{\mu}$ and $\bbias$ with i.i.d. components from $N(0, 0.25 / p)$, and set $\bbeta_{\cT} = (\boldsymbol{\mu} - t\bbias) / \|\boldsymbol{\mu} - t\bbias\|$ and $\bbeta_{\cS} = (\boldsymbol{\mu} + t\bbias) / \|\boldsymbol{\mu} + t\bbias\|$ for $0\leq t \leq 1$. In this example, the norm of the coefficient affects the variance of estimators. Thus, we normalized $\boldsymbol{\mu} - t\bbias$ and $\boldsymbol{\mu} + t\bbias$ by their norm when setting $\bbeta_{\cT}$ and $\bbeta_{\cS}$ to make the coefficient norm unchanged across different $t$'s.
		The target population-based estimator $\hbeta_{\cT}$ and the external summary statistics $\halpha_{\cS}$ are constructed based on the generalized estimating equation estimation for $\bbeta_{\cT}$ and $\bbeta_{\cS}$, respectively. In this setting, $p = q = p_{X}$. Figure \ref{fig: sim error correlated} presents the MSEs of different estimators.  We took $\bbH = \bbI_{p}$ in the implementation of dShrink.
		\begin{figure}[h]
			\centering
			\subfigure[]{
				\includegraphics[scale = 0.25]{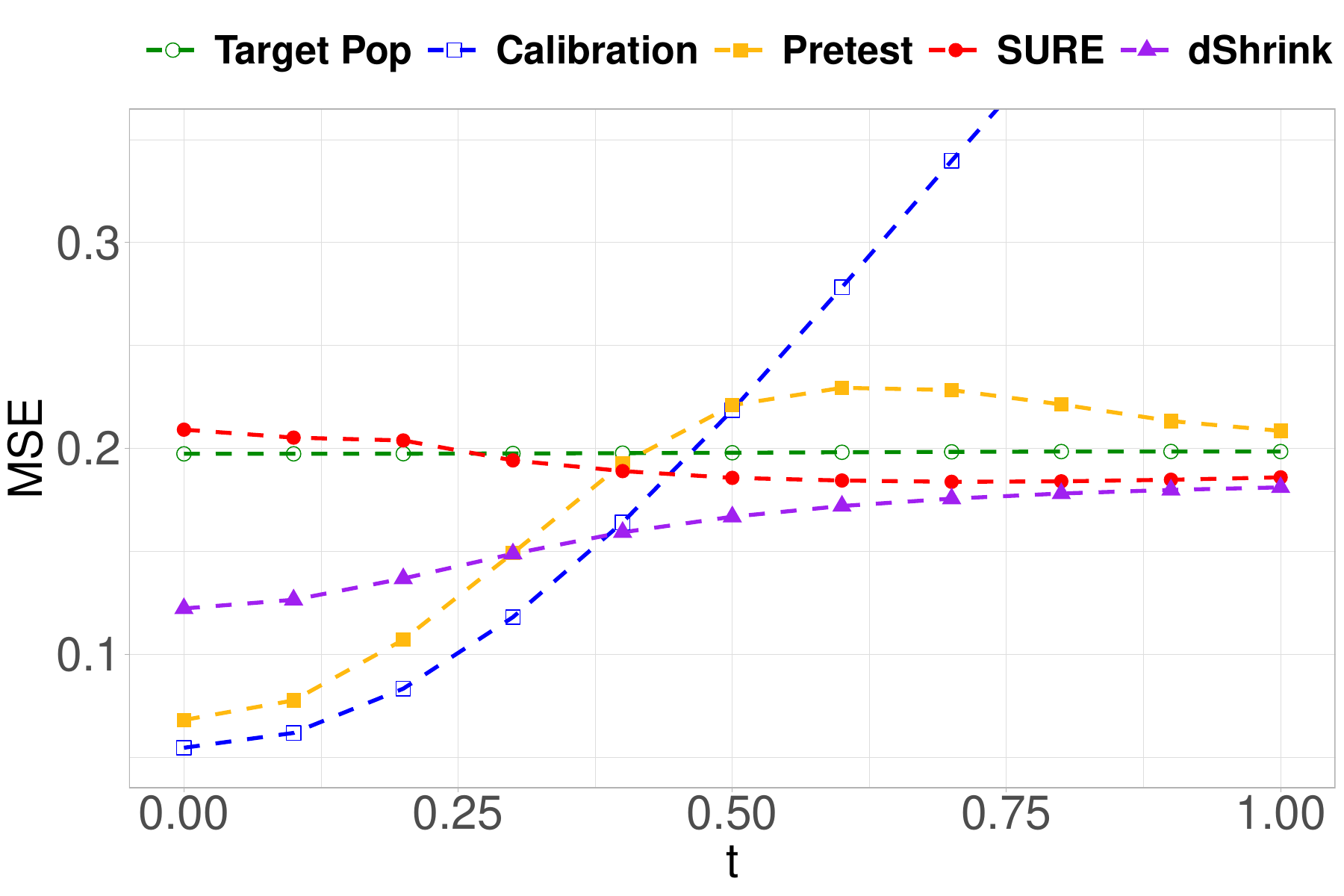}
			}
			\subfigure[]{
				\includegraphics[scale = 0.25]{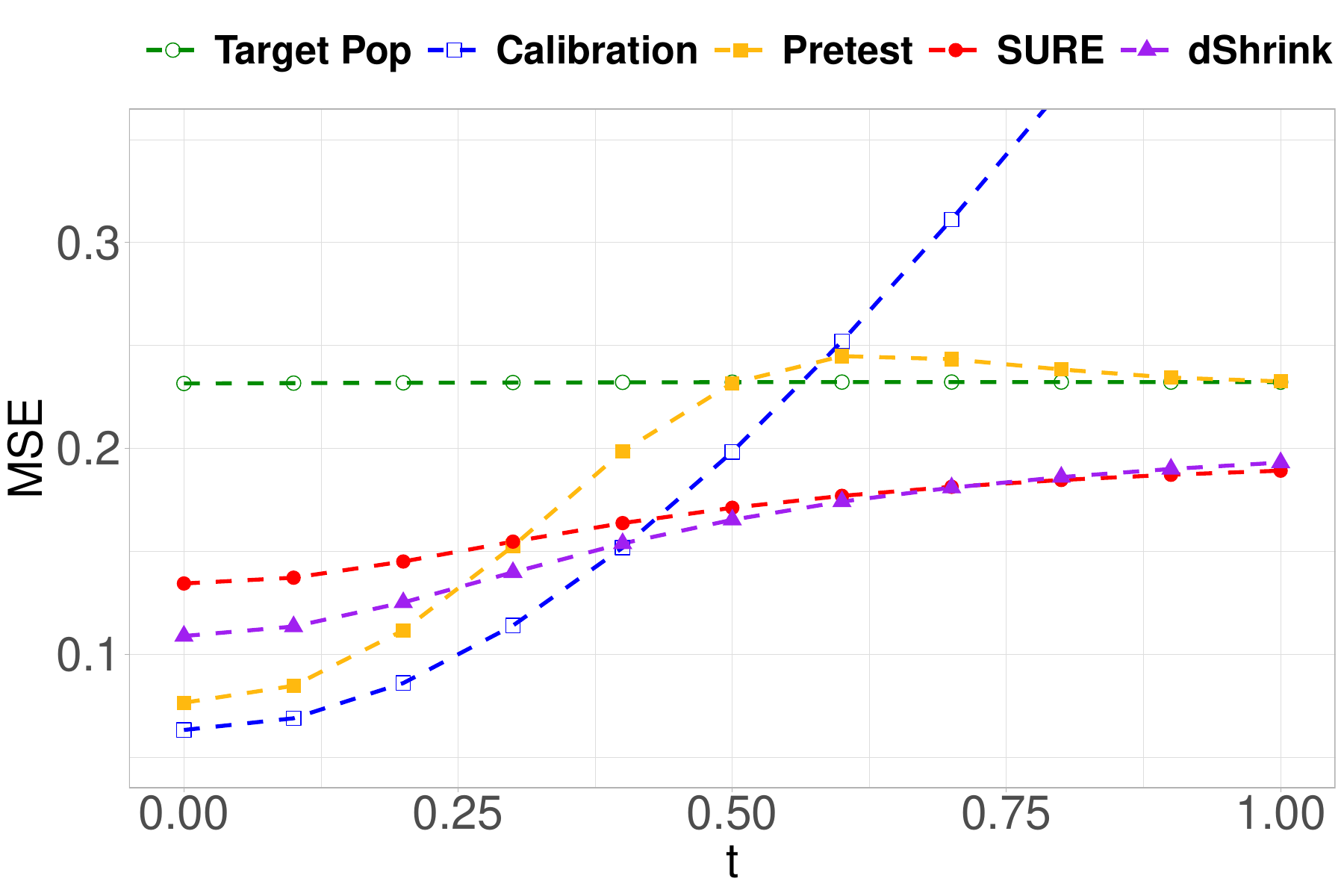}
			}
			\caption{\it MSEs of $\hbeta_{\cT}$, $\hbeta_{\cC}$, $\hbeta_{\rm prt}$, $\hbeta_{\SURE}$, and $\hbeta_{\rm ds}$ in the simulation with correlated errors. (a) $p = 6$, $n_{\cT} = 300$, and $n_{\cS} = 600$; (b)  $p_{X} = 12$, $n_{\cT} = 600$, and $n_{\cS} = 1200$.}\label{fig: sim error correlated}
		\end{figure}
		
		Figure \ref{fig: sim CI correlated} shows the average coverage rates and average widths over CIs of different parameter components constructed using the $\hbeta_{\cT}$ and $\hbeta_{\rm ds}$.
		We set $\bbH$ to be an estimate of $\bbSig_{\beta, \cT}^{-1}$ in the implementation of dShrink. 
		
		\begin{figure}[h]
			\centering
			\subfigure[]{
				\includegraphics[scale = 0.28]{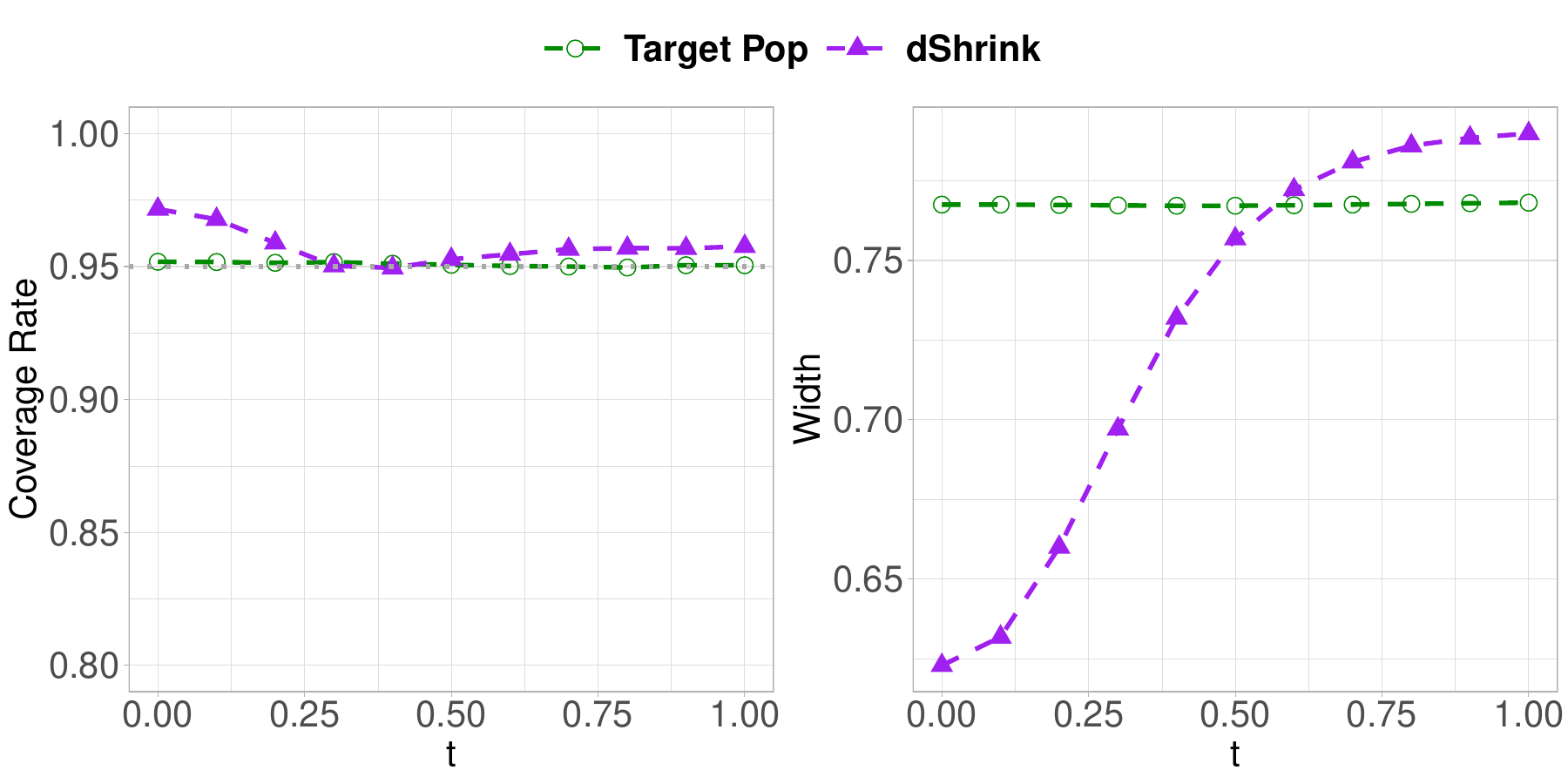}
			}
			\subfigure[]{
				\includegraphics[scale = 0.28]{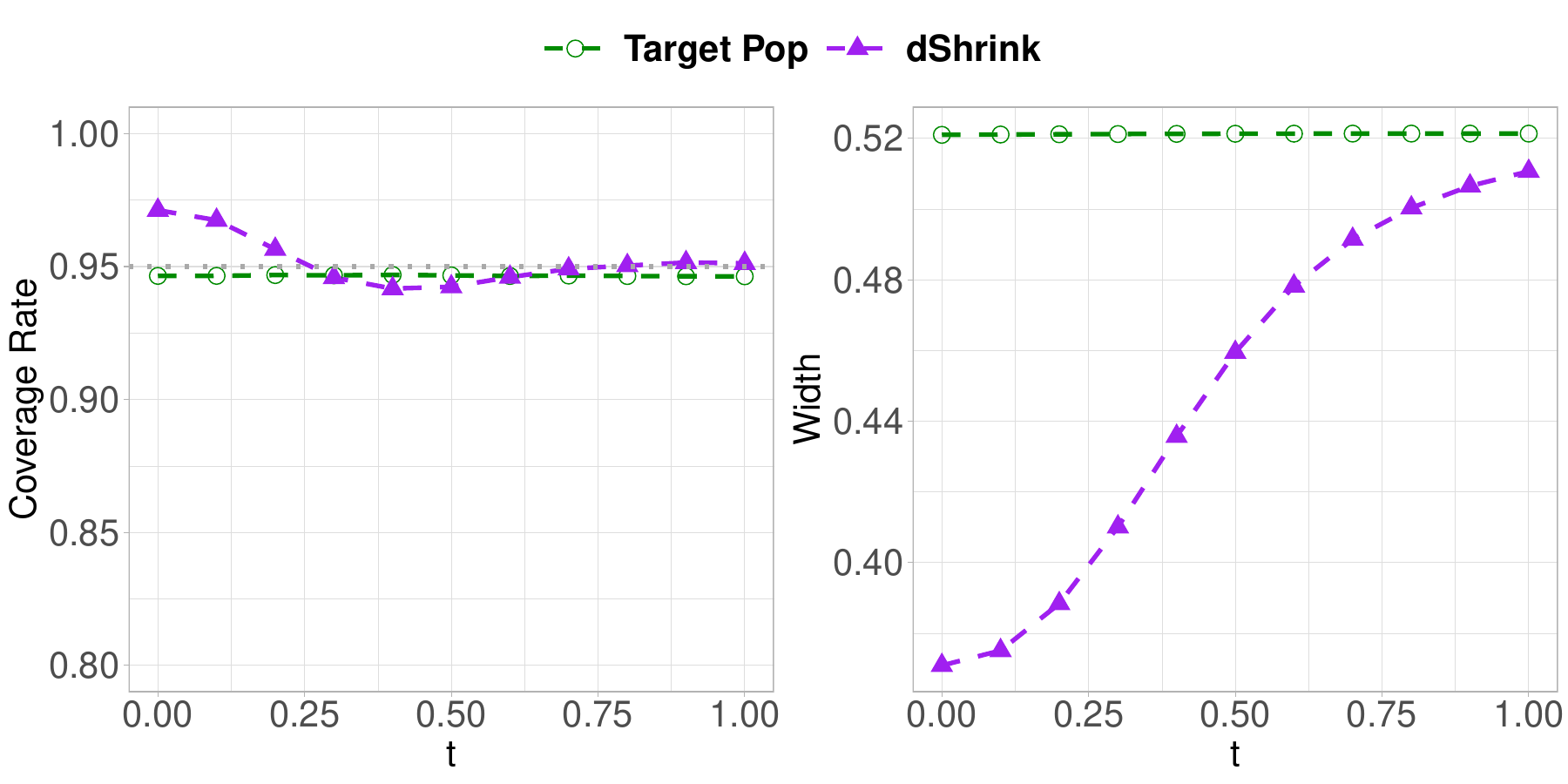}
			}
			\caption{\it Average coverage rates and average widths over CIs of different parameter components based on $\hbeta_{\cT}$ and $\hbeta_{\rm ds}$. (a) $p_{X} = 6$, $n_{\cT} = 300$, and $n_{\cS} = 600$; (b)  $p_{X} = 12$, $n_{\cT} = 600$, and $n_{\cS} = 1200$.}\label{fig: sim CI correlated}
		\end{figure}
		
		Figures \ref{fig: sim error correlated} and \ref{fig: sim CI correlated} show that dShrink can achieve improved MSE compared to $\hbeta_{\cT}$ and guaranteed coverage rate in the setting with non-normal correlated data. Moreover, it often yields narrower confidence intervals than those based on $\hbeta_{\cT}$. These results demonstrate the broad applicability of the proposed dShrink estimator.

		\subsection{Simulations when the same variables are observed in target and source populations and individual-level data are available}\label{app: sim same variable}
		In this Section~\ref{app: sim same variable}, we compare the proposed method with data-enriched linear regression \citep{chen2015data}, TransLasso \citep{li2022transfer}, TransGLM \citep{tian2022transfer}, and ISEDI \citep{hector2024turning} in a setting where the same variables are observed in both populations and individual-level data are available.
		
		Data were generated from the same model considered in Section \ref{subsec: sim} in the main text, but the same variables were observed in target and source populations, and individual-level data were available. Note that TransLasso and TransGLM are originally designed for high dimensional problems. We set the dimension of the covariate to be higher compared to that in Section \ref{subsec: sim}.
		The components of $\bbeta_{\cT}$ were generated independently from ${\rm Bernoulli}(0.8) \times N(0, 2 / p)$. Let $\etab$ be a $p$-dimensional vector whose components were generated independently from ${\rm Bernoulli}(0.5) \times N(0, 2 / p)$ and $\bbeta_{\cS} = \bbeta_{\cT} + t\etab$, where $0\leq t\leq 1$ is a parameter that characterizes the difference between $P_{\cT}$ and $P_{\cS}$. The matrices $\bbH$, $\bbSig_{\cT}$, and $\bbSig_{\cS}$ were set in the same way as in Section \ref{subsec: sim}. The target population-based estimator $\hbeta_{\cT}$ is the least squares estimator. The expected quadratic error is approximated by the mean of the quadratic error across $5000$ simulation runs. Figure \ref{fig: sim error} shows the errors of different estimators under different settings.  We varied the number of variables $p_{X}$ and the sample sizes $n_{\cT}$ and $n_{\cS}$ in Figure \ref{fig: sim cmpr samevariable}.
		
		\begin{figure}[h]
			\centering
			\subfigure[]{
				\includegraphics[scale = 0.25]{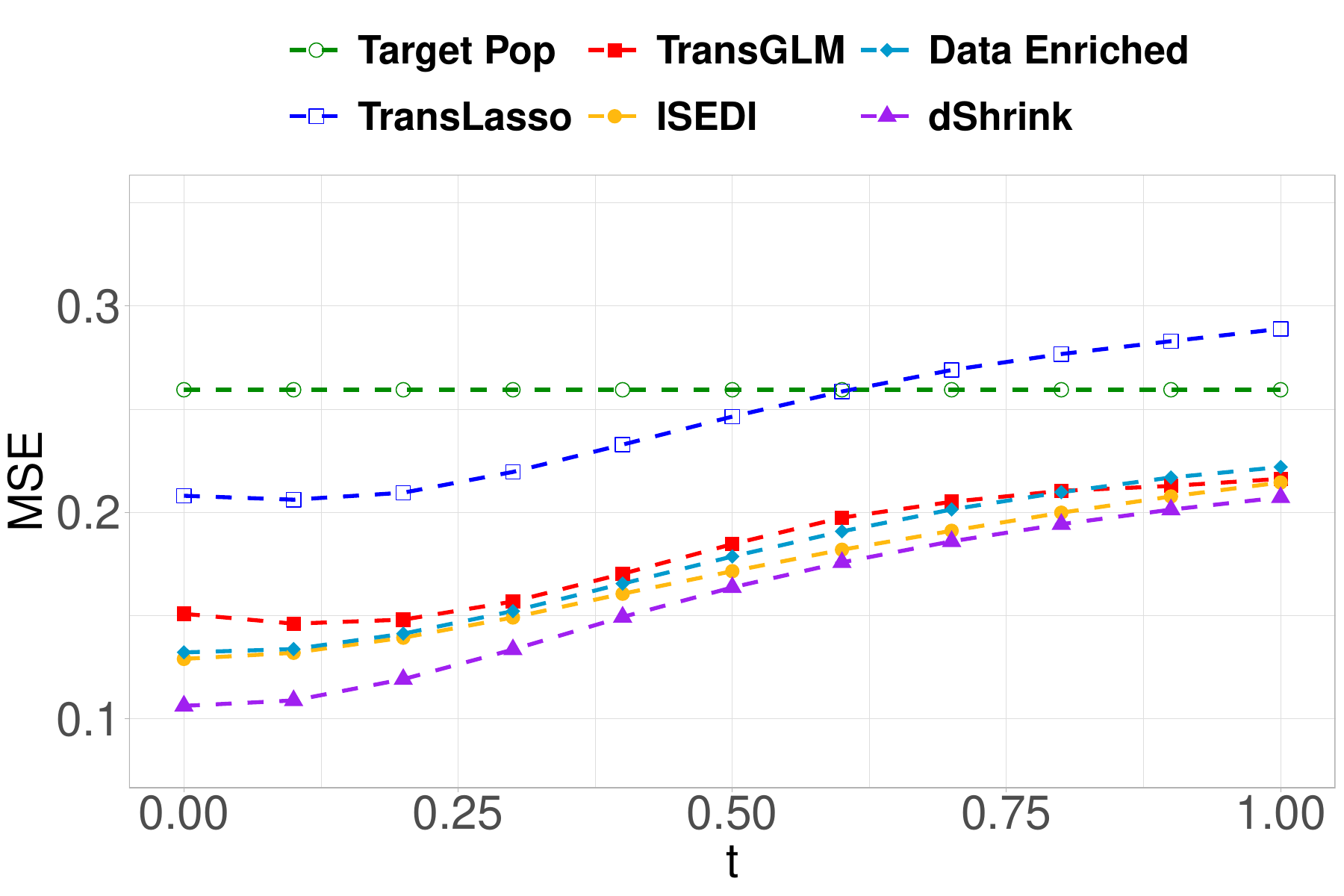}
			}
			\subfigure[]{
				\includegraphics[scale = 0.25]{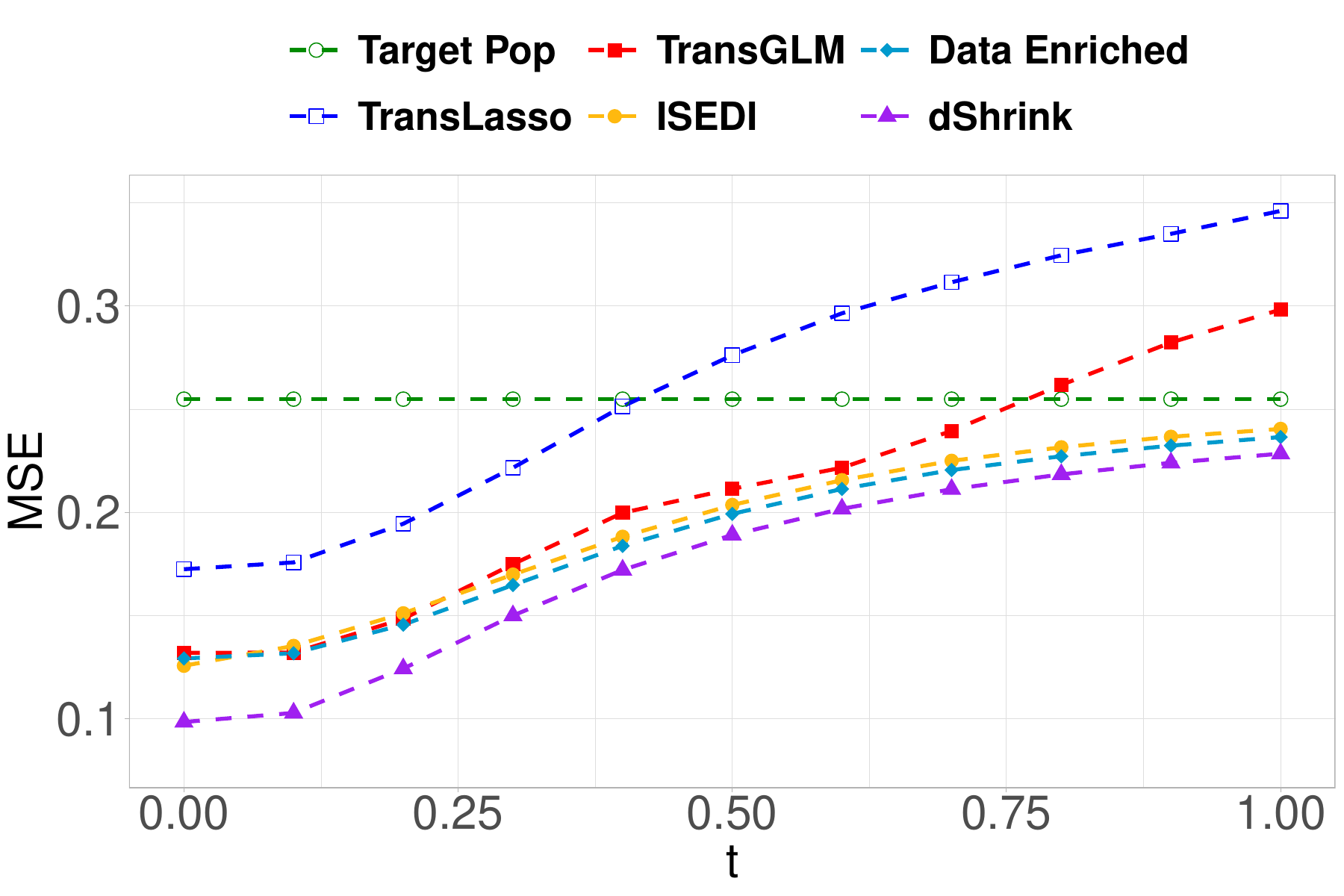}
			}
			\caption{\it MSEs of the regression functions derived from $\hbeta_{\cT}$, $\hbeta_{\rm ds}$, and the existing methods CML, GIM, PCML in simulation studies. (a) $p_{X} = 25$, $n_{\cT} = 100$, and $n_{\cS} = 200$; (b)  $p_{X} = 50$, $n_{\cT} = 200$, and $n_{\cS} = 400$.}\label{fig: sim cmpr samevariable}
		\end{figure}
		
		Figure \ref{fig: sim CI sameVariable} presents the average coverage rates and average widths over different parameter components of CIs based on different methods. Data-enriched linear regression and TransLasso are not included in the comparison because they do not produce inference results.  We set $\bbH$ to be an estimate of $\bbSig_{\beta, \cT}^{-1}$ in the implementation of dShrink. 
		\begin{figure}[h]
			\centering
			\subfigure[]{
				\includegraphics[scale = 0.28]{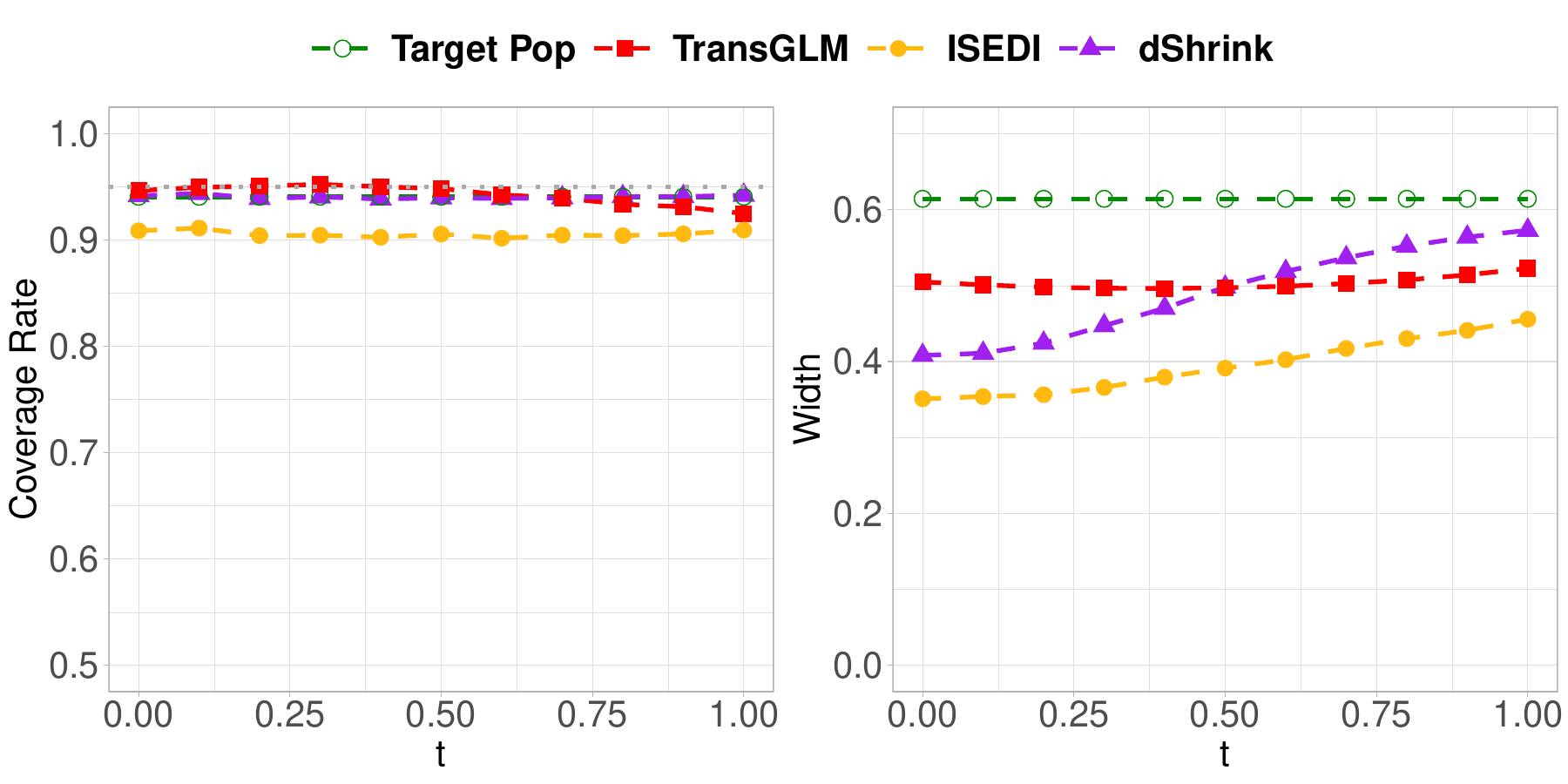}
			}
			\subfigure[]{
				\includegraphics[scale = 0.28]{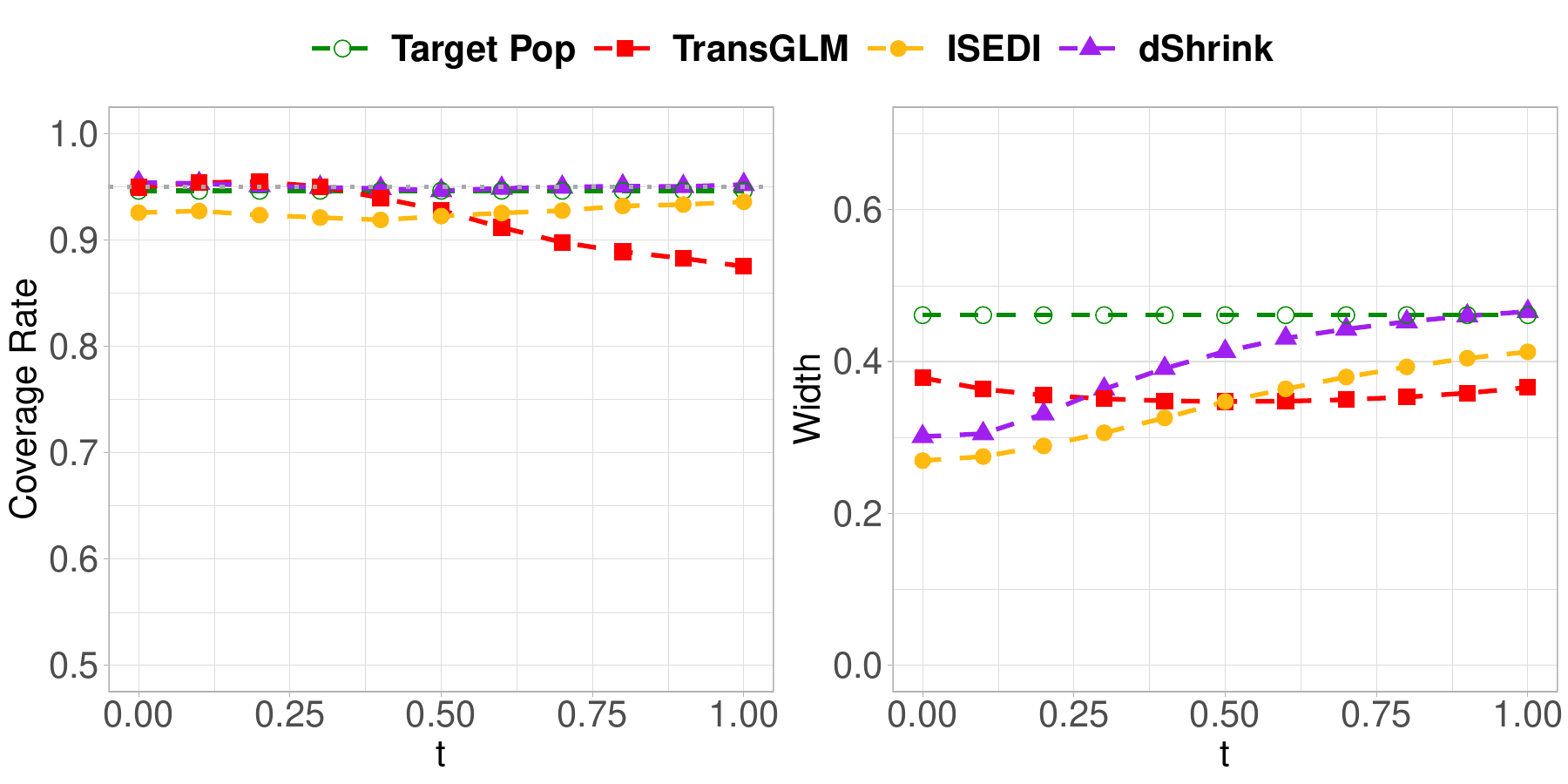}
			}
			\caption{\it Average coverage rates and average widths over CIs of different parameter components based on $\hbeta_{\cT}$, TransGLM, ISEDI, and $\hbeta_{\rm ds}$. (a) $p_{X} = 25$, $n_{\cT} = 100$, and $n_{\cS} = 200$; (b)  $p_{X} = 50$, $n_{\cT} = 200$, and $n_{\cS} = 400$.}\label{fig: sim CI sameVariable}
		\end{figure}
		
		Figures \ref{fig: sim cmpr samevariable} and \ref{fig: sim CI sameVariable} demonstrate that dShrink outperforms the existing methods in terms of MSE and coverage rate, and can reduce the width of the CIs compared to the target population-based estimator in most cases.

\end{document}